\documentclass[reqno,a4paper,12pt]{amsart}
\usepackage{fixltx2e}
\usepackage[utf8x]{inputenc}
\usepackage{tgpagella}
\usepackage[small]{eulervm}
\usepackage{amsmath,amssymb,amstext,amsthm,amscd,mathrsfs,eucal}
\usepackage{a4wide}
\usepackage{graphicx,color}
\usepackage{array,colortbl,xtab}
\usepackage[all]{xy}
\usepackage[noadjust]{cite}
\usepackage{hyperref}
\hypersetup{%
  pdftitle   = {Homogeneous M2 duals},
  pdfkeywords = {AdS/CFT, M2, M-theory, Killing spinors, homogeneity},
  pdfauthor  = {José Figueroa-O'Farrill, Mara Ungureanu},
  pdfcreator = {\LaTeX\ with package hyperref}
}
\let\ten\natural
\newcommand{\half}{\tfrac12}
\newcommand{\fg}{\mathfrak{g}}
\newcommand{\fk}{\mathfrak{k}}
\newcommand{\fq}{\mathfrak{q}}

\newcommand{\fgl}{\mathfrak{gl}}
\newcommand{\fh}{\mathfrak{h}}

\newcommand{\fm}{\mathfrak{m}}
\newcommand{\fn}{\mathfrak{n}}
\newcommand{\fp}{\mathfrak{p}}

\newcommand{\fso}{\mathfrak{so}}
\newcommand{\fosp}{\mathfrak{osp}}

\newcommand{\fsp}{\mathfrak{sp}}
\newcommand{\fsu}{\mathfrak{su}}
\newcommand{\fu}{\mathfrak{u}}

\newcommand{\SO}{\mathrm{SO}}
\newcommand{\Spin}{\mathrm{Spin}}
\newcommand{\Sp}{\mathrm{Sp}}
\newcommand{\SU}{\mathrm{SU}}

\newcommand{\GL}{\mathrm{GL}}

\newcommand{\RR}{\mathbb{R}}

\newcommand{\CC}{\mathbb{C}}
\newcommand{\ZZ}{\mathbb{Z}}
\newcommand{\eD}{\mathscr{D}}

\newcommand{\be}{\mathbold{e}}

\newcommand{\bx}{\mathbold{x}}
\newcommand{\by}{\mathbold{y}}

\newcommand{\nablat}{\widetilde \nabla}
\newcommand{\etati}{\widetilde \eta}
\newcommand{\Pti}{\widetilde P}
\newcommand{\Sti}{\widetilde S}
\newcommand{\gti}{\widetilde g}

\renewcommand{\Re}{\mathrm{Re}}
\DeclareMathOperator{\AdS}{AdS}
\DeclareMathOperator{\dS}{dS}
\DeclareMathOperator{\Aut}{Aut}

\DeclareMathOperator{\dvol}{dvol}

\DeclareMathOperator{\ad}{ad}

\DeclareMathOperator{\Ric}{Ric}
\DeclareMathOperator{\id}{id}

\definecolor{gray}{rgb}{0.5,0.5,0.5}
\theoremstyle{plain}
\newtheorem*{lem}{Lemma}

\theoremstyle{definition}

\newcommand{\MUNCH}[1]{\relax}

\makeatletter
\newcommand{\interitemtext}[1]{%
\begin{list}{}
{\itemindent=0mm\labelsep=0mm
\labelwidth=0mm\leftmargin=0mm
\addtolength{\leftmargin}{-\@totalleftmargin}}
\item #1
\end{list}}
\makeatother
\newcommand{\subalg}[2]{\genfrac{}{}{0pt}{0}{#1}{\left<#2\right>}}
\newcommand{\Hdim}[1]{{\color{gray}\scriptstyle #1 \ar@{.}[r]}}
\allowdisplaybreaks
\begin{document}
\title{Homogeneous M2 duals}
\author[Figueroa-O'Farrill]{José Figueroa-O'Farrill}
\address[JMF]{School of Mathematics and Maxwell Institute for Mathematical
  Sciences, The University of Edinburgh, James Clerk Maxwell Building,
  The King's Buildings, Peter Guthrie Tait Road, Edinburgh EH9 3FD,
  Scotland, UK}
\author[Ungureanu]{Mara Ungureanu}
\address[MU]{Humboldt-Universität zu Berlin, Institut für Mathematik, Unter den Linden 6, 10099 Berlin, Germany}
\email{J.M.Figueroa@ed.ac.uk, ungurean@math.hu-berlin.de}
\thanks{EMPG-15-19}
\begin{abstract}
  Motivated by the search for new gravity duals to M2 branes with
  $N>4$ supersymmetry --- equivalently, M-theory backgrounds with
  Killing superalgebra $\fosp(N|4)$ for $N>4$ --- we classify
  homogeneous M-theory backgrounds with symmetry Lie algebra
  $\fso(n) \oplus \fso(3,2)$ for $n=5,6,7$.  We find that there are no
  new backgrounds with $n=6,7$ but we do find a number of new (to us)
  backgrounds with $n=5$.  All backgrounds are metrically products of
  the form $\AdS_4 \times P^7$, with $P$ riemannian and homogeneous
  under the action of $\SO(5)$, or $S^4 \times Q^7$ with $Q$
  lorentzian and homogeneous under the action of $\SO(3,2)$.  At least
  one of the new backgrounds is supersymmetric (albeit with only
  $N=2$) and we show that it can be constructed from a supersymmetric
  Freund--Rubin background via a Wick rotation.  Two of the new
  backgrounds have only been approximated numerically.  (The second
  version of this paper includes an appendix by Alexander~S.~Haupt,
  closing a gap in our original analysis.)
\end{abstract}
\maketitle
\tableofcontents

\section{Introduction}
\label{sec:introduction}

The gauge/gravity correspondence for M2 branes \cite{Malda} suggests
that with every three-dimensional superconformal field theory, there
should be associated a supersymmetric background of eleven-dimensional
supergravity, whose Killing (or more generally, symmetry) superalgebra
is isomorphic to the superconformal algebra of the field theory.  It
follows from Nahm's classification \cite{Nahm} that the
three-dimensional conformal superalgebra is isomorphic to $\fosp(N|4)$
for some $N\leq 8$.  The even subalgebra of $\fosp(N|4)$ is
$\fso(N) \oplus \fsp(4;\RR)$, where $\fsp(4;\RR) \cong \fso(3,2)$ is
the conformal algebra of three-dimensional Minkowski spacetime.  This
Lie algebra is also isomorphic to the isometry algebra of $\AdS_4$ of
which (the conformal compactification of) Minkowski spacetime is the
conformal boundary.

The original observation in \cite{Malda} makes use of the fact that
the near-horizon geometry of the elementary M2-brane solution of
eleven-dimensional supergravity \cite{DS2brane} is isometric to
$\AdS_4 \times S^7$ \cite{GibbonsTownsend}.  This is a maximally
supersymmetric background of eleven-dimensional supergravity
\cite{FreundRubin}, its Killing superalgebra is isomorphic to
$\fosp(8|4)$ and hence the dual superconformal field theory has $N=8$
supersymmetry.  One can replace $S^7$ by other manifolds admitting
real Killing spinors and in this way obtain backgrounds with Killing
superalgebra $\fosp(N|4)$ for lower values of $N$
\cite{AFHS,MorrisonPlesser}.  Recently the classification of smooth
Freund--Rubin backgrounds of the form $\AdS_4 \times X^7$ with
$N\geq 4$ has been achieved \cite{deMedeiros:2009pp}: they are
necessarily such that $X = S^7/\Gamma$, where $\Gamma < \Spin(8)$ is a
discrete group acting freely on $S^7$ and described as the image of a
twisted embeddings of an ADE subgroup of quaternions.  A
classification of singular quotients with $N\geq 4$ has also recently
been obtained \cite{deMedeiros:2010dn}, this time in terms of fibered
products of ADE subgroups.

The question remains whether there are any eleven-dimensional
supergravity backgrounds with Killing superalgebra isomorphic to
$\fosp(N|4)$ but which are \emph{not} Freund--Rubin backgrounds of the
form $\AdS_4 \times X^7$.  Classifying such backgrounds would complete
the determination of possible dual geometries to three-dimensional
superconformal field theories.  The purpose of this paper is to
investigate their existence.

It has recently been shown \cite{HomogThm} that backgrounds preserving
more than half of the supersymmetry --- i.e., $N>4$ in the present
context --- are (locally) homogeneous and moreover that it is the
group whose Lie algebra is generated by the Killing spinors of the
background which already acts transitively.  This allows us to
restrict ourselves to backgrounds which are homogeneous under a
prescribed group.  Homogeneous lorentzian manifolds can be described
locally by a pair $(\fg,\fh)$, where $\fh$ is a Lie subalgebra of
$\fg$ preserving a lorentzian inner product on the representation
$\fg/\fh$ induced by the restriction to $\fh$ of the adjoint
representation of $\fg$.  In addition, the group corresponding to
$\fh$ must be a closed subgroup of the group corresponding to $\fg$.
In this paper we are interested in the particular case where
$\fg = \fso(n) \oplus \fso(3,2)$ and $\fh$ has dimension
$\binom{n}{2} - 1$ for $n>4$.  Given the huge number of such
subalgebras, this task seems at first to be impractical or at the very
least, very tiresome.  Luckily, the fact that $\fg$ is semisimple,
allows us to exploit a wonderful theorem by Nadine Kowalsky
\cite{MR1426887}, generalised by Deffaf, Melnick and Zeghib
\cite{MR2421545}, which characterises those homogeneous lorentzian
manifolds of semisimple Lie groups.  Such lorentzian manifolds come in
two flavours: either the action is proper, in which case $\fh$ is the
Lie algebra of a compact group, or else the manifold is locally
isometric to the product of (anti) de Sitter space with a riemannian
homogeneous manifold.  In either case, we can essentially restrict to
compact subalgebras $\fh$, which are much better known, not to mention
much fewer in number.

We therefore set ourselves two tasks in this paper.  The first is the
classification (up to local isometry) of homogeneous backgrounds with
an effective and locally transitive action of
$\fg = \fso(n) \oplus \fso(3,2)$ for $n>4$, where the geometry is not
of the form $\AdS \times X$.  To this end we will first determine the
compact Lie subalgebras of $\fg$ (of the right dimension), up to the
action of automorphisms.  Lie subalgebras can be found by iterating
the simpler problem of finding maximal compact subalgebras.  Since
$\fg$ is a product, this requires the Lie algebra version of Goursat's
Lemma characterising the subgroups of a direct product of groups,
which curiously plays such a crucial rôle in the results of
\cite{deMedeiros:2009pp,deMedeiros:2010dn}.  The second task is the
classification (again up to local isometry) of homogeneous (anti) de
Sitter backgrounds with a locally transitive action of
$\fso(n) \oplus \fso(3,2)$ for $n>4$, but which are \emph{not} of
Freund--Rubin type; that is, where the flux is not just equal to the
volume form of a four-dimensional factor in the geometry.

The paper is organised as follows.  In
Section~\ref{sec:homog-lorentz-manif} we review the geometry of
homogeneous lorentzian manifolds.  In Section
\ref{sec:basic-notions-about} we settle the notation and discuss the
basics of homogeneous geometry, specialising at the end on the
lorentzian case and review briefly the results of Kowalsky and of
Deffaf, Melnick and Zeghib.  In Section \ref{sec:comp-homog-spac} we
record the necessary formulae to perform geometric computations on
homogeneous lorentzian manifolds.  In
Section~\ref{sec:lie-subalg-prod} we prove a Goursat-type lemma for
Lie algebras, characterising the Lie subalgebras of a direct product
of two Lie algebras in terms of fibered products of Lie subalgebras of
each of them.  In Section~\ref{sec:lie-subalgebras-fson-2} we record
the Lie subalgebras of $\fso(n)$ for low values of $n$.  This is
well-known material (covered, for example, in \cite{Slansky}) but our
purpose here is to have concrete formulae for the generators of the
subalgebras in terms of the standard basis of $\fso(n)$.  In
Section~\ref{sec:supergr-field-equat} we write down the field
equations for $d=11$ supergravity in a homogeneous Ansatz, which
become a system of algebraic equations in the parameters for the
metric and the 4-form.  This allows us to describe our methodology in
some detail in Section~\ref{sec:methodology}.  Section
\ref{sec:non-AdS-backgrounds} contains our solution of the first of
the above two tasks: the determination of homogeneous backgrounds of
$\SO(n)\times\SO(3,2)$ for $n>4$, which are not of anti-de Sitter
type.  We will show that there are no (new) $n>5$ backgrounds, but we
will exhibit a number of new (at least to us) backgrounds for $n=5$,
at least one of which is supersymmetric, albeit with only $N=2$.  We
will explore its geometry in more detail in
Section~\ref{sec:freund-rubin-backgr}, where we show that it is a
Freund--Rubin background with underlying geometry $S^4 \times P^7$,
where $P$ is seven-dimensional, lorentzian Sasaki--Einstein.  We also
show that the background can be obtained by a ``Wick rotation'' from a
known homogeneous $\AdS_4$ Freund--Rubin background.  In
Section~\ref{sec:homogeneous-anti-de} we tackle the second of the two
tasks above, namely: the determination of $\AdS_4$ backgrounds which
are not of Freund--Rubin type.  In Section~\ref{sec:there-are-no} we
show that there are no de Sitter backgrounds, and we exhibit a number
of new (to us) backgrounds for $n=5$ in addition to recovering some
well-known backgrounds with $n=7$ (Englert), $n=6$ (Pope--Warner) and
$n=5$ (Castellani--Romans--Warner).  In the first version of this
paper we were unable to fully analyse one case.  In this version we
include an appendix written by Alexander~S.~Haupt which closes that
gap, albeit without finding any new backgrounds.  Finally, in
Section~\ref{sec:geometry-n=5} we discuss the geometry of some of the
$n=5$ backgrounds found above: some of the backgrounds can only be
approximated numerically, and we will have little else to say about
them beyond their existence.  In particular, using the method
described in Appendix~\ref{sec:isom-homog-space}, we determine the
actual isometry group of the backgrounds, which in some cases is
slightly larger than $\SO(n) \times \SO(3,2)$.  The paper ends with
two appendices: an appendix on the determination of the full isometry
algebra of a homogeneous riemannian manifold and the appendix by Haupt
mentioned above.

\section{Homogeneous lorentzian manifolds}
\label{sec:homog-lorentz-manif}

In this section we review the basic notions concerning homogeneous
spaces and the useful formulae for reducing their differential
geometry to Lie algebraic data.

\subsection{Basic notions about homogeneous spaces}
\label{sec:basic-notions-about}

A lorentzian manifold $(M,g)$ is homogeneous if it admits a transitive
action of some Lie group by isometries.  In other words, $(M,g)$ is
homogeneous if there is a Lie group $G$ acting on $M$ smoothly,
preserving the metric and such that any two points of $M$ are related
by some element of $G$.

Let us unpack this definition.  First of all, we have an action of $G$
on $M$.  This is a smooth map $\alpha: G \times M \to M$, which we
will denote simply by $(x,p) \mapsto x \cdot p$, such that for all
$x_i \in G$ and $p \in M$, $x_1\cdot (x_2 \cdot p) = (x_1 x_2) \cdot
p$ and $e \cdot p = p$, where $e$ denotes the identity element of $G$.
The action is transitive if for some (and hence all) $p \in M$, the
map $\alpha_p: G \to M$, defined by $\alpha_p(x) = x\cdot p$ is
surjective.  The group acts by isometries if the diffeomorphisms
$\alpha_x : M \to M$, defined by $\alpha_x(p) = x \cdot p$, preserve
the metric; that is, $\alpha_x^* g = g$ for all $x \in G$.

In a homogeneous space, every point is equivalent to any other point.
Let us choose a point $o \in M$ and let us think of it as the
\emph{origin} of $M$.  Let $H_o$ denote the subgroup of $G$ which
fixes the point $o$.  Since $H_o = \alpha_o^{-1}(\{o\})$ is the
inverse image of a point under a continuous map, it is a closed
subgroup of $G$.  We call $H_o$ the \textbf{stabiliser} (subgroup) of
$o$.  Then $M$ is diffeomorphic to the space $G/H_o$ of right
$H_o$-cosets in $G$.  The point $o\in M$ corresponds to the identity
coset, whereas the point $x\cdot o$ corresponds to the coset $xH_o$,
since any one of the group elements in the coset $xH_o$ takes $o$ to
$x \cdot o$.  The differential $(\alpha_x)_*$ defines a family of
linear maps $T_o M \to T_{x \cdot o}M$.  If $x \in H_o$, then the
differential at $o$ is an invertible linear transformation of $T_o M$.
This is called the \textbf{isotropy representation} of $H_o$ on $T_o
M$.

The metric $g$ defines a lorentzian inner product $g_o$ on each
tangent space $T_o M$.  The condition $\alpha_x^* g = g$ becomes that
for all $o \in M$, $\alpha_x^* g_{x\cdot o} = g_o$.  In particular, if
$x \in H_o$, $\alpha_x^* g_o = g_o$, whence the isotropy
representation of $H_o$ is orthogonal with respect to $g_o$.

Let $\fg$ denote the Lie algebra of $G$, whose underlying vector space
we take to be the tangent space $T_eG$ at the identity in $G$.  Let
$\fh_o$ denote the Lie subalgebra of $G$ corresponding to the
stabiliser subgroup $H_o$ of $o$.  The differential at $e\in G$ of the
map $\alpha_o : G \to M$ is a linear map $\fg \to T_o M$ which is
surjective because the action is transitive and has kernel precisely
$\fh_o$.  In other words, we have an exact sequence
\begin{equation}
  \label{eq:SESaction}
  \begin{CD}
    0 @>>> \fh_o @>>> \fg @>(\alpha_o)_*>> T_o M @>>> 0~,
  \end{CD}
\end{equation}
not just of vector spaces, but in fact of $\fh_o$-modules.  Indeed,
$\fh_o$ acts on $\fg$ by restricting the adjoint representation of
$\fg$ to $\fh_o$, and $\fh_o$ is a submodule precisely because $\fh_o$
is a Lie subalgebra.  Then $T_o M$ is isomorphic as an $\fh_o$-module
to $\fg/\fh_o$.  This representation is none other than the
linearisation of the isotropy representation of $H_o$ on $T_o M$.  Let
us prove this.

Let $h(t)$ be a regular curve in $H_o$ such that $h(0)=e$.  Then
$\alpha_{h(t)}: T_o M \to T_o M$ and the action of $h'(0) \in \fh_o$
on $T_o M$ is obtained by differentiating at $t=0$.  Indeed, let $v
\in T_o M$ and choose any regular curve $\gamma(s)$ on $M$ with
$\gamma(0)=o$ and $\gamma'(0)=v$.  Then
\begin{equation}
  \begin{aligned}[m]
    h'(0) \cdot v &= \frac{d}{dt}\biggr|_{t=0} (\alpha_{h(t)})_* v\\
    &= \frac{d}{dt}\biggr|_{t=0} \frac{d}{ds}\biggr|_{s=0} \alpha_{h(t)}(\gamma(s))\\
    &= \frac{d}{dt}\biggr|_{t=0} \frac{d}{ds}\biggr|_{s=0} h(t) \cdot \gamma(s)~.
  \end{aligned}
\end{equation}
Now let $g(s)$ be a regular curve in $G$ with $g(s) \cdot o = \gamma(s)$ and $g(0)=e$.  Then
\begin{equation}
  \begin{aligned}[m]
    h'(0) \cdot v &= \frac{d}{dt}\biggr|_{t=0} \frac{d}{ds}\biggr|_{s=0} h(t) \cdot (g(s) \cdot o)\\
    &= \frac{d}{dt}\biggr|_{t=0} \frac{d}{ds}\biggr|_{s=0} (h(t) g(s)) \cdot o\\
    &= \frac{d}{dt}\biggr|_{t=0} \frac{d}{ds}\biggr|_{s=0} (h(t) g(s) h(t)^{-1})) \cdot o\\
    &= \frac{d}{dt}\biggr|_{t=0} \frac{d}{ds}\biggr|_{s=0} \alpha_o(h(t) g(s) h(t)^{-1}))\\
    &= \frac{d}{dt}\biggr|_{t=0} (\alpha_o)_* (h(t) g'(0) h(t)^{-1}))\\
    &= (\alpha_o)_* ([h'(0), g'(0)])~,
 \end{aligned}
\end{equation}
where $(\alpha_o)_* g'(0) = v$.  Choosing a different curve $\tilde
g(s)$ with $\tilde g(s)\cdot o = g(s)\cdot o$, then $\tilde h(s) =
g(s)^{-1}g(s)$ is a curve in $H_o$ with $\tilde h(0)=e$.  This means
$\tilde g(s) = g(s)\tilde h(s)$, whence ${\tilde g}'(0) = g'(0) +
{\tilde h}'(0)$, but ${\tilde h}'(0) \in \fh_o$ and hence
$(\alpha_o)_*([h'(0),{\tilde h}'(0)])=0$, so that $h'(0) \cdot v$ is
unchanged.  In other words, in order to compute the action of $X \in
\fh_o$ on $v \in T_o M$, we choose $Y \in \fg$ with $(\alpha_o)_* Y =
v$ and then compute $(\alpha_o)_*([X,Y])$, which is independent of the
lift $Y$ of $v$.

When $H_o$ is connected, the isotropy representation of $H_o$ is
determined by the above representation of $\fh_o$.   In practice we
will assume without loss of generality that $M$ is simply connected
and then the exact homotopy sequence of the principal $H_o$-bundle $G
\to M$ will imply that $H_o$ is connected.

We can realise the linear isotropy representation explicitly by
choosing a complement $\fm$ of $\fh_o$ in $\fg$, so that $\fg = \fh_o
\oplus \fm$, and defining the action of $X\in \fh_o$ on $Y \in \fm$ by
\begin{equation}
  X \cdot Y =  [X, Y]_\fm~,
\end{equation}
where, here and in the following, the subscript $\fm$ indicates the
projection onto $\fm$ along $\fh_o$; that is, we simply discard the
$\fh_o$-component of $[X, Y]$.  If $\fm$ is stable under $\ad(\fh_o)$,
so that the projection is superfluous, we say that $\fg = \fh_o\oplus
\fm$ is a \textbf{reductive split}, and the pair $(\fg,\fh_o)$ is said
to be \textbf{reductive}.  This is equivalent to the splitting (in the
sense of homological algebra) of the exact sequence
\eqref{eq:SESaction} in the category of $\fh_o$-modules.  In this
case, one often says that $(M,g)$ is reductive; although this is an
abuse of notation in that reductivity is not an intrinsic property of
the homogeneous space, but of its description as an orbit of $G$.  Not
all lorentzian homogeneous manifolds need admit a reductive
description; although it is known to be the case in dimension $\leq 4$
as a consequence of the classifications \cite{MR2287304,Komrakov}.

Different points of $M$ can have different stabilisers, but these are
conjugate in $G$, hence in particular they are isomorphic.  This is
why one often abbreviates homogeneous spaces as $G/H$, where $H$
denotes one of the $H_o$ subgroups of $G$.  Let $\fg$ denote the Lie
algebra of $G$ and let $\fh$ denote the Lie subalgebra corresponding
to the subgroup $H$.  Then a lorentzian homogeneous manifold is
described locally by a pair $(\fg,\fh)$ and an $\fh$-invariant
lorentzian inner product on $\fg/\fh$, with the proviso that $\fh$ is
the Lie algebra of a \emph{closed} subgroup of $G$.

We are interested in classifying (simply-connected) eleven-dimensional
homogeneous lorentzian manifolds with a transitive action of the
universal covering group $G$ of $\SO(n) \times \SO(3,2)$ for $n>4$.
The above discussion might suggest the problem of classifying those
Lie subalgebras $\fh$ of $\fg = \fso(n) \oplus \fso(3,2)$ of the right
dimension: namely, $\dim \fh = \binom{n}{2} -1$, which are the Lie
algebras of a closed subgroup $H$ of $G$.  Even in the relatively low
dimension we are working in, the classification of Lie subalgebras of
a semisimple Lie algebra can be a daunting task (see, e.g.,
\cite{MR0455998} for the low-dimensional (anti) de Sitter algebras).
Luckily, since $G$ is semisimple we may appeal to results of Nadine
Kowalsky \cite{MR1426887} and Deffaf, Melnick and Zeghib
\cite{MR2421545}, which reduce the task at hand considerably by
allowing us to focus on Lie algebras $\fh$ of compact subgroups of
$G$.  We will highlight the main results, which we learnt from the
recent paper \cite{Alekseevsky:2011p9355} by Dmitri Alekseevsky.

Let us recall that a continuous map between topological spaces is
called proper if the inverse image of a compact set is compact.  If
$G$ is a Lie group acting on a manifold $M$, we say that the action is
\emph{proper} if the map $f: G \times M \to M \times M$, defined by
$f(a,x) = (ax,x)$, is proper.  Given a proper action of $G$ on $M$, we
notice that $f^{-1}(x,x) = \left\{(a,x)\middle | ax =x \right\} = H
\times \{ x \}$, where $H$ is the stabiliser of $x$.  Since the action
is proper and $\{(x,x)\}$ is a compact set, so is $H$.  Now suppose
that $G$ acts properly and transitively on $M$, so that $M \cong G/H$
with $H$ compact.  Then by averaging over $H$, we can assume that the
linear isotropy representation of $H$ on $\fm$ leaves invariant a
positive-definite inner product.  In particular, $M = G/H$ is a
reductive homogeneous space.  It is proved in
\cite[Prop.~4]{Alekseevsky:2011p9355} that $M$ admits a $G$-invariant
lorentzian metric if and only if the linear isotropy representation of
$H$ leaves a line $\ell \subset \fm$ invariant.  Then letting $h$
denote the positive-definite inner product on $\fm$ and $\alpha \in
\fm^*$ such that $\ker \alpha = \ell^\perp$, where $\ell^\perp$ is the
$h$-perpendicular complement of $\ell$ in $\fm$, the $G$-invariant
lorentzian metrics on $M$ are obtained from the inner products
\begin{equation}
  h - \lambda \alpha \otimes \alpha~,
\end{equation}
which are lorentzian for $\lambda \gg 0$.

What about if the action of $G$ on $M$ is \emph{not} proper?  It is a
remarkable result \cite{MR1426887} of Nadine Kowalsky's that if a
simple Lie group $G$ acts transitively by isometries on a lorentzian
manifold $M$ in such a way that the action is not proper, then $M$ is
locally isometric to (anti) de Sitter spacetime.  Deffaf, Melnick and
Zeghib \cite{MR2421545} extended this result to the case of $G$
semisimple, with the conclusion that $M$ is now locally isometric to
the product of (anti) de Sitter spacetime with a riemannian
homogeneous space.  Notice that in either case, we can always describe
$M$ as a reductive homogeneous space.

These results will allow us to consider either $\AdS_d \times
M^{11-d}$ backgrounds (one can show that there are no de Sitter
backgrounds) or else restrict ourselves to the case of compact $H$.
Supersymmetric Freund--Rubin backgrounds with $N\geq 4$ of the form
$\AdS \times M$ have been classified --- see \cite{deMedeiros:2009pp}
for the smooth case and \cite{deMedeiros:2010dn} for orbifolds --- but
we still need to investigate more general anti de Sitter backgrounds
with flux along the internal manifold $M$.  This problem was studied
in the early Kaluza--Klein supergravity literature, albeit not
exhaustively (see, e.g., \cite{DNP} and references therein, for the
progress on this problem circa 1985).  We will re-examine such
backgrounds, recover the known ones and exhibit ones which to our
knowledge are new.  Concerning the latter class of backgrounds, those
with compact $H$, we must in principle distinguish between two cases:
when $\SO(3,2)$ acts effectively and when it acts trivially; although
the latter case is of dubious relevance to the AdS/CFT correspondence
and will be ignored in this paper.  In the former case we must look
for compact Lie subalgebras $\fh$ of $\fso(n) \oplus \fso(3) \oplus
\fso(2)$, which is the maximally compact subalgebra of $\fg$, whereas
in the latter case we must look for Lie subalgebras $\fk$ of
$\fso(n)$, with then $\fh = \fk \oplus \fso(3,2)$.  We will classify
all such Lie subalgebras admitting an $\fh$-invariant lorentzian inner
product on $\fm$.  Since $\fg$ is a direct product, this will require
us to learn how to determine the Lie subalgebras of a direct product
of Lie algebras.  This will be explained in Section
\ref{sec:lie-subalg-prod}, but not before collecting some useful
formulae to do calculations in lorentzian homogeneous spaces.

\subsection{Computations in homogeneous spaces}
\label{sec:comp-homog-spac}

The purpose of this section, which overlaps with \cite[§2.3]{FMPHomPL}
somewhat, is to record some useful formulae for doing calculations in
reductive homogeneous spaces in terms of Lie algebraic data.  For more
details one can consult, for example, the book \cite{Besse}.

Let $M=G/H$ be a reductive homogeneous space with $H$ a closed
connected subgroup of $G$ and let $\fg = \fh \oplus \fm$ be a
reductive split.  The isotropy representation of $\fh$ on $\fm$ is the
restriction of the adjoint action: $X \cdot Y  = [X,Y]$, for $X \in
\fh$ and $Y \in \fm$.  Let $\left<-,-\right>$ denote an inner product
on $\fm$ which is invariant under the isotropy representation; that
is, for all $X,Y \in \fm$ and $Z \in \fh$,
\begin{equation}
  \left<[Z,X], Y \right> +   \left<X, [Z,Y] \right>~.
\end{equation}
This defines a $G$-invariant metric on $M$.

More generally, there is a one-to-one correspondence between
$\fh$-invariant tensors on $\fm$ and $G$-invariant tensor fields on
$M$.  If $F$ is a $G$-invariant tensor field on $M$, its evaluation at
$o$ together with the identification of $T_oM$ with $\fm$ defines a
tensor $F_o$ on $\fm$.  Since $F$ is $G$-invariant, its Lie derivative
at $o$ along any Killing vector vanishes.  Now let $X$ be a Killing
vector coming from $\fh$.  Since its value at $o$ vanishes, the Lie
derivative along $X$ is the action of the corresponding element of
$\fh$ under the linear isotropy representation.  Therefore $F_o$ is
$\fh$-invariant.  Conversely, let $F_o$ be an $\fh$-invariant tensor
on $\fm$.  We define a tensor field $F$ on $M$ by the condition $F(x)
= a \cdot F_o$, where $a \in G$ is such that $a\cdot o = x$, which
exists since $G$ acts transitively.  This is actually well defined
because $F_0$ is $H$-invariant.  Indeed, let $b\in G$ be such that
$b\cdot o = x$.  Then $b^{-1}a \cdot o = o$, whence $b^{-1}a \in H$.
Therefore $b \cdot F_o = b\cdot b^{-1}a\cdot F_o  = a \cdot F_o$.  The
tensor field $F$ so defined is clearly $G$-invariant, since for all
$a\in G$ and $x\in M$, $F(a\cdot x) = a \cdot F(x)$, since both sides
equal $ab \cdot F_o$, where $b \in G$ is any element such that $b\cdot
o = x$.

Let $X,Y,Z$ be Killing vectors on $M = G/H$.  The Koszul formula for
the Levi-Civita connection reads
\begin{equation}
  \label{eq:Koszul}
  2 g(\nabla_X Y, Z) = g([X,Y],Z) + g([X,Z],Y) + g(X,[Y,Z])~.
\end{equation}
At the identity coset $o\in M$ and assuming that $X,Y,Z$ are Killing
vectors in $\fm$, the chosen complement of $\fh$ in $\fg$, then
\begin{equation}
  \label{eq:levicivita}
  \nabla_X Y\bigr|_o = -\half [X,Y]_{\fm} + U(X,Y)~,
\end{equation}
where $U:\fm \times \fm \to \fm$ is a symmetric tensor given
by\footnote{The apparent difference in sign between equation
  \eqref{eq:Koszul} and equations \eqref{eq:levicivita} and
  \eqref{eq:Utensor} stems from the fact that Killing vectors on $G/H$
  generate left translations on $G$, whence they are
  right-invariant. Thus the map $\fg \to \text{Killing vector fields}$
  is an anti-homomorphism.}
\begin{equation}
  \label{eq:Utensor}
  2\left<U(X,Y),Z\right> = \left<[Z,X]_{\fm},Y\right> + \left<[Z,Y]_{\fm},X\right>~,
\end{equation}
for all $Z\in\fm$ and where the subscript denotes the projection of
$[Z,X]\in\fg$ to $\fm$.   It should be remarked that
\eqref{eq:levicivita} is only valid at $o\in M$, since $\nabla_X Y$ is
not generally a Killing vector.  Of course, since $\nabla$ is
G-invariant, then one can determine $\nabla_X Y\bigr|_p$ at any other
point by acting with any isometry relating $o$ and $p$.

For a reductive homogeneous space, the $U$-tensor is invariant under
the linear isotropy representation.  The vanishing of the $U$-tensor
characterises the class of homogeneous spaces known as
\textbf{naturally reductive}.  In those spaces, the geodesics of the
invariant connection and the Levi-Civita connection agree.

The Riemann curvature tensor is $G$-invariant and it can be computed
at $o$.  One obtains, for $X,Y,Z,W$ vectors in $\fm$, the curvature
tensor at $o$ is given by
\begin{multline}
  \label{eq:riemann}
  R(X,Y,Z,W) = \left<U(X,W),U(Y,Z)\right> - \left<U(X,Z),U(Y,W)\right>\\
  + \tfrac1{12} \left<[X,[Y,Z]]_{\fm},W\right> - \tfrac1{12} \left<[X,[Y,W]]_{\fm},Z\right>\\
  - \tfrac16 \left<[X,[Z,W]]_{\fm},Y\right> - \tfrac1{12} \left<[Y,[X,Z]]_{\fm},W\right>\\
  + \tfrac1{12} \left<[Y,[X,W]]_{\fm},Z\right> + \tfrac16 \left<[Y,[Z,W]]_{\fm},X\right>\\
  - \tfrac16 \left<[Z,[X,Y]]_{\fm},W\right> - \tfrac1{12} \left<[Z,[X,W]]_{\fm},Y\right>\\
  + \tfrac1{12} \left<[Z,[Y,W]]_{\fm},X\right> + \tfrac16 \left<[W,[X,Y]]_{\fm},Z\right>\\
  + \tfrac1{12} \left<[W,[X,Z]]_{\fm},Y\right> - \tfrac1{12} \left<[W,[Y,Z]]_{\fm},X\right>\\
  - \half \left<[X,Y]_{\fm},[Z,W]_{\fm}\right> - \tfrac14 \left<[X,Z]_{\fm},[Y,W]_{\fm}\right> + \tfrac14 \left<[X,W]_{\fm},[Y,Z]_{\fm}\right>~,
\end{multline}
which can be obtained by polarisation from the simpler expression for
$K(X,Y):=\left<R(X,Y)X,Y\right>$, which is also easier to derive.
Indeed, and for completeness, one has
\begin{multline}
 6 R(X,Y,Z,W) = K(X+Z,Y+W) - K(Y+Z,X+W)\\
  - K(Y+W,X) + K(Y+Z,X) - K(X+Z,Y) + K(X+W,Y)\\
  - K(Y+W,Z) + K(X+W,Z) - K(X+Z,W) + K(Y+Z,W)\\
  + K(X,W) - K(X,W) - K(Y,W) + K(Y,Z) - K(X,Z)~,
\end{multline}
where
\begin{equation}
  K(X,Y) = - \tfrac34 |[X,Y]_{\fm}|^2 - \half \left<[X,[X,Y]]_{\fm},Y\right> - \half \left<[Y,[Y,X]]_{\fm},X\right>
  + |U(X,Y)|^2 - \left<U(X,X),U(Y,Y)\right>
\end{equation}
and where $|-|^2$ is the (indefinite) norm associated to
$\left<-,-\right>$.

Similarly, we can obtain the Ricci tensor by polarisation from
\begin{multline}
  \label{eq:ricci}
  \Ric(X,X) = - \half \sum_i \left<[X,X_i]_\fm,[X,X^i]_\fm\right> - \half \sum_i \left<[X,[X,X_i]_\fm]_\fm,X^i\right> - \sum_i \left<[X,[X,X_i]_\fh],X^i\right>\\
  - \sum_i \left<[U(X_i,X^i),X]_\fm,X\right> + \tfrac14 \sum_{i,j} \left<[X_i,X_j]_\fm,X\right>\left<[X^i,X^j]_\fm,X\right>~,
\end{multline}
where $X_i$ is a pseudo-orthonormal basis with
$\left<X_i,X^j\right>=\delta_i^j$.  The Ricci scalar is given by $R =
\sum_i \Ric(X_i,X^i)$.

It is convenient to write down the expression for the Ricci tensor in
terms of a local frame, since this is what is used in computations.
So let $Y_i$ denote a basis for $\fm$ with $\left<Y_i,Y_j\right>=
g_{ij}$ and let $X_a$ denote a basis for $\fh$.  The structure
constants are $[X_a,Y_i] = f_{ai}{}^j Y_j$ (assumed reductive) and
$[Y_i, Y_j] = f_{ij}{}^k Y_k + f_{ij}{}^a X_a$.  We can raise and
lower $\fm$-indices using $g$.  In this notation, we find that the
Ricci tensor is given by:
\begin{equation}
  \label{eq:riccicomps}
  R_{ij} = -\half f_{i}{}^{k\ell} f_{jk\ell} - \half f_{ik}{}^\ell f_{j\ell}{}^k + \half f_{ik}{}^a f_{aj}{}^k + \half f_{jk}{}^a f_{ai}{}^k - \half f_{k\ell}{}^\ell f^k{}_{ij} - \half f_{k\ell}{}^\ell f^k{}_{ji} + \tfrac14 f_{k\ell i} f^{k\ell}{}_j~.
\end{equation}

Let $\Omega^\bullet(M)$ denote the de~Rham complex on $M$ and
$\Omega^\bullet(M)^G$ the subcomplex of $G$-invariant differential
forms.  The value at $o\in M$ of a $G$-invariant differential $k$-form
$\omega$ on $M$ is an $H$-invariant element of $\Lambda^k \fm^*$.  Its
exterior derivative and its codifferential can be expressed purely in
terms of the Lie algebraic data defining the homogeneous space.  If
$X_i$ are Killing vectors in $\fm$, then the exterior derivative of
$\omega$ is given by
\begin{equation}
  \label{eq:d}
  d\omega(X_1,\dots,X_{k+1}) = \sum_{1\leq i < j \leq k+1} (-1)^{i+j} \omega([X_i,X_j]_\fm, X_1,\dots,\widehat{X_i},\dots,\widehat{X_j},\dots,X_{k+1})~,
\end{equation}
where a hat adorning a symbol denotes its omission.  Perhaps the
simplest proof of this statement is to localise the complex
$\Omega^\bullet(M)^G$ as the subcomplex of left-invariant differential
forms on $G$ which are \emph{basic}.  In other words, we view $M$ as
the base of a principal $H$-bundle with total space $G$.  A
$G$-invariant differential form $\omega$ on $M$ pulls back to a
left-invariant form on $G$ whose value $\omega_e$ at the identity is
both \emph{horizontal}: $\imath_X \omega_e = 0$ for all $X \in \fh$,
and \emph{invariant} under the adjoint action of $\fh$.  We then use
the standard formulae (see, e.g., \cite{ChevalleyEilenberg}) for the
differential of a left-invariant form on $G$, after checking that the
basic forms indeed form a subcomplex.

In computations, a more convenient way to compute the exterior
derivative of an invariant form is the following.  Let $(Y_i)$ be a
basis for $\fm$ such that $[Y_i,Y_j]_\fm = \sum_k f_{ij}{}^k Y_k$.
Then let $(\theta^i)$ be the canonically dual basis for $\fm^*$.  Then
it follows from equation \eqref{eq:d} that
\begin{equation}
  d\theta^k = -\half \sum_{i,j} f_{ij}{}^k \theta^i\wedge\theta^j~.
\end{equation}
We then extend $d$ as a derivation to a general invariant form.
Therefore, if $F$ is an invariant $4$-form, so that $F = \tfrac1{4!}
F_{ijkl} \theta^{ijkl}$ (with the Einstein summation convention in
force), then
\begin{equation}
  \label{eq:dF}
  dF = -\tfrac1{12}f_{mn}{}^i F_{ijkl}  \theta^{jklmn}~,
\end{equation}
or explicitly,
\begin{multline}
  \label{eq:dFcomps}
  (dF)_{jklmn} = f_{jk}{}^i F_{ilmn} - f_{jl}{}^i F_{ikmn} +
  f_{jm}{}^i F_{ikln} - f_{jn}{}^i F_{iklm}  + f_{kl}{}^i F_{ijmn}\\
  - f_{km}{}^i F_{ijln}  + f_{kn}{}^i F_{ijlm}  + f_{lm}{}^i F_{ijkn}
  - f_{ln}{}^i F_{ijkm}  + f_{mn}{}^i F_{ijkl}~.
\end{multline}

To describe the codifferential, let us introduce dual bases $(Y_i)$
and $(Y^i)$ for $\fm$ such that $\left<Y_i,Y^j\right> = \delta_i^j$.
Then we have
\begin{multline}
  \label{eq:delta}
  \delta\omega(X_1,\dots,X_{k-1}) = \sum_{i=1}^{\dim M} \sum_{j=1}^{k-1} \omega(Y^i, X_1,\dots,-\half[Y_i,X_j]_\fm -U(Y_i,X_j),\dots,X_{k-1})\\
  - \sum_{i=1}^{\dim M} \omega(U(Y^i,Y_i),X_1,\dots,X_{k-1})~.
\end{multline}

We can write this in terms of a frame $Y_i$ for $\fm$, which is
perhaps more useful in computations.  For $F$ an invariant 4-form, we
have
\begin{equation}
  \label{eq:deltaF}
  (\delta F)_{ijk} = - (\half f_{mi}{}^n + U_{mi}{}^n) F^m{}_{njk} - (\half f_{mj}{}^n + U_{mj}{}^n) F^m{}_{ink} - (\half f_{mk}{}^n + U_{mk}{}^n) F^m{}_{ijn} - U_m{}^{mn} F_{nijk}~.
\end{equation}

\section{Lie subalgebras of a direct product}
\label{sec:lie-subalg-prod}

In this section we prove a result characterising Lie subalgebras of
the direct product of two Lie algebras. This result is necessary for
the determination of the Lie subalgebras of $\fso(n)\oplus \fso(3,2)$.
It is by no means original, but we know of no good reference.

Let $\fg_L$ and $\fg_R$ be two real Lie algebras and let $\fg = \fg_L
\oplus \fg_R$ be their product.  Elements of $\fg$ are pairs
$(X_L,X_R)$ with $X_L \in \fg_L$ and $X_R \in \fg_R$.  The Lie bracket
in $\fg$ of two such elements $(X_L,X_R)$ and $(Y_L,Y_R)$ is given by
the pair $([X_L,Y_L], [X_R,Y_R])$.

We are interested in Lie subalgebras $\fh$ of $\fg$.  This is
analogous to the determination of subgroups of a product group, which
is solved by Goursat's Lemma \cite{Goursat}.  As a result we will also
call this the Goursat Lemma for Lie algebras.

Let $\pi_L : \fg \to \fg_L$ and $\pi_R : \fg \to \fg_R$ denote the
projections onto each factor: they are Lie algebra homomorphisms.  Let
$\fh_L$ and $\fh_R$ denote, respectively, the image of the subalgebra
$\fh$ under $\pi_L$ and $\pi_R$.  They are Lie subalgebras of $\fg_L$
and $\fg_R$, respectively.  Let us define $\fh^0_L := \pi_L(\ker \pi_R
\cap \fh)$ and $\fh^0_R := \pi_R(\ker \pi_L \cap \fh)$.  One checks
that they are ideals of $\fh_L$ and $\fh_R$, respectively.  This means
that on $\fh_L/\fh^0_L$ and $\fh_R/\fh^0_R$ we can define Lie algebra
structures.  Goursat's Lemma says that these two Lie algebras are
isomorphic.  Let us understand this.

The Lie algebra $\fh_L$ consists of those $X_L \in \fg_L$ such that
there is some $X_R \in \fg_R$ with $X_L + X_R \in \fh$, and similarly
$\fh_R$ consists of those $X_R \in \fg_R$ such that there is some $X_L
\in \fg_L$ with $X_L + X_R \in \fh$.  At the same time, $\fh^0_L$
consists of those $X_L \in \fg_L$ which are also in $\fh$, whereas
$\fh^0_R$ consists of those $X_R \in \fg_R$ which are also in $\fh$.
Let us define a linear map $\varphi: \fh_L \to \fh_R/\fh^0_R$ as
follows.  Let $X_L \in \fh_L$.  Then this means that there is some
$X_R \in \fh_R$ such that $X_L + X_R \in \fh$.  Define $\varphi(X_L) =
X_R \mod \fh^0_R$.  This map is well defined because if both $X_L +
X_R$ and $X_L + X'_R$ are in $\fh$, so is their difference, whence
$X_R - X'_R \in \fh^0_R$.  Now $\varphi$ is surjective, since for
every $X_R \in \fh_R$, there is some $X_L \in \fh_L$ with $X_L + X_R
\in \fh$, whence $\varphi(X_L) = X_R \mod \fh^0_R$.  Finally, the
kernel of $\varphi$ consists of those $X_L \in \fh_L$ such that there
is some $X_R \in \fh^0_R$ such that $X_L + X_R \in \fh$.  But $X_R \in
\fh$, whence $X_L \in \fh$ and hence $X_L \in \fh^0_L$.  Conversely if
$X_L \in \fh^0_L$, $X_L \in \fh$ so that $\varphi(X_L) = 0 \mod
\fh^0_R$, hence $\ker \varphi = \fh^0_L$.  In summary, $\varphi$
defines an isomorphism $\fh_L/\fh^0_L \cong \fh_R/\fh^0_R$.

Notice that the dimension of $\fh$ obeys
\begin{equation}
  \label{eq:dimh}
  \dim \fh = \dim \fh_L + \dim \fh^0_R = \dim \fh_R + \dim \fh^0_L~,
\end{equation}
as a consequence of the Euler--Poincaré principle applied to the exact
sequences
\begin{equation}
  \begin{CD}
    0 @>>> \fh^0_R @>>> \fh @>\pi_L>> \fh_L @>>> 0
  \end{CD}
\end{equation}
and
\begin{equation}
  \begin{CD}
    0 @>>> \fh^0_L @>>> \fh @>\pi_R>> \fh_R @>>> 0~.
  \end{CD}
\end{equation}

Goursat's Lemma suggests a systematic approach to the determination of
the Lie subalgebras of $\fg_L \oplus \fg_R$, which is particularly
feasible when $\fg_L$ and $\fg_R$ have low dimension.

Namely, we look for Lie subalgebras $\fh_L \subset \fg_L$ and $\fh_R
\subset \fg_R$ which have quotients isomorphic to $\fq$, say. Let
$f_L:\fh_L \to \fq$ and $f_R:\fh_R \to \fq$ be the corresponding
surjections.  Let $\varphi \in \Aut\fq$ denote an automorphism of
$\fq$.  Then we may form the Lie subalgebra $\fh_L
\oplus_{(\fq,\varphi)}\fh_R$ of $\fh_L \oplus \fh_R$ defined by
\begin{equation}
 \fh_L \oplus_{(\fq,\varphi)}\fh_R := \left\{(X_L,X_R) \in \fh_L \oplus \fh_R \middle | f_L(X_L) = \varphi(f_R(X_R))\right\}~.
\end{equation}
Of course, we need only consider automorphisms $\varphi$ which are not
induced by automorphisms of $\fh_L$ or $\fh_R$.  We record here the
following useful dimension formula which follows from equation
\eqref{eq:dimh}:
\begin{equation}
  \label{eq:dimgoursat}
 \dim \left(\fh_L \oplus_{(\fq,\varphi)}\fh_R\right) = \dim \fh_L +
 \dim \fh_R - \dim \fq~.
\end{equation}

A commonly occurring special case is when one of $\fh_L \to \fq$ or
$\fh_R \to \fq$ is an isomorphism.  For definiteness let us assume
that it is $\fh_R \to \fq$ which is an isomorphism.  Then we get a Lie
algebra homomorphism $\fh_L \to \fh_R$ obtained by composing $\fh_L
\to \fq$ with the inverse of $\fh_R \to \fq$.  In fact, we get a
family of such homomorphisms labelled by the automorphisms of $\fq$
or, equivalently, of $\fh_R$.  The fibered product which Goursat's
Lemma describes is now the graph in $\fh_L \oplus \fh_R$ of such a
homomorphism $\fh_L \to \fh_R$.  The resulting Lie algebra is
abstractly isomorphic to $\fh_L$.

\section{Lie subalgebras of $\fso(n)$}
\label{sec:lie-subalgebras-fson-2}

We first consider the Lie subalgebras of $\fso(n)$.  We will be
interested in $n\leq 7$, since the maximally supersymmetric
backgrounds have been classified \cite{FOPMax} and there are precisely
two such classes of backgrounds with $\fosp(8|4)$ Killing
superalgebra: namely, $\AdS_4 \times S^7$ and $\AdS_4 \times
S^7/\ZZ_2$.  For backgrounds of the form $\AdS_4 \times X^7$, it is
known that $N>6$ implies maximal supersymmetry, but this has no been
shown for more general backgrounds.  Let us work our way to $n=7$.

Let us say that a Lie subalgebra is \textbf{maximal} if it is proper
and is not properly contained in a proper Lie subalgebra.  Clearly, it
is enough to determine the maximal subalgebras and iterate in order to
determine all the proper subalgebras.  The maximal subalgebras of the
simple Lie algebras we shall be interested in have been tabulated in
\cite{Slansky} using methods introduced by Dynkin.

For us, the Lie algebra $\fso(n)$ is the real span of $L_{ab}$, for
$1\leq a<b\leq n$, with Lie brackets
\begin{equation}
  [L_{ab},L_{cd}] = \delta_{bc} L_{ad} - \delta_{ac} L_{bd} - \delta_{bd} L_{ac} + \delta_{ad} L_{bc}~.
\end{equation}
Notice that for any $k < n$, the subspace spanned by $L_{ab}$ where we
restrict $1 \leq a < b \leq k$ is a Lie subalgebra isomorphic to
$\fso(k)$.  We will attempt to label Lie algebras in such a way that
$\fso(k)$ will always denote this subalgebra.  Other subalgebras
isomorphic to $\fso(k)$ will be adorned in various ways in order to
distinguish them.  Hopefully this will not be too confusing.

\subsection{Lie subalgebras of $\fso(2)$}
\label{sec:so(2)}

First of all, it is clear that $\fso(2) = \RR\left<L_{12}\right>$ has
no proper subalgebras.

\subsection{Lie subalgebras of $\fso(3)$}
\label{sec:so(3)}

Next we consider $\fso(3) = \RR\left<L_{12}, L_{13}, L_{23}\right>$.
There is only one proper Lie subalgebra of $\fso(3)$ up to equivalence
and it is one-dimensional.  Indeed, $\fso(3)$ can be identified with
$\RR^3$ with the Lie bracket given by the vector cross product.  Hence
if a Lie subalgebra $\fh \subset \fso(3)$ has dimension greater than
$1$ it means that there are two linearly independent vectors $\bx$ and
$\by$ in $\fh$, but then their cross product $\bx \times \by$ is in
$\fh$ but is linearly independent from $\bx$ and $\by$, whence $\fh =
\fso(3)$.  We will choose the unique (up to equivalence) Lie
subalgebra of $\fso(3)$ to be $\fso(2)$, spanned by $L_{12}$.

\subsection{Lie subalgebras of $\fso(4)$}
\label{sec:so(4)}

Unlike $\fso(n)$ for all other $n\geq 3$, $\fso(4)$ is not simple: it
is isomorphic to two copies of $\fso(3)$, which we will call
$\fso(3)_\pm$ since they correspond to the $\pm1$ eigenspaces of the
Hodge star acting on $\Lambda^2\RR^4$ to which $\fso(4)$ is isomorphic
as a vector space and indeed as a representation.  More precisely, let
us define $L_i^\pm := \mp \half \left( L_{i4} \pm \half
  \varepsilon_{ijk}L_{jk}\right)$, for $i=1,2,3$, where
$\varepsilon_{123}=1$ in our conventions.  In other words,
\begin{equation}
  L_1^\pm = \mp \half (L_{14} \pm L_{23}) \qquad
  L_2^\pm = \mp \half (L_{24} \mp L_{13}) \qquad
  L_3^\pm = \mp \half (L_{34} \pm L_{12})~,
\end{equation}
which obey the following Lie brackets $[L_i^\pm, L_j^\pm] =
\varepsilon_{ijk} L_k^\pm$ and $[L_i^+, L_j^-] = 0$.

There are two inequivalent maximal subalgebras of $\fso(4)$: namely,
$\fso(3)_+ \oplus \fso(2)_-$, with generators
$\left(L_i^+,L_3^-\right)$, for $i=1,2,3$, and the diagonal subalgebra
of $\fso(3)_+ \oplus \fso(3)_-$, with generators $\left(L_i^+ +
  L_i^-\right)$, for $i=1,2,3$, which is thus precisely $\fso(3)$ as
defined above.  One might expect also a subalgebra $\fso(2)_+ \oplus
\fso(3)_-$, but this is related to $\fso(3)_+ \oplus \fso (2)-$ via an
automorphism of $\fso(4)$: namely, $L_i^\pm \mapsto L_i^\mp$.
Geometrically it corresponds to orientation reversal in $\RR^4$.  The
maximal subalgebras of $\fso(3)$ have been determined above, so it
remains to determine the maximal subalgebras of $\fso(3)_+ \oplus
\fso(2)_-$.

There are two inequivalent maximal subalgebras of $\fso(3)_+ \oplus
\fso(2)_-$: namely, $\fso(3)_+$, spanned by $\left(L_i^+\right)$ for
$i=1,2,3$ and $\fso(2)_+\oplus \fso(2)_-$, spanned by
$\left(L_3^+,L_3^-\right)$.  All proper subalgebras of
$\fso(2)_+\oplus \fso(2)_-$ are one-dimensional and hence maximal.
There is a pencil of such subalgebras, corresponding to the span of
$\alpha L_3^+ + \beta L_3^-$, for fixed $\alpha,\beta$, where the pair
$(\alpha,\beta)$ is defined up to multiplication by a nonzero real
number: that is, $(\alpha,\beta) \sim (\lambda\alpha,\lambda\beta)$
for some $\lambda\neq 0$.  Notice that the automorphism corresponding
to orientation reversal on $\RR^4$ exchanges $\alpha$ and $\beta$,
whence one must impose the condition $\alpha \geq \beta$, say, in
order not to over-count.  We can set $\alpha =1$ without loss of
generality and parametrise the subalgebras by a real number $\beta\in
[0,1]$.  Thus we let $\fso(2)_\beta$ denote the span of $L_3^+ + \beta
L_3^-$.  Notice that $\fso(2)_{\beta=0} = \fso(2)_+$ and
$\fso(2)_{\beta =1} = \fso(2)$, whence the need to impose
$0<\beta<1$.

At this moment we should point out a generic fact.  We are interested
in manifolds $G/H$, whence $H$ is a closed subgroup of $G$.  This
condition typically translates into the rationality of the parameters
defining the Lie subalgebra.  For example, the Lie subalgebra
$\fso(2)_\beta$ of $\fso(4)$ is the Lie algebra of a subgroup which is
dense in a torus if $\beta$ is irrational, hence for it to correspond
to a closed subgroup, we must impose that $\beta$ be rational.

Putting all this together we get the following Hasse diagram of
nontrivial subalgebras of $\fso(4)$ up to equivalence.  Following an
edge upwards denotes inclusion of a maximal subalgebra and subalgebras
at the same height have the same dimension, as indicated in the
left-hand column.
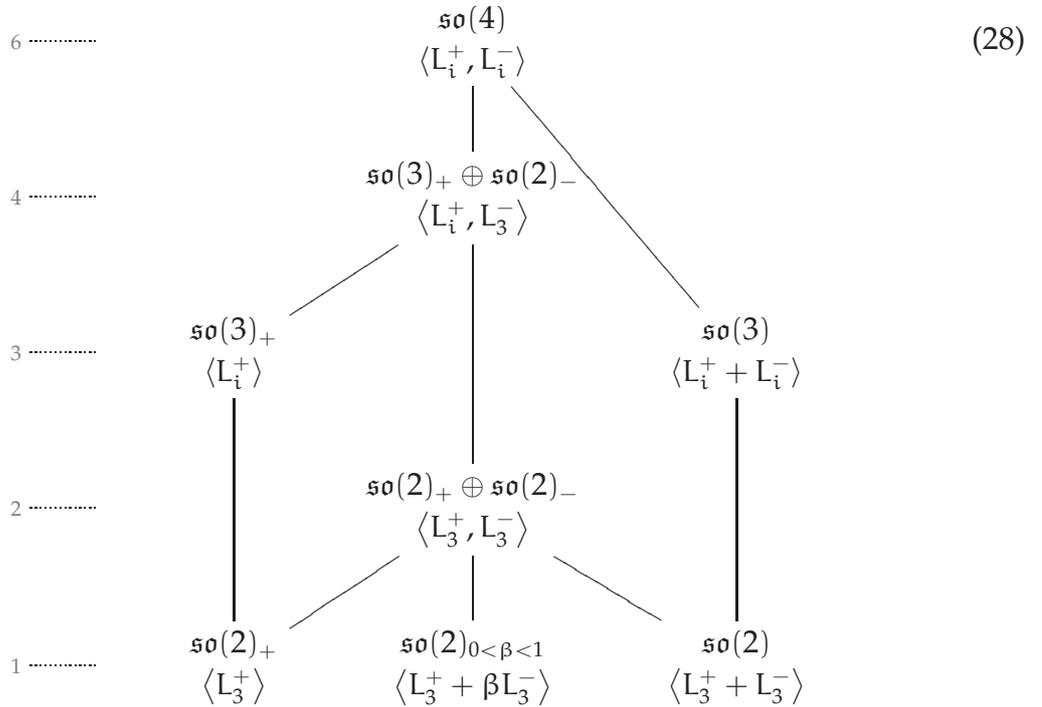
\begin{equation}
  \label{eq:so4}
  \xymatrix{
   \Hdim{6} &&& \subalg{\fso(4)}{L_i^+,L_i^-} \ar@{-}[d] \ar@{-}[ddr] & \\
   \Hdim{4} &&& \subalg{\fso(3)_+ \oplus \fso(2)_-}{L_i^+, L_3^-} \ar@{-}[dd] \ar@{-}[dl] & \\
   \Hdim{3} & & \subalg{\fso(3)_+}{L_i^+} \ar@{-}[dd] && \subalg{\fso(3)}{L_i^++L_i^-} \ar@{-}[dd]\\
   \Hdim{2} &&& \subalg{\fso(2)_+ \oplus \fso(2)_-}{L_3^+,L_3^-} \ar@{-}[dr] \ar@{-}[d] \ar@{-}[dl] & \\
   \Hdim{1} && \subalg{\fso(2)_+}{L_3^+} & \subalg{\fso(2)_{0<\beta<1}}{L_3^+ + \beta L_3^-} & \subalg{\fso(2)}{L_3^++L_3^-}
  }
\end{equation}

\subsection{Lie subalgebras of $\fso(5)$}
\label{sec:so(5)}

The Lie algebra $\fso(5)$ has three inequivalent maximal subalgebras.
Two of them decompose the $5$-dimensional real representation:
$\fso(4)$, which leaves invariant a line, and $\fso(3) \oplus
\fso(2)_{45}$, spanned by $\left(L_{12}, L_{13},L_{23},L_{45}\right)$.
The third maximal subalgebra, isomorphic to $\fso(3)$, acts
irreducibly both on the vector and spinor representations.  We denote
it $\fso(3)_{\text{irr}}$ and an explicit basis is given by
\begin{equation}
  \label{eq:so3irrinso5}
  \fso(3)_{\text{irr}} = \RR \left<L_{15}+ 2 L_{24}, \sqrt{3} L_{35} + L_{12} - L_{45}, \sqrt{3}L_{13}+ L_{14}+L_{25}\right>~.
\end{equation}
Any $\fso(2)$ subalgebra of $\fso(3)_{\text{irr}}$ leaves invariant
precisely a line in $\RR^5$.  This means that it is contained in the
maximal $\fso(4)$ subalgebra.  In fact, comparing characteristic
polynomials of the resulting linear transformations of $\RR^5$ shows
that it is equivalent to an $\fso(2)_{\beta = \frac13}$ subalgebra,
hence already included under the subalgebras of $\fso(4)$.  There are
two maximal subalgebras of $\fso(3)\oplus \fso(2)_{45}$.  One of them
is of course $\fso(3)$, whereas the other is equivalent to $\fso(2)_+
\oplus \fso(2)_-$.  This allows us to determine the Hasse diagram of
nontrivial subalgebras of $\fso(5)$ from that of $\fso(4)$.

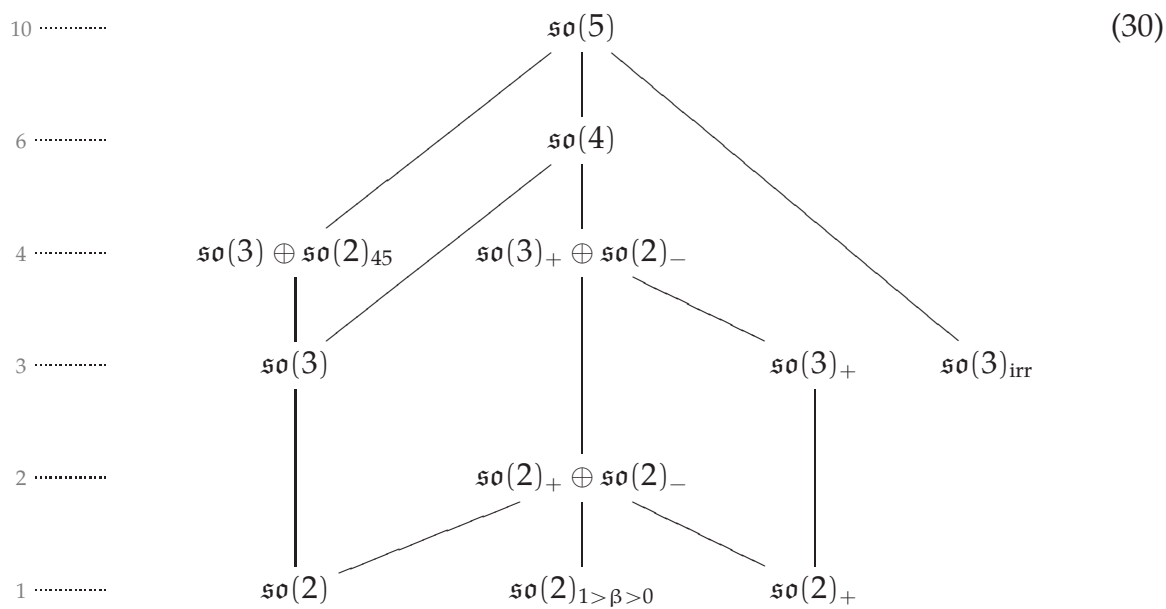
\begin{equation}
  \label{eq:so5}
  \xymatrix{
   \Hdim{10} &&  & \fso(5) \ar@{-}[dddrr] \ar@{-}[ddl]  \ar@{-}[d] & &\\
   \Hdim{6} &&   & \fso(4) \ar@{-}[d] \ar@{-}[ddl] & & \\
   \Hdim{4} && \fso(3)\oplus \fso(2)_{45} \ar@{-}[d]  & \fso(3)_+\oplus\fso(2)_- \ar@{-}[dd] \ar@{-}[dr]& & \\
   \Hdim{3} &&  \fso(3) \ar@{-}[dd] & & \fso(3)_+ \ar@{-}[dd] & \fso(3)_{\text{irr}}\\
   \Hdim{2} &&  & \fso(2)_+ \oplus \fso(2)_- \ar@{-}[dl] \ar@{-}[d] \ar@{-}[dr] & & \\
   \Hdim{1} &&  \fso(2) & \fso(2)_{1>\beta>0} & \fso(2)_+ & 
  }
\end{equation}

We have omitted some subalgebras of $\fso(3)\oplus\fso(2)_{45}$ and of
$\fso(3)_{\text{irr}}$ since as explained above, they are equivalent
to (albeit not the same as) subalgebras already included in the
diagram.

\subsection{Lie subalgebras of $\fso(6)$}
\label{sec:so(6)}

The Lie algebra $\fso(6)$ has four inequivalent maximal subalgebras.
Three of them decompose the $6$-dimensional representation: namely,
$\fso(5)$, $\fso(4)\oplus \fso(2)_{56}$, $\fso(3)\oplus
\fso(3)_{456}$; whereas one acts irreducibly on this representation:
namely, $\fu(3) \cong \fsu(3) \oplus \fu(1)$.  The top of the Hasse
diagram is given below.
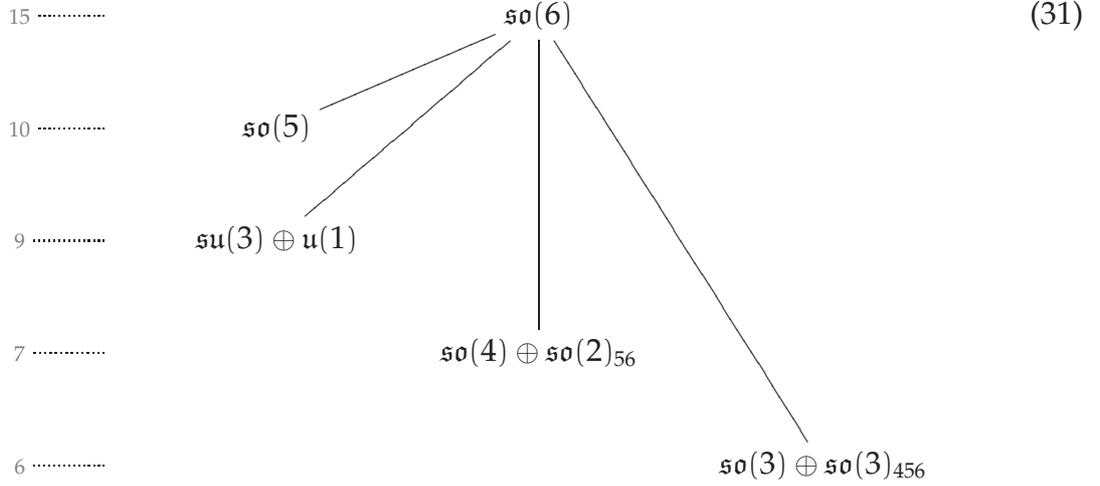
\begin{equation}
  \label{eq:so6}
  \xymatrix{
   \Hdim{15} & & & \fso(6) \ar@{-}[dl] \ar@{-}[ddl]  \ar@{-}[ddd] \ar@{-}[ddddr] & \\
   \Hdim{10} & & \fso(5) & & \\
   \Hdim{9} & & \fsu(3)\oplus \fu(1) & & \\
   \Hdim{7} & &  & \fso(4)\oplus\fso(2)_{56} &\\
   \Hdim{6} & & & & \fso(3)\oplus\fso(3)_{456}
  }
\end{equation}

The subalgebra $\fsu(3)\oplus \fu(1)$ can be described explicitly as
the centraliser of $L_{12} + L_{34} + L_{56}$, which spans the
$\fu(1)$ subalgebra of $\fsu(3)\oplus \fu(1)$.  A basis for the
$\fsu(3)$ subalgebra is given by
\begin{equation}
  \label{eq:su3inso6}
  \begin{split}
    &L_{13} + L_{24} \quad L_{14} - L_{23} \quad L_{15} + L_{26} \quad L_{16} - L_{25}\\
    &L_{35} + L_{46} \quad L_{36} - L_{45} \quad L_{12} - L_{34} \quad L_{34} - L_{56}~.
  \end{split}
\end{equation}

The Lie algebra $\fsu(3)$ has two inequivalent maximal subalgebras.
First we have a regular subalgebra isomorphic to $\fsu(2) \oplus
\fu(1)$.  Up to equivalence, we may choose it to lie inside $\fso(4)$
and corresponds to $\fso(3)_- \oplus \fso(2)_+$, which is itself
equivalent to $\fso(3)_+ \oplus \fso(2)_-$.  The second inequivalent
maximal subalgebra of $\fsu(3)$ is a singular subalgebra isomorphic to
$\fso(3)$ and denoted $\fso(3)_S$.  This subalgebra acts irreducibly
on the fundamental $3$-dimensional representation and in fact consists
of the real matrices in that representation.  It follows that any of
its proper subalgebras decomposes the fundamental representation of
$\fsu(3)$ and this is why it is already contained in the other maximal
subalgebra.  The corresponding Hasse diagram is given by
\begin{equation}
  \label{eq:su3}
  \xymatrix{
   \Hdim{8} & & & \fsu(3) \ar@{-}[dl] \ar@{-}[ddr] & \\
   \Hdim{4} & & \subalg{\fso(3)_- \oplus\fso(2)_+}{L_i^-,L_3^+}  \ar@{-}[dd] \ar@{-}[dr]& & \\
   \Hdim{3} & & & \subalg{\fso(3)_-}{L_i^-} \ar@{-}[dd] & \subalg{\fso(3)_S}{L_{14}-L_{23}, L_{15}+L_{26}, L_{36}-L_{45}} \\
   \Hdim{2} & & \subalg{\fso(2)_-\oplus\fso(2)_+}{L_3^-,L_3^+} \ar@{-}[dr] \ar@{-}[d] & & & \\
   \Hdim{1} & & \subalg{\fso(2) _{\beta \neq 0}}{L_3^+ + \beta L_3^-}  & \subalg{\fso(2)_+}{L_3^+} &
  }
\end{equation}
where we have omitted the $\fso(2)$ subalgebra of $\fso(3)_S$ since it
does not coincide with any of the $\fso(2)_\beta$ subalgebras, but
only equivalent to $\fso(2)_-$.  In summary, the only subalgebra of
$\fsu(3)$ which is not already (equivalent to) a subalgebra of
$\fso(4)$ is $\fso(3)_S$.

However it is $\fsu(3) \oplus \fu(1)$ which is the maximal subalgebra
of $\fso(6)$ and it behoves us to classify its subalgebras.  Goursat's
Lemma guarantees that such subalgebras are fibered products
$\fh_L\oplus_\fq \fh_R$, where $\fh_L < \fsu(3)$ and $\fh_R < \fu(1)$
and where $\dim\fq = 0$ or $1$.  In the former case, the subalgebras
are direct products, whereas in the latter they are graphs of nonzero
homomorphisms $\fh_L \to \fso(2)$, where $\fh_L < \fsu(3)$ is one of
the subalgebras admitting such homomorphisms.  A compact Lie algebra
admits a nonzero homomorphism to $\fso(2)$ if and only if it has
itself an $\fso(2)$ factor.  Of the subalgebras of $\fsu(3)$ with this
property, all are contained in $\fso(4)$ and hence they will be
counted among the subalgebras of $\fso(4)\oplus \fso(2)_{56}$.  The
reason is that if $\fh_L<\fso(4)$ then $\fh_L\oplus \fu(1)$ will be
equivalent to a subalgebra of $\fso(4)\oplus\fso(2)_{56}$.  Of the
direct product subalgebras all except for $\fso(3)_S$ itself and
$\fso(3)_S \oplus \fu(1)$ are subalgebras of $\fso(4)\oplus
\fso(2)_{56}$.  An explicit basis for $\fso(3)_S \oplus \fu(1)$ is
given by
\begin{equation}
  \fso(3)_S \oplus \fu(1) = \RR \left<L_{14}-L_{23}, L_{15}+L_{26}, L_{36}-L_{45},L_{12}+L_{34}+L_{56}\right>~.
\end{equation}

It thus remains to determine the subalgebras of $\fso(4)\oplus
\fso(2)_{56}$.  Goursat's Lemma says that they are products of
subalgebras of $\fso(4)$ and $\fso(2)_{56}$ fibered over some Lie
algebra $\fq$.  Since $\dim \fso(2)_{56} =1$, $\dim \fq \leq 1$ and we
have two cases to consider: $\dim \fq =0$, which corresponds to the
case of direct products of subalgebras, and $\dim \fq = 1$. In this
latter case, the map $\fso(2)_{56} \to \fq$ is an isomorphism, and
thus the subalgebras are graphs of nonzero homomorphisms $\fh_L \to
\fso(2)_{56}$, where $\fh_L < \fso(4)$ is a subalgebra admitting such
homomorphisms.  A quick glance at the Hasse diagram \eqref{eq:so4} for
$\fso(4)$ identifies such $\fh_L$ as one of
$\fso(3)_+\oplus\fso(2)_-$, $\fso(2)_+\oplus\fso(2)_-$ or
$\fso(2)_\beta$.  The resulting subalgebras of $\fso(4) \oplus
\fso(2)_{56}$ are explicitly given as follows:
\begin{itemize}
\item $\left(\fso(3)_+ \oplus \fso(2)_- \right)\oplus_{\fso(2)} \fso(2)_{56} = \RR \left<L_i^+, L_3^-+\alpha L_{56}\right>$, $\alpha \neq 0$;
\item $\left(\fso(2)_+ \oplus \fso(2)_- \right)\oplus_{\fso(2)} \fso(2)_{56} = \RR \left<L_3^+ + \beta L_{56}, L_3^-+\alpha L_{56}\right>$, $(\alpha,\beta) \in \RR^2$ not both zero; and
\item $\fso(2)_\beta \oplus_{\fso(2)} \fso(2)_{56} = \left<L_3^+ + \beta L_3^- + \alpha L_{56}\right>$, $\beta \in (0,1)$ and $\alpha \neq 0$.
\end{itemize}
Among the product subalgebras, those which are contained in $\fso(4)$
are already included inside $\fso(5)$, so we must consider those of
the form $\fh \oplus \fso(2)_{56}$, with $\fh < \fso(4)$, but only
those which are not contained inside $\fso(5)$; that is, those which
do not leave any nonzero vector invariant in $\RR^6$.  A quick glance
at the Hasse diagram \eqref{eq:so4} of subalgebras of $\fso(4)$
reveals that the following product subalgebras of $\fso(4)\oplus
\fso(2)_{56}$ have not appeared before:
\begin{itemize}
\item $\fso(3)_+\oplus \fso(2)_- \oplus \fso(2)_{56} = \RR\left<L_i^+,L_3^-,L_{56}\right>$;
\item $\fso(3)_+\oplus \fso(2)_{56} = \RR\left<L_i^+,L_{56}\right>$;
\item $\fso(2)_+\oplus \fso(2)_- \oplus \fso(2)_{56} = \RR\left<L_3^+,L_3^-,L_{56}\right>$; and
\item $\fso(2)_{0\leq \beta <1} \oplus \fso(2)_{56} = \RR\left<L_3^+ + \beta L_3^-,L_{56}\right>$.
\end{itemize}

Finally, we consider the maximal subalgebra $\fso(3)\oplus
\fso(3)_{456}$, which is isomorphic to $\fso(4)$, but embedded in a
different way in $\fso(6)$.  Being isomorphic to $\fso(4)$, its
subalgebras can be read (after some translation) from the Hasse
diagram \eqref{eq:so4} for $\fso(4)$.  It is not hard to see that all
subalgebras are already contained in at least one of the other maximal
subalgebras of $\fso(6)$.  Indeed, the Hasse diagram of subalgebras
for $\fso(3)\oplus \fso(3)_{456}$ is given by
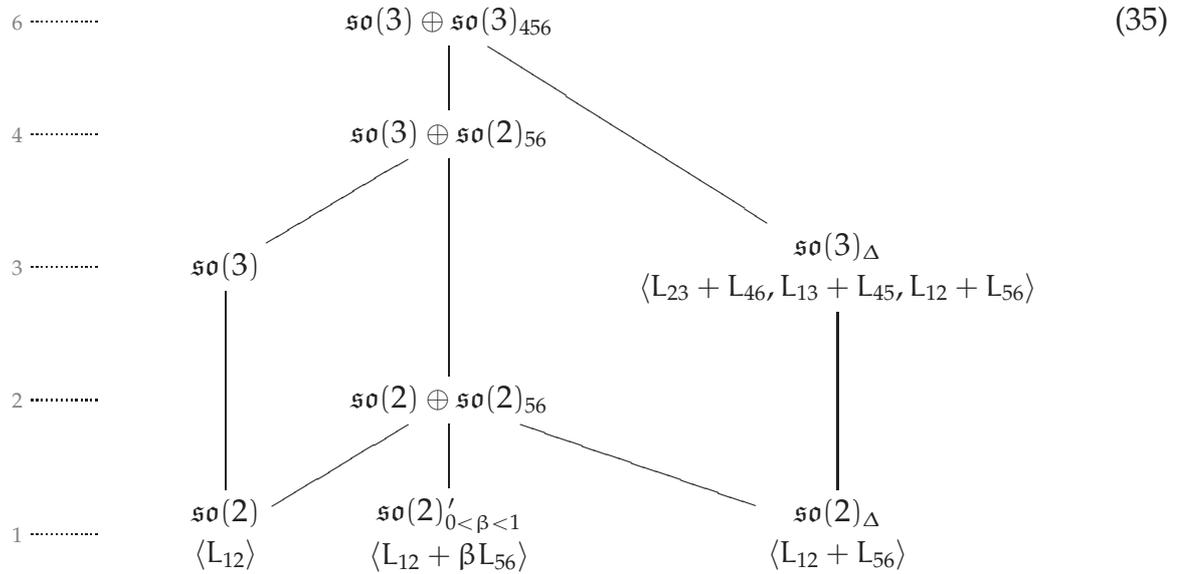
\begin{equation}
  \xymatrix{
    \Hdim{6} &&& \fso(3)\oplus\fso(3)_{456} \ar@{-}[d] \ar@{-}[ddr] & \\
    \Hdim{4} &&& \fso(3) \oplus \fso(2)_{56} \ar@{-}[dd] \ar@{-}[dl] & \\
    \Hdim{3} && \fso(3) \ar@{-}[dd] && \subalg{\fso(3)_\Delta}{L_{23}+L_{46},L_{13}+L_{45},L_{12}+L_{56}} \ar@{-}[dd]\\
    \Hdim{2} &&& \fso(2) \oplus \fso(2)_{56} \ar@{-}[dr] \ar@{-}[d] \ar@{-}[dl] & \\
    \Hdim{1} && \subalg{\fso(2)}{L_{12}} & \subalg{\fso(2)'_{0<\beta<1}}{L_{12} + \beta L_{56}} & \subalg{\fso(2)_\Delta}{L_{12}+L_{56}}
  }
\end{equation}
Hence we see that all but $\fso(3)_\Delta$ are contained in
$\fso(4)\oplus \fso(2)_{56}$, whereas it is not hard to see that
$\fso(3)_\Delta$ preserves a symplectic structure in $\RR^6$ and hence
it is contained in a $\fu(3)$ subalgebra of $\fso(6)$.  In fact, it
equivalent to the singular subalgebra $\fso(3)_S$ of $\fsu(3)$.

In summary, a proper Lie subalgebra of $\fso(6)$ is one of the
following subalgebras, which have been described explicitly above:
\begin{itemize}
\item $\fso(5)$ or one of its subalgebras, described in diagram
  \eqref{eq:so5},
\item $\fso(3)\oplus \fso(3)_{456} =
  \RR\left<L_{12},L_{13},L_{23},L_{45},L_{46},L_{56}\right>$,
\item $\fsu(3)\oplus \fu(1)$ or one of the subalgebras:
  \begin{itemize}
  \item $\fso(3)_S \oplus \fu(1)$, or
  \item $\fso(3)_S$,
  \end{itemize}
\item $\fso(4)\oplus \fso(2)_{56}$ or one of the subalgebras:
  \begin{itemize}
  \item $\fso(3)_+\oplus \fso(2)_- \oplus \fso(2)_{56}$,
  \item $\fso(3)_+\oplus \fso(2)_{56}$,
  \item $\fso(2)_+\oplus \fso(2)_- \oplus \fso(2)_{56}$,
  \item $\fso(2)_{0\leq \beta <1} \oplus \fso(2)_{56}$,
  \item $\left(\fso(3)_+ \oplus \fso(2)_- \right)\oplus_{\fso(2)} \fso(2)_{56}$,
  \item $\left(\fso(2)_+ \oplus \fso(2)_- \right)\oplus_{\fso(2)} \fso(2)_{56}$, or
  \item $\fso(2)_\beta \oplus_{\fso(2)} \fso(2)_{56}$.
 \end{itemize}
\end{itemize}

It is satisfying to find among these subalgebras precisely the four
inequivalent $\fso(3)$ subalgebras of $\fso(6)$: $\fso(3)$ and
$\fso(3)_+$ inside $\fso(4)$, $\fso(3)_{\text{irr}}$ inside $\fso(5)$
and $\fso(3)_S$ inside $\fu(3)$.

\subsection{Lie subalgebras of $\fso(7)$}
\label{sec:so(7)}

The Lie algebra $\fso(7)$ too has four inequivalent maximal
subalgebras.  Three of them decompose the $7$-dimensional
representation: namely, $\fso(6)$, $\fso(5)\oplus \fso(2)_{67}$ and
$\fso(4)\oplus \fso(3)_{567}$; whereas one acts irreducibly: namely,
$\fg_2$.  The Lie algebra $\fg_2$ has three inequivalent maximal
subalgebras: $\fsu(3)$ and $\fsu(2)\oplus \fsu(2)$, which decompose
the 7-dimensional irreducible representation, and one acting
irreducibly there: namely, $\fsu(2)_{\text{irr}}$.  This yields the
following subdiagram of the Hasse diagram of subalgebras of
$\fso(7)$.
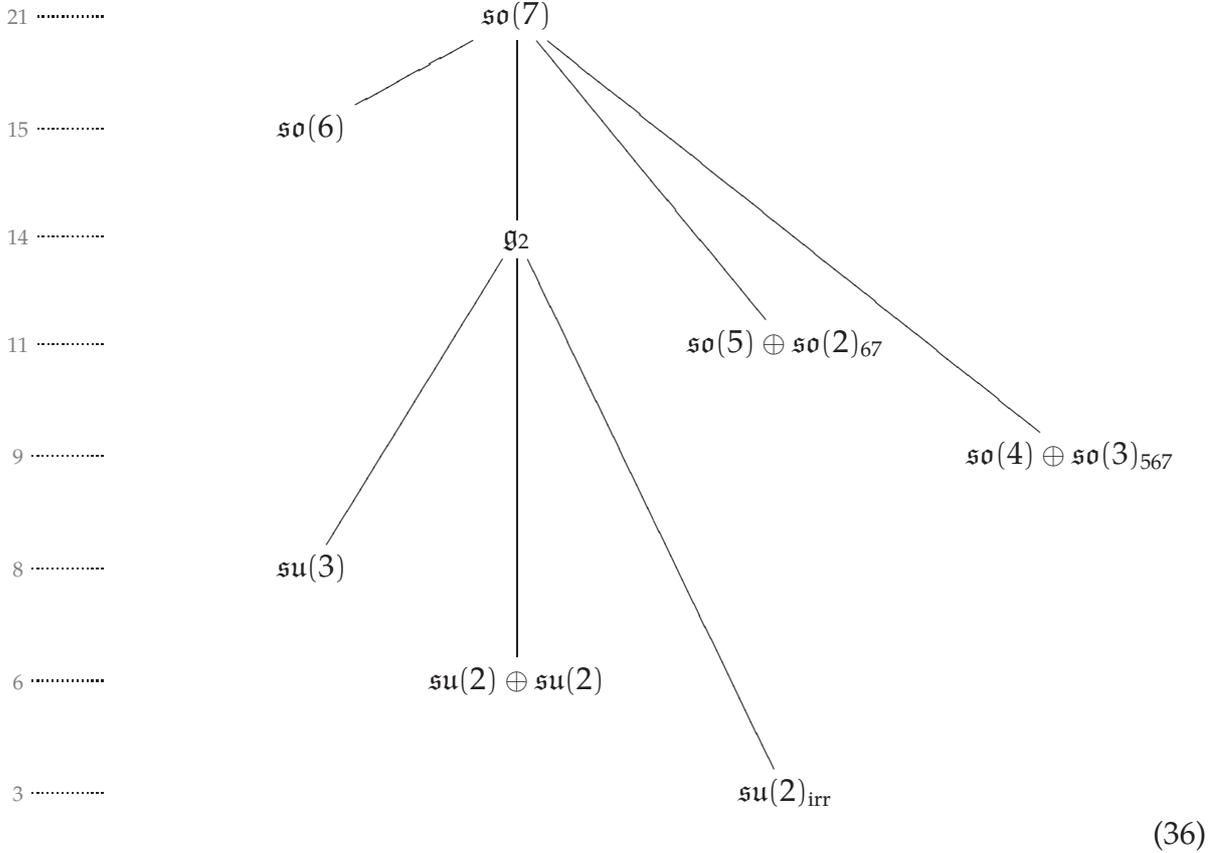
\begin{equation}
  \label{eq:so7}  
  \xymatrix{
   \Hdim{21} && & & \fso(7) \ar@{-}[dd] \ar@{-}[dddr] \ar@{-}[ddddrr] \ar@{-}[dl] & &\\
   \Hdim{15} && & \fso(6) & & &\\
   \Hdim{14} && & & \fg_2 \ar@{-}[dddl] \ar@{-}[dddd]  \ar@{-}[dddddr]& & \\
   \Hdim{11} && & & & \fso(5)\oplus\fso(2)_{67} &\\
   \Hdim{9} && & & & & \fso(4)\oplus\fso(3)_{567}\\
   \Hdim{8} && & \fsu(3) & & &\\
   \Hdim{6} && & & \fsu(2)\oplus\fsu(2) & & \\
   \Hdim{3} && & & & \fsu(2)_{\text{irr}} &
  }
\end{equation}
Although in order to fully specify the Hasse diagram for $\fso(7)$ we
would have to determine the subalgebras of $\fso(5)\oplus \fso(2)$ and
$\fso(4)\oplus \fso(3)$, and as tempting as that is, it is also
unnecessary for what follows.  We record here an explicit basis for
the $\fg_2$ subalgebra of $\fso(7)$:
\begin{equation}
  \label{eq:g2inso7}
  \begin{split}
    &L_{14} - L_{23} \quad  L_{13} + L_{24} \quad  L_{17} + L_{26} \quad  L_{16} - L_{27} \quad  L_{12} - L_{34} \quad  L_{17} - L_{35} \quad  L_{25} - L_{36} \\
    &L_{15} + L_{37} \quad  L_{16} + L_{45} \quad  L_{15} - L_{46} \quad  L_{25} - L_{47} \quad  L_{14} + L_{56} \quad  L_{13} - L_{57} \quad  L_{12} + L_{67} ~.
 \end{split}
\end{equation}

\section{The supergravity field equations for homogeneous backgrounds}
\label{sec:supergr-field-equat}

The above results allow us in principle to determine all
eleven-dimensional lorentzian homogeneous spaces with a transitive
action of a group $G$ locally isomorphic to $\SO(n) \times \SO(3,2)$.
For each such lorentzian manifold, we wish to investigate whether
there are any solutions to the supergravity field equations.  The
field equations are partial differential equations but they become
algebraic in a homogeneous Ansatz, by which we mean that the 4-form is
also $G$-invariant.  In this section we will write down the field
equations in a homogeneous Ansatz.

\subsection{The field equations of eleven-dimensional supergravity}
\label{sec:field-equat-elev}

Following the conventions of \cite{DS2brane}, the bosonic part of the
action of $d=11$ supergravity is (setting Newton's constant to $1$)
\begin{equation}
  \label{eq:lag}
  \int_M \left( \half R \dvol - \tfrac14 F\wedge\star F + \tfrac1{12} F \wedge F \wedge A \right)~,
\end{equation}
where $F=dA$ locally, $R$ is the scalar curvature of $g$ and $\dvol$
is the (signed) volume element
\begin{equation}
  \dvol := \sqrt{|g|}\, dx^0 \wedge dx^1 \wedge \cdots \wedge
  dx^{10}~.
\end{equation}
The Euler-Lagrange equations following from \eqref{eq:lag} are
\begin{equation}
  \label{eq:EL}
  \begin{aligned}
    d \star F &= \half F\wedge F\\
    \Ric(X,Y) &= \half \left< \iota_X F, \iota_Y F \right> - \tfrac16 g(X,Y) |F|^2~,
  \end{aligned}
\end{equation}
for all vector fields $X,Y$ on $M$.  In this equation we have
introduced the inner product $\left<-,-\right>$ on differential forms,
defined by
\begin{equation}
   \left<\theta,\omega\right> \, \dvol = \theta \wedge \star \omega~,
\end{equation}
and the associated norm
\begin{equation}
  |\theta|^2 = \left<\theta, \theta\right>~,
\end{equation}
which in a lorentzian manifold is \emph{not} positive-definite.

The field equations \eqref{eq:EL} are invariant under the homothetic
action of $\RR^+$: $(g,F) \mapsto (e^{2t}g, e^{3t}F)$, where $t\in
\RR$.  Indeed, under $g \mapsto e^{2t} g$, the Levi-Civita connection,
consisting of terms of the form $g^{-1}dg$, does not change.  This
means that the $(3,1)$ Riemann curvature tensor is similarly
invariant, and so is any contraction such as the Ricci tensor.  Under
$F \mapsto e^{3t}F$, the tensor in the right-hand side of the Einstein
equation is similarly invariant, since the $e^{6t}$ coming from the
two $F$s cancels the $e^{-6t}$ coming from the three $g^{-1}$s.  On
the other hand, the Bianchi identity $dF=0$ is clearly invariant under
homotheties and the Maxwell-like equation is as well.  Indeed, using
that the Hodge $\star$ acting on $p$-forms in a $D$-dimensional
manifold, scales like $e^{(D-2p)t}$ under $g \mapsto e^{2t}g$, we see
that $\star$ acting on $4$-forms in $11$-dimensions scales like
$e^{3t}$, just like $F$, whence both sides of the Maxwell-like
equation scale in the same way: namely, $e^{6t}$.  This means that the
moduli spaces of solutions of the field equations are always cones.
It is possible to extend this to a homothetic action of $\RR^\times$
(the nonzero real numbers) if we take the point of view that the
vielbeins scale by $\lambda \neq 0$, whence if $\lambda < 0$ the
orientation changes.  The particular homothety where $\lambda = -1$,
which is just orientation reversal, is known as ``skew-whiffing'' in
the early supergravity literature, as described for example in
\cite{DNP}.

\subsection{The equivalent algebraic equations}
\label{sec:equiv-algebr-equat}

Let us assume that we are looking for homogeneous supergravity
backgrounds.  This means that the spacetime is a homogeneous
eleven-dimensional lorentzian manifold $G/H$ and that the 4-form $F$
is $G$-invariant.  Algebraically, such a background is determined by a
split $\fg = \fh \oplus \fm$ of the Lie algebra of $G$ into the Lie
algebra of $H$ and a complement $\fm$.  As explained in Section
\ref{sec:basic-notions-about}, for $G$ semisimple we may restrict
ourselves to the case where $\fg = \fh \oplus \fm$ is a reductive
split.

Let us introduce bases $Y_i$ for $\fm$ and $X_a$ for $\fh$, relative
to which the Lie brackets are given by
\begin{equation}
  [X_a, X_b] =  f_{ab}{}^c X_c \qquad [X_a,Y_i]= f_{ai}{}^j Y_j \qquad\text{and}\qquad [Y_i,Y_j] = f_{ij}{}^k Y_k + f_{ij}{}^a X_a~.
\end{equation}

The metric is given by a lorentzian inner product, denoted $g$, on
$\fm$ which is invariant under the linear isotropy representation of
$\fh$ and with components $\left<Y_i,Y_j\right> = g_{ij}$ relative to
the chosen basis.  The 4-form is given by an element $F \in
\Lambda^4\fm^*$ which is similarly invariant and has components
$F(Y_i,Y_j,Y_k,Y_l)=F_{ijkl}$.

The data $(\fg,\fh,\fm,g,F)$ defines a homogeneous background of
eleven-dimensional supergravity if and only if the following equations
are satisfied:
\begin{itemize}
\item the Bianchi identity $dF=0$, which relative to the basis is
  given explicitly by setting expression \eqref{eq:dFcomps} to zero;
\item the nonlinear Maxwell equation
  \begin{equation}
    \label{eq:maxwell}
    \delta F = -\star \half F \wedge F~;
  \end{equation}
\item and the Einstein equation
  \begin{equation}
    \label{eq:einstein}
    R_{ij} = \tfrac1{12} F_{iklm} F_{j}{}^{klm} - \tfrac1{144} g_{ij} F_{klmn}F^{klmn}~,
  \end{equation}
where $R_{ij}$ is given by equation \eqref{eq:riccicomps}
\end{itemize}

\subsection{The methodology}
\label{sec:methodology}

Let us now explain the method by which we search for homogeneous
backgrounds.  Having chosen $\fg = \fh \oplus \fm$ we first determine
whether there is an $\fh$-invariant lorentzian inner product on $\fm$.
As mentioned in Section \ref{sec:basic-notions-about}, for the case
when $\fh$ is compact, this will be the case if and only if $\fh$
leaves invariant some nonzero vector in $\fm$; in other words, if
$\fm^\fh \neq 0$, where $\fm^\fh$ denotes the subspace of $\fm$ which
is fixed pointwise by $\fh$.  If $\fm$ admits an $\fh$-invariant
lorentzian inner product we say that $(\fg,\fh,\fm)$ is
\emph{admissible}.

Let $(\fg,\fh,\fm)$ be admissible.  Then next step is to determine the
(nontrivial) vector space $(S^2\fm^*)^\fh$ of $\fh$-invariant
symmetric bilinear forms and the subset consisting of invariant
lorentzian inner products.  This subset will be an open subset of
$(S^2\fm^*)^\fh$ and will thus be parametrised by $\dim
(S^2\fm^*)^\fh$ parameters subject to some inequalities to ensure that
the symmetric bilinear form is nondegenerate and has lorentzian
signature.  Clearly, it is a cone, since rescaling a lorentzian inner
product by a positive real number yields another lorentzian inner
product.   Let $\{\gamma_\alpha\}$ denote the parameters associated to
the inner product.  Similarly we determine the vector space
$(\Lambda^4\fm^*)^\fh$ of $\fh$-invariant 4-forms on $\fm$ and the
subspace consisting of closed 4-forms; namely, those obeying equation
\eqref{eq:dFcomps}.  Choosing a basis for the closed invariant
4-forms, we can specify every such form by some parameters
$\{\varphi_\alpha\}$.  The Maxwell and Einstein equations then give a
set of algebraic equations for the parameters $\gamma_\alpha$ and
$\varphi_\alpha$ which we must solve.  (They are in fact polynomial in
$\varphi_\alpha$ and in $\sqrt{\gamma_\alpha}$.)

Two small simplifications can be made to reduce the number of free
parameters.  First of all, the homothety invariance of the equations
allows us to eliminate one of the $\gamma_\alpha$: if one of the
$\gamma_\alpha$ is known to be different from zero, then we can assume
that it has magnitude $1$ via a homothety.  Typically we will choose
the $\gamma_\alpha$ corresponding to the timelike direction and set it
equal to $-1$.

The second simplification is a little more subtle and consists of
exploiting the normaliser of $\fh$ in $\fg$ in order to eliminate one
or more of the $\gamma_\alpha$ parameters.  Let $\fn$ denote the
normaliser of $\fh$ in $\fg$: that is, the largest subalgebra of $\fg$
which contains $\fh$ as an ideal.  More formally, we say that
\begin{equation}
  X \in \fn \iff [X,Y] \in \fh\quad \forall Y \in \fh~.
\end{equation}
Let $X \in \fn$.  Since $\fg = \fh \oplus \fm$, we may decompose $X =
X_\fh + X_\fm$ uniquely, where $X_\fh \in \fh$ and $X_\fm \in \fm$.
Since $X \in \fn$, it obeys $[X,Y] \in \fh$ for all $Y\in \fh$, or
equivalently $[X_\fh + X_\fm, Y] \in \fh$ for all $Y\in\fh$.  Since
$\fh$ is a subalgebra, $[X_\fh, Y] \in \fh$ and hence the only
condition rests on $X_\fm$: $[X_\fm, Y] \in \fh$ for all $Y \in \fh$.
However since the split is reductive, $[X_\fm, Y] \in \fm$ for all $Y
\in \fh$ and hence it must happen that $[X_\fm, Y] = 0$ for all $Y \in
\fh$; in other words, $X_\fm \in \fm^\fh$.  That is to say, the
normaliser of $\fh$ in $\fg$ is given by $\fn = \fh \oplus \fm^\fh$.
Let $N$ be the normaliser of $H$ in $G$, so that
\begin{equation}
  x \in N \iff  x y x^{-1} \in H \quad \forall y \in H~.
\end{equation}
Then $N$ is a subgroup of $G$ with Lie algebra $\fn$.  It is
convenient to define the abstract group $W = N/H$, which is a group
because $H$ is normal in $N$ by definition.  The Lie algebra of $W$ is
precisely $\fm^\fh$.  Indeed, suppose that $X = X_\fh + X_\fm \in \fg$
belongs to the normaliser of $\fh$ in $\fg$.  Then for all $Y \in
\fh$, $[X,Y] \in \fh$.  This is equivalent to $[X_\fm, Y] \in \fh$ for
all $Y \in \fh$, but since the split is reductive, $[X_\fm, Y] = 0$
for all $Y \in \fh$, whence $X_\fm \in \fm^\fh$.  In other words, the
Lie algebra of the normaliser of $H$ in $G$ is $\fh \oplus \fm^\fh$,
from where the claim follows.

We saw above that in the case where $H$ is compact, $\fm^\fh$ is
nonzero if $G/H$ is to admit a homogeneous lorentzian metric, whence
in that case $W$ is a Lie group of dimension at least one.

It turns out that $W$ may be used to reduce the number of parameters
defining the lorentzian metrics in $G/H$.

The idea is the following.  Let $o \in G/H$ be the origin; that is,
any point with stability subgroup $H$; that is,
\begin{equation}
  x \in H \iff x \cdot o = o~.
\end{equation}
Let $x \in N$ and consider the point $o' = x \cdot o$.  We claim that
$o'$ also has stability subgroup $H$.  Indeed,
\begin{equation}
  y \cdot o' = o'  \iff yx \cdot o = x \cdot o \iff x^{-1} y x \cdot o = o \iff x^{-1} y x \in H \iff y \in x H x^{-1} = H~.
\end{equation}
Now suppose that $\Theta$ is a $G$-invariant tensor field on $G/H$.
As explained in Section \ref{sec:basic-notions-about}, $\Theta$ is
determined uniquely by its value $\Theta_o$ at $o$ (or indeed at any
point).  Now $\Theta_o$ is a tensor in $\fm$ invariant under the
linear isotropy representation of $\fh$.  Now consider the value of
$\Theta$ at the point $o'$ defined above.  Since $\Theta$ is a
$G$-invariant tensor, $\Theta_{o'} = \Theta_{x \cdot o} = x \cdot
\Theta_o $, which is again an $\fh$-invariant tensor in $\fm$, since
$o'$ has stability subgroup $H$.  In other words, the group $N$ acts
on the space of $\fh$-invariant tensors in $\fm$.  In fact, since the
subgroup $H$ (assumed connected) of $N$ acts trivially, what we have
is actually an action of $W = N/H$ on the $\fh$-invariant tensors.  It
is this action which we can use to bring the invariant tensor to a
simpler form.  The idea is that at another point $o'$ with the same
stabiliser, our tensor will take a simpler form and we could have been
working at that point from the start.

We will now proceed to systematically explore the possible homogeneous
backgrounds of $G = \SO(n) \times \SO(3,2)$, for $n>4$.  We will only
consider the case where $G$ acts effectively; although one could also
study admissible $G/H = \SO(n)/K$, where $K$ is a closed subgroup of
$\SO(n)$, so that $H = \SO(3,2)\times K$.  By dimension these only
exist for $n\geq 6$.  One can rule out the existence of such
backgrounds for $n=7$, but we have not completed the analysis of the
$n=6$ backgrounds.  This is of questionable interest, though, since
(the dual of) a conformal field theory should have a nontrivial action
of the conformal group.

\section{Homogeneous non-AdS backgrounds}
\label{sec:non-AdS-backgrounds}

We now systematically explore the possible eleven-dimensional
homogeneous spaces with infinitesimal data $(\fg,\fh)$ with $\fg =
\fso(n) \oplus \fso(3,2)$, with $n=4,5,6,7$, and $\fh < \fg$ the Lie
algebra of a compact subgroup; that is $\fh < \fso(n) \oplus \fso(3)
\oplus \fso(2)$.

\subsection{Still no $n=7$ duals}
\label{sec:possible-new-n=7}

Here $\fg = \fso(7) \oplus \fso(3,2)$ has dimension $31$, whence we
are looking for subalgebras $\fh$ of dimension $20$.  There are
however none.  Indeed, by Goursat's Lemma, every such subalgebra is
given by Lie subalgebras $\fh_L < \fso(7)$ and $\fh_R < \fso(3) \oplus
\fso(2)$ fibered over a common quotient $\fq$.  By the dimension
formula \eqref{eq:dimgoursat}, we have that
\begin{equation}
  \dim \fh = \dim \fh_L + \dim \fh_R - \dim \fq \leq \dim \fh_L + \dim \fh_R \leq \dim\fh_L + 4~,
\end{equation}
but also $\dim\fh \geq \dim\fh_L$ from formula \eqref{eq:dimh}.  Since
$\dim\fso(7)=21$, we have to take a proper subalgebra $\fh_L <
\fso(7)$.  It follows from the Hasse diagram \eqref{eq:so7} of maximal
subalgebras of $\fso(7)$, that $\dim \fh_L \leq 15$, whence from the
first of the above inequalities $\dim \fh \leq 19$.

\subsection{No new $n=6$ duals}
\label{sec:possible-new-n=6}

Here $\fg = \fso(6) \oplus \fso(3,2)$ has dimension $25$, whence we
are looking for subalgebras $\fh$ of dimension $14$.  By Goursat's
Lemma, $\fh$ is given by subalgebras $\fh_L < \fso(6)$ and $\fh_R <
\fso(3) \oplus \fso(2)$ fibered over a common quotient $\fq$.  The
dimension formula \eqref{eq:dimgoursat} says that
\begin{equation}
  \dim \fh \leq \dim \fh_L + \dim \fh_R~,
\end{equation}
but as before we cannot take $\fh_L = \fso(6)$ since $\dim \fso(6) =
15 > \dim \fh$, violating equation \eqref{eq:dimh}.  So we have to
take a proper subalgebra $\fh_L < \fso(6)$.  From the Hasse diagram
\eqref{eq:so6} we see that the largest dimension of a proper
subalgebra is $10$, corresponding to $\fso(5)$.  By the above
inequality, this is also the smallest dimension we could take, hence
there is precisely one such subalgebra, with $\fq =0$ and hence a
direct product: $\fh = \fso(5) \oplus \fso(3) \oplus \fso(2)$.  Being
a product, the geometry is also a product, and we have a homogeneous
space locally isometric to $\SO(6)/\SO(5) \times
(\SO(3,2)/\SO(3)\times\SO(2))$.  However this homogeneous space does
not admit an invariant lorentzian metric.  Indeed, in the first factor
$\SO(6)/\SO(5)$ the linear isotropy representation is irreducible and
in fact $\SO(6)/\SO(5)$ is locally isometric to the round $S^5$.  As
for the second factor, $\fso(3,2) = \fso(3) \oplus \fso(2) \oplus
\fm$, where $\fm = \textbf{3}\otimes\textbf{2}$ is the tensor product
of the fundamental vectorial representations of $\fso(3)$ and
$\fso(2)$.  Since there is no invariant line, there is no $\fso(3)
\oplus \fso(2)$-invariant lorentzian inner product on $\fm$.

\subsection{Possible new $n=5$ duals}
\label{sec:possible-new-n=5}

Here $\fg = \fso(5) \oplus \fso(3,2)$ has dimension $20$, whence we
are looking for subalgebras $\fh$ of dimension $9$.  From Goursat's
Lemma, such a subalgebra will be given by two subalgebras $\fh_L
\subset \fso(5)$ and $\fh_R \subset \fso(3) \oplus \fso(2)$ (the
maximal compact subalgebra of $\fso(3,2)$) fibered over a common
quotient $\fq$.  Again we have to take a proper subalgebra $\fh_L <
\fso(5)$, since $\dim \fso(5) > 9$.  The Hasse diagram \eqref{eq:so5}
of subalgebras of $\fso(5)$ identifies precisely one such possible
$\fh_L$ which obeys the inequality $9 \leq \dim\fh_L + \dim\fh_R \leq
\dim\fh_L + 4$: namely, $\fh_L= \fso(4)$.  This means that $\fq=0$
since although $\fso(4)$ is not simple, the smallest nonzero quotient
has dimension $3$ and that results in $\fh$ not of enough dimension.
This in turn forces $\dim\fh_R=3$, whence $\fh_R$ is isomorphic to an
$\fso(3)$ subalgebra of $\fso(3,2)$.  Therefore, up to equivalence,
there is precisely one choice for $\fh$: namely, $\fso(4)\oplus
\fso(3)$.  The geometry will also therefore be locally isometric to a
product: $\SO(5)/\SO(4) \times \SO(3,2)/\SO(3)$.  The first factor is
locally isometric to the round $S^4$ and the second factor now does
possess an invariant lorentzian metric.  Indeed, $\fso(3,2) = \fso(3)
\oplus \fp$, where $\fp \cong \RR^3 \oplus \RR^3 \oplus \RR$
decomposes under the linear isotropy representation as two copies of
the three-dimensional vector representation of $\fso(3)$ and a
one-dimensional trivial representation.

Let $L_{ab},L_{a5}$ denote the standard generators of $\fso(5)$, where
$a,b=1,2,3,4$ and let $J_{ij},J_{iA},J_{45}$ denote the standard
generators of $\fso(3,2)$, where $i,j=1,2,3$ and $A=4,5$.  Then $\fh$
is spanned by $L_{ab},J_{ij}$ and $\fm$ by $L_{a5},J_{iA},J_{45}$.
The $L_{a5}$ transform as a vector of
$\fso(4)=\RR\left<L_{ab}\right>$, whereas $J_{iA}$ transform as two
copies of the vector representation of
$\fso(3)=\RR\left<J_{ij}\right>$.  The index $A$ is a vector of the
$\fso(2)$ with generator $J_{45}$ which is the nontrivial part of the
normaliser of $\fh$ in $\fg$.  There is a 5-parameter family of
invariant lorentzian inner products on $\fm$:
\begin{equation}
  \left<J_{45},J_{45}\right> = \gamma_0~,\qquad \left<L_{a5},L_{b5}\right> = \gamma_1 \delta_{ab}~,\qquad \left<J_{iA},J_{jB}\right> = \delta_{ij} \Omega_{AB}~,
\end{equation}
where $\gamma_0<0$, $\gamma_1>0$ and $\Omega_{AB}$ is a
positive-definite symmetric $2\times 2$ matrix.  The $\SO(2)$ subgroup
generated by $J_{45}$ acts by rotating the basis $J_{iA}$.  Let
$R_\vartheta \in \SO(2)$ denote the rotation by an angle $\vartheta$.
Then the matrix $\Omega$ transforms as $\Omega \mapsto R_\vartheta^T
\Omega R_\vartheta$.  The off-diagonal component $\Omega_{12}$
transforms as
\begin{equation}
  \Omega_{12} \mapsto \half(\Omega_{11}-\Omega_{22})\sin2\vartheta + \Omega_{12} \cos2\vartheta~.
\end{equation}
If $\Omega_{12} \neq 0$, simply let $\vartheta \in (0,\pi/2)$ be given by
\begin{equation}
  \vartheta = \half \cot^{-1}\left(\frac{\Omega_{22}-\Omega_{11}}{2\Omega_{12}}\right).
\end{equation}
With this choice, the transformed $\Omega$ is diagonal.  Therefore,
without loss of generality, we can assume that
$\left<J_{i4},J_{j5}\right> =0$ and that
\begin{equation}
  \left<J_{i4},J_{j4}\right> = \gamma_2 \delta_{ij} \qquad\text{and}\qquad   \left<J_{i5},J_{j5}\right> = \gamma_3 \delta_{ij}~,
\end{equation}
where $\gamma_2$ and $\gamma_3$ are positive.  Furthermore, using the
homothety invariance of the equations of motion, we can set
$\gamma_0=-1$ without loss of generality.  This then leaves three
positive parameters $\gamma_{1,2,3}$ for the metric.

In order to compute the curvature, we need to compute the $U$ tensor.
Since $S^4$ is a symmetric space, the $U$ tensor has no legs along the
subspace spanned by the $L_{a5}$.  It is then not too hard to show,
using equation \eqref{eq:Utensor}, that all other components vanish
except for the following:
\begin{equation}
  \label{eq:utensorn=5}
  U(J_{i4},J_{j5}) = \half (\gamma_3 - \gamma_2) \delta_{ij} J_{45}~,\quad
  U(J_{45}, J_{i4}) = \frac{1 - \gamma_2}{2\gamma_3} J_{i5} \quad\text{and}\quad
  U(J_{45}, J_{i5}) = \frac{\gamma_3-1}{2\gamma_2} J_{i4}~.
\end{equation}

Defining $\Lambda = 1-(\gamma_2+\gamma_3)^2$, one computes the
following nonzero components of the Ricci tensor:
\begin{equation}
\label{eq:riccin=5}
  \begin{aligned}[b]
    \Ric(J_{45},J_{45}) &= 6 + \frac{3\Lambda}{2\gamma_2\gamma_3}\\
    \Ric(L_{a5},L_{b5}) &= 3 \delta_{ab}
  \end{aligned}
  \qquad\qquad
  \begin{aligned}[b]
    \Ric(J_{i4},J_{j4}) &= \left( (\gamma_2 + \gamma_3) + \frac{\Lambda}{2\gamma_3} - 3\right) \delta_{ij}\\
    \Ric(J_{i5},J_{j5}) &= \left( (\gamma_2 + \gamma_3) + \frac{\Lambda}{2\gamma_2} - 3\right) \delta_{ij}~,
  \end{aligned}
\end{equation}
whence the Ricci scalar becomes
\begin{equation}
  R = \left(6 + \frac{3\Lambda}{2\gamma_2\gamma_3} \right)(\gamma_2 + \gamma_3 - 1) - 6~.
\end{equation}

The space of invariant 4-forms is six-dimensional.  A possible basis
is given by the following 4-forms.  Firstly, we have the volume form
on the $S^4$, which is given algebraically by
\begin{equation}
  L_{15}^* \wedge L_{25}^* \wedge L_{35}^* \wedge L_{45}^*~.
\end{equation}
We then have an invariant 2-form
\begin{equation}
  \omega = \sum_{i=1}^3 J_{i4}^* \wedge J_{i5}^*
\end{equation}
and squaring it we get an invariant 4-form.  Finally we have
\begin{equation}
  \sum_{i,j,k=1}^3\sum_{A,B,C=4}^5 t_{ABC} \varepsilon_{ijk} J_{iA}^* \wedge J_{jB}^* \wedge J_{kC}^* \wedge J_{45}^*~,
\end{equation}
where $t_{ABC}$ is a symmetric 3-tensor, whence it has four
components.  It turns out that all invariant 4-forms are already
closed, so the Bianchi identity is identically satisfied in this
homogeneous Ansatz.  (This is not always the case, though.)  An
explicit basis for the space of invariant closed 4-forms is then given
by the following six 4-forms:
\begin{equation}
  \begin{aligned}
    \Phi_1 &= L_{15}^* \wedge L_{25}^* \wedge L_{35}^* \wedge L_{45}^*\\
    \Phi_2 &= J_{45}^* \wedge J_{14}^* \wedge J_{24}^* \wedge J_{34}^* \\
    \Phi_3 &= J_{45}^* \wedge J_{15}^* \wedge J_{25}^* \wedge J_{35}^* \\
    \Phi_4 &= J_{45}^* \wedge \left( J_{14}^* \wedge J_{24}^* \wedge J_{35}^* + J_{14}^* \wedge J_{25}^* \wedge J_{34}^*  + J_{15}^* \wedge J_{24}^* \wedge J_{34}^*\right)\\
    \Phi_5 &= J_{45}^* \wedge \left(J_{14}^* \wedge J_{25}^* \wedge J_{35}^* + J_{15}^* \wedge J_{24}^* \wedge J_{35}^* + J_{15}^* \wedge J_{25}^* \wedge J_{34}^* \right)\\
    \Phi_6 &= - J_{14}^* \wedge J_{15}^* \wedge J_{24}^* \wedge J_{25}^* - J_{14}^* \wedge J_{15}^* \wedge J_{34}^* \wedge J_{35}^* - J_{24}^* \wedge J_{25}^* \wedge J_{34}^* \wedge J_{35}^*~,
  \end{aligned}
\end{equation}
whence the most general invariant closed 4-form is $F =
\sum_{\alpha=1}^6 \varphi_\alpha \Phi_\alpha$.  The Maxwell and
Einstein equations now become algebraic equations on the 9 real
parameters $\gamma_{1,2,3}>0$ and $\varphi_{1,\ldots,6}$.

It is convenient to analyse these equations to choose an ordered basis
$(X_\mu)_{\mu=0,1,\dots,9,\ten}$ for $\fm$:
\begin{equation}
  X_\mu = \left(J_{45}, L_{15},\dots,L_{45},J_{14},\dots,J_{34},J_{15},\dots,J_{35} \right)~,
\end{equation}
with corresponding canonical dual basis $\theta^\mu$ for $\fm^*$.
Then the inner product is given by
\begin{equation}
  g = -(\theta^0)^2 + \gamma_1 \left((\theta^1)^2 + \dots +
    (\theta^4)^2\right) + \gamma_2 \left((\theta^5)^2 + \dots +
    (\theta^7)^2\right) + \gamma_3 \left((\theta^8)^2 + \dots + (\theta^\ten)^2\right)~,
\end{equation}
and the most general closed 4-form by
\begin{multline}
  F = \varphi_1 \theta^{1234} + \varphi_2 \theta^{0567} + \varphi_3 \theta^{089\ten} + \varphi_4 \left(\theta^{056\ten} - \theta^{0579} + \theta^{0678}\right)\\
  + \varphi_5 \left(\theta^{059\ten} - \theta^{068\ten} + \theta^{0789}\right) + \varphi_6 \left(\theta^{5689} + \theta^{578\ten} + \theta^{679\ten}\right)~.
\end{multline}
It follows that if we let $F = \varphi_1 \theta^{1234} + \overline F$,
then $\half F \wedge F = \varphi_1 \theta^{1234} \wedge \overline F$.
In addition, from equation \eqref{eq:deltaF}, one finds that
\begin{multline}
  \label{eq:deltaFN=5}
  \delta F = -\frac{3\gamma_2\varphi_4}{\gamma_3} \theta^{567} +
  \frac{3\gamma_3\varphi_5}{\gamma_2} \theta^{89\ten} +
  \frac{2\varphi_6}{\gamma_2\gamma_3} \left(\theta^{058} +
    \theta^{069} + \theta^{07\ten}\right)\\ +
  \left(\frac{\gamma_3\varphi_2}{\gamma_2} - \frac{2\gamma_2
      \varphi_5}{\gamma_3}\right) \left(\theta^{56\ten} - \theta^{579}
    + \theta^{678}\right) + \left(\frac{2\gamma_3\varphi_4}{\gamma_2}
    - \frac{\gamma_2 \varphi_3}{\gamma_3}\right) \left(\theta^{59\ten}
    - \theta^{68\ten} + \theta^{789}\right)~.
\end{multline}
We note \emph{en passant} that, as expected, the only invariant
harmonic 4-form is proportional to the volume form on $S^4$: namely,
$\varphi_1 \theta^{1234}$.

The nonlinear Maxwell equation is equation \eqref{eq:maxwell}.  In
order to compute the Hodge $\star$ it is perhaps better to work with
an orthonormal coframe $\overline\theta^\mu$, where
\begin{equation}
  \overline\theta^\mu =
  \begin{cases}
    \theta^0 & \mu = 0\\
    \frac{1}{\sqrt{\gamma_1}}\theta^\mu & \mu \in \{1,2,3,4\}~,\\
    \frac{1}{\sqrt{\gamma_2}} \theta^\mu & \mu \in \{5,6,7\}~,\\
    \frac{1}{\sqrt{\gamma_3}} \theta^\mu & \mu \in \left\{8,9,\ten
    \right\}~,
 \end{cases}
\end{equation}
where we choose the positive square roots of the positive quantities
$\gamma_i$.  A short calculation later, one finds that
\begin{multline}
  -\star\half F \wedge F =
  -\frac{\gamma_2^{3/2}\varphi_1\varphi_3}{\gamma_1^2 \gamma_3^{3/2}}
  \theta^{567} +
  \frac{\gamma_3^{3/2}\varphi_1\varphi_2}{\gamma_1^2\gamma_2^{3/2}}
  \theta^{89\ten} -
  \frac{\varphi_1\varphi_6}{\gamma_1^2\gamma_2^{1/2}\gamma_3^{1/2}}
  \left(\theta^{058} + \theta^{069} + \theta^{07\ten}\right)\\ +
  \frac{\gamma_2^{1/2}\varphi_1\varphi_5}{\gamma_1^2\gamma_3^{1/2}}
  \left(\theta^{56\ten} - \theta^{579} + \theta^{678}\right) -
  \frac{\gamma_3^{1/2}\varphi_1\varphi_4}{\gamma_1^2\gamma_2^{1/2}}
  \left(\theta^{59\ten} - \theta^{68\ten} + \theta^{789}\right)~,
\end{multline}
which can be readily compared with equation \eqref{eq:deltaFN=5} in
order to arrive at the algebraic Maxwell equations: 
\begin{equation}
  \label{eq:men=5}
  \begin{aligned}[t]
    \varphi_4 &= \frac{\gamma_2^{1/2}\varphi_1\varphi_3}{3\gamma_1^2
      \gamma_3^{1/2}}\\
    \varphi_5 &=
    \frac{\gamma_3^{1/2}\varphi_1\varphi_2}{3\gamma_1^2\gamma_2^{1/2}}\\
    0 &= \left(\frac{\varphi_1}{\gamma_1^2} +
      \frac{2}{\gamma_2^{1/2}\gamma_3^{1/2}}\right) \varphi_6~.
 \end{aligned}
  \qquad\qquad
  \begin{aligned}[t]
    \frac{\gamma_2 \varphi_3}{\gamma_3} &=
    \left(\frac{2\gamma_3}{\gamma_2} +
      \frac{\gamma_3^{1/2}\varphi_1}{\gamma_1^2\gamma_2^{1/2}}\right)
    \varphi_4\\
  \frac{\gamma_3\varphi_2}{\gamma_2} &=
  \left(\frac{2\gamma_2}{\gamma_3} +  \frac{\gamma_2^{1/2}
      \varphi_1}{\gamma_1^2\gamma_3^{1/2}}\right) \varphi_5
  \end{aligned}
\end{equation}
The bottom equation defines two main branches of solutions, depending
on whether $\varphi_6$ vanishes.  The first two equations express
$\varphi_4$ and $\varphi_5$ in terms of $\varphi_3$ and $\varphi_2$,
respectively; whereas the remaining equations become:
\begin{equation}
  \varphi_3 = \frac13 \left(\frac{2\gamma_3}{\gamma_2} + \frac{\gamma_3^{1/2}\varphi_1}{\gamma_1^2\gamma_2^{1/2}}\right) \frac{\gamma_3^{1/2}\varphi_1}{\gamma_1^2 \gamma_2^{1/2}} \varphi_3 \qquad\text{and}\qquad
 \varphi_2 = \frac13 \left(\frac{2\gamma_2}{\gamma_3} +  \frac{\gamma_2^{1/2} \varphi_1}{\gamma_1^2\gamma_3^{1/2}}\right) \frac{\gamma_2^{1/2}\varphi_1}{\gamma_1^2\gamma_3^{1/2}}\varphi_2~.
\end{equation}
Notice the invariance of the equations under the simultaneous
exchanges: $\gamma_2\leftrightarrow \gamma_3$ and $\varphi_2
\leftrightarrow \varphi_3$.  This is nothing but the remnant of the
action of the normaliser of $\fh$ in $\fg$, which our choice of
diagonal inner product broke down to a $\ZZ/2\ZZ$ exchanging the $4$
and $5$ labels in $\fso(3,2)$.  This discrete symmetry relates some of
the branches below.

Each of these equations also defines two branches, depending on
whether $\varphi_3$ and $\varphi_2$ vanish or not.  In all, we have 8
branches of solutions, two pairs of which are related by the remaining
$\ZZ/2\ZZ$ action mentioned above.  They are given as follows:

\begin{enumerate}
\item $\varphi_2 = \varphi_3 = \varphi_6 =0$.  This implies that
  $\varphi_{i\neq 1} = 0$ and $\varphi_1$ remains free.

\item $\varphi_2 = \varphi_6 = 0$, but $\varphi_3$ remains free.
  Then there are two sub-branches, distinguished by the choice of sign
  in 
  \begin{equation}
    \frac{\varphi_1}{\gamma_1^2} = - \sqrt{\frac{\gamma_3}{\gamma_2}}
    \pm \sqrt{\frac{\gamma_3}{\gamma_2} +
      3\frac{\gamma_2}{\gamma_3}}~.
  \end{equation}

\item $\varphi_3 = \varphi_6 = 0$, but $\varphi_2$ remains free, and
  again two sub-branches distinguished by the sign in
  \begin{equation}
    \frac{\varphi_1}{\gamma_1^2} = - \sqrt{\frac{\gamma_2}{\gamma_3}}
    \pm \sqrt{\frac{\gamma_2}{\gamma_3} +
      3\frac{\gamma_3}{\gamma_2}}~.
  \end{equation}

\item $\varphi_6=0$, but $\varphi_2,\varphi_3$ remain free.  In this
  case, symmetry says that $\gamma_2 = \gamma_3$, and hence we have
  two sub-branches:
  \begin{enumerate}
  \item $\varphi_1 = \gamma_1^2$, $\varphi_4 = \varphi_3/3$, $\varphi_5 = \varphi_2/3$, and
  \item $\varphi_1 = -3 \gamma_1^2$, $\varphi_4 = -\varphi_3$, $\varphi_5 = -\varphi_2$.
  \end{enumerate}

\item $\varphi_2 = \varphi_3 = 0$, but $\varphi_6$ remains free.  This
  implies that $\varphi_1 = -2\gamma_1^2/\sqrt{\gamma_2\gamma_3}$.

\item $\varphi_2=0$, but $\varphi_3,\varphi_6$ remain free.  Then
  $\varphi_1 = - 2\gamma_1^2/\sqrt{\gamma_2\gamma_3}$ and then
  $\gamma_3 = 1 - \frac34 \gamma^2_2$.

\item $\varphi_3=0$, but $\varphi_2,\varphi_6$ remain free.  Again
  $\varphi_1 = - 2\gamma_1^2/\sqrt{\gamma_2\gamma_3}$ and then
  $\gamma_2 = 1 - \frac34 \gamma^3_2$.

\item $\varphi_2,\varphi_3,\varphi_6$ remain free.  In this case,
  symmetry dictates (and one can also check) that $\gamma_2 = \gamma_3
  = \frac23$, whence $\varphi_1 = -3 \gamma_1^2$.
\end{enumerate}

The Einstein equations become five algebraic equations on the
$\gamma_i$ and the $\varphi_i$.  We may use two of the Maxwell
equations in \eqref{eq:men=5} to eliminate $\varphi_4$ and $\varphi_5$
from the equations and we may use that $\gamma_1\gamma_2\gamma_3\neq
0$ to clear denominators and arrive after some simplification at the
following (almost) polynomial equations:
\begin{equation}
  \begin{aligned}[t]
  \label{eq:een=5}
 0 &= \varphi_1 \varphi_2 \varphi_3 \left(3 \gamma_1^2(\gamma_2 + \gamma_3)+ 2\sqrt{\gamma_2\gamma_3}\varphi_1\right)\\
  0 &= 3 \gamma_1^4 \varphi_6^2 + \gamma_2^2 \gamma_3^2 (12 \gamma_1^3 - \varphi_1^2) + 3\gamma_1^4 \gamma_2 \gamma_3 \left((\gamma_2-\gamma_3)^2 -1 \right)\\
  0 &= \left(3 \gamma_2 \gamma_1^4 + \gamma_3\varphi_1^2\right) \gamma_2^2\varphi_3^2 + 
  \left(3 \gamma_3 \gamma_1^4 + \gamma_2\varphi_1^2\right)\gamma_3^2 \varphi_2^2 - 9\gamma_1^4\gamma_2\gamma_3 \varphi_6^2 + 6\gamma_2^3\gamma_3^3\varphi_1^2 - 54 \gamma_1^3 \gamma_2^3 \gamma_3^3\\
  0 &= 6\gamma_1^4\gamma_2^3 \varphi_3^2 + \left( \gamma_2\varphi_1^2 - 3\gamma_1^4\gamma_3 \right)\gamma_3^2 \varphi_2^2 - 9\gamma_1^4\gamma_2 \gamma_3 \varphi_6^2 + 3\gamma_2^3\gamma_3^3\varphi_1^2 + 9\gamma_1^4\gamma_2^2\gamma_3^2 \left(\gamma_2^2 - \gamma_3^2 - 6 \gamma_2 + 1\right)\\
  0 &= 6\gamma_1^4\gamma_3^3 \varphi_2^2 + \left( \gamma_3\varphi_1^2 - 3\gamma_1^4\gamma_2 \right)\gamma_2^2 \varphi_3^2 - 9\gamma_1^4\gamma_2 \gamma_3 \varphi_6^2 + 3\gamma_2^3\gamma_3^3\varphi_1^2 + 9\gamma_1^4\gamma_2^2\gamma_3^2 \left(\gamma_3^2 - \gamma_2^2 - 6 \gamma_3 + 1\right)~.
\end{aligned}
\end{equation}
Notice that the first three equations are invariant under the remnant
$\ZZ/2\ZZ$ symmetry, whereas the last two equations are mapped into
each other.

We now insert each of the solution branches of the Maxwell equations
in turn into the Einstein equations.  We have used a mixture of
symbolic and numerical computation to arrive at the following results,
where the enumeration coincides with that of the solutions of the
Maxwell equations.

\begin{enumerate}
\item We find one solution: $\gamma_2=\gamma_3 = \frac23$, $\gamma_1 =
  \frac49$ and $\varphi_1 = \pm \frac89$.  This is a Freund--Rubin
  background, since the 4 - form is proportional to the volume form on
  the $S^4$.  We will see below that the 7-dimensional geometry is
  that of a lorentzian Sasaki--Einstein manifold, whence this
  background is supersymmetric.  This geometry will be studied in
  detail in Section \ref{sec:freund-rubin-backgr}, where we will show
  that it is Wick-related to a known Freund--Rubin $\AdS_4$
  background.
  
\item There are two branches, distinguished by the sign of the root in
  $\varphi_1$.
  \begin{itemize}
  \item[($+$)] In the positive branch, we find the following numerical
    solution:
    \begin{equation}\label{eq:numericalS4bg}
      \begin{aligned}[m]
        \gamma_1 &= 0.22776420155467458\\
        \gamma_2 &= 0.4670546272324634\\
        \gamma_3 &= 0.12728016028858763
      \end{aligned}
      \qquad\text{whence}\qquad
      \begin{aligned}[m]
        \varphi_1 &= 0.14715771499261474\\
        \varphi_3 &= \pm 0.27380714065085027~.
      \end{aligned}
    \end{equation}
    We do not discard the possibility that one can do better and write
    this solution in some iterated quadratic extension of the
    rationals, but we have not been able to do it.  The source of the
    difficulty comes from the fact that the solutions are built out of
    roots of a sixth order integer polynomial and we do not know if
    its Galois group is solvable.
    
  \item[($-$)] In the negative branch, we find the following solution
    \begin{equation}
      \gamma_1 = \gamma_2 = \gamma_3 = \frac13 \qquad \varphi_1 =
      -\frac13 \qquad\text{and}\qquad \varphi_3 = \pm
      \frac1{\sqrt{3}}~.
    \end{equation}
    We will see that this is part of a more general solution.
  \end{itemize}
  
\item This is just the previous branch \emph{mutatis mutandis}:
  exchange $\varphi_2$ and $\varphi_3$ and similarly $\gamma_2$ and
  $\gamma_3$.
  
\item
  \begin{enumerate}
  \item There are no solutions.
  \item There is a one-parameter family of solutions:
    \begin{equation}
      \gamma_1 = \gamma_2 = \gamma_3 = \frac13, \quad \varphi_1 =
      -\frac13, \quad \varphi_2 = \frac{1}{\sqrt{3}}\cos\alpha
      \quad\text{and}\quad \varphi_3 = \frac{1}{\sqrt{3}}\sin\alpha~.
    \end{equation}
    This generalises two of the solutions mentioned above, to which it
    reduces when the angle $\alpha$ obeys $2\alpha \in \pi\ZZ$.  It
    will be studied in more detail in Section
    \ref{sec:circle-backgrounds}.
  \end{enumerate}
  
\item There are no solutions.
  
\item There are no solutions.
  
\item This is the previous branch \emph{mutatis mutandis}, hence there are no solutions.
  
\item There are no solutions.
\end{enumerate}

In summary, we have found three classes of homogeneous supergravity
backgrounds not of AdS type with symmetry group locally isomorphic to
$\SO(3,2)\times \SO(5)$, and which we describe in more detail in
Section~\ref{sec:other-backgrounds}.

\section{Homogeneous anti de Sitter backgrounds}
\label{sec:homogeneous-anti-de}

In this section we study the existence of homogeneous (anti) de Sitter
backgrounds $G/H$ with $G$ locally isomorphic to $\SO(n) \times
\SO(3,2)$, for $n>4$.  Our first result is that there are no de Sitter
backgrounds, which allows us to focus on backgrounds of the form
$\AdS_4 \times X^7$, where $X = \SO(n)/H$ for $n=5,6,7$.  This means
that $H$ is a closed Lie subgroup of dimension $3,8,14$,
respectively.

One could ask whether there are backgrounds of the type $\AdS_p \times
X^{11-p}$ for $p\neq4$ and still of the form $G/H$.  This would
require $\SO(3,2)$ acting locally transitively on $\AdS_p$.  By
dimension, and since $\SO(3,2)$ must act effectively, $p\geq 4$.  One
can easily show that $\SO(3,2)$ cannot act locally transitively on
$\AdS_5$.  This is done by comparing the possible subgroups of
$\SO(3,2)$ which admit an embedding into $\SO(4,1)$, as listed in
\cite{MR0455998}, and checking that the linear isotropy representation
of the unique such subgroup (with Lie algebra of type $A_{5,35}^{0,1}$
in that paper's notation) does not in fact lie in $\SO(4,1)$.  We do
not know whether $\AdS_p$ for  $p>5$ admits an isometric transitive
action of $\SO(3,2)$ \cite{MO78272}.

\subsection{There are no de Sitter backgrounds}
\label{sec:there-are-no}

It is probably the case that $\SO(3,2)$ does not act isometrically on
any de Sitter space, but let us in any case show that the Einstein
equations for homogeneous backgrounds rule out a de Sitter solution.
Let us consider a geometry of the form $\dS_p \times M^{11-d}$.  Since
the only invariant forms on $\dS_p$ are the constant 0-forms and
constant multiples of the volume $p$-form $\nu$, if $p>4$ $F$
cannot have legs along the de Sitter directions, whence the de Sitter
components of the Einstein equation are given by
\begin{equation}
  R_{\mu\nu} = -\tfrac16 g_{\mu\nu} |F|^2 \implies g^{\mu\nu}R_{\mu\nu} = - \tfrac{p}6 |F|^2 \leq 0~,
\end{equation}
contradicting that de Sitter space has positive scalar curvature.  If
$p\leq 4$, then the most general invariant $F$ has the form
\begin{equation}
  F = \nu \wedge \alpha + \varphi \qquad\text{for}\quad \alpha \in \Omega^{4-p}(M), \varphi \in \Omega^4(M)~.
\end{equation}
The de Sitter components of the Einstein equation are now given by
\begin{equation}
  R_{\mu\nu} = - \half g_{\mu\nu} |\alpha|^2 - \tfrac16 g_{\mu\nu} (-|\alpha|^2 + |\varphi|^2) \implies g^{\mu\nu}R_{\mu\nu} = -\tfrac{p}6 (2|\alpha|^2 + |\varphi|^2) \leq 0~,
\end{equation}
again yielding a contradiction.  In summary, there are no homogeneous
de Sitter backgrounds.

\subsection{No new $n=7$ $\AdS_4$ backgrounds}
\label{sec:n=7}

From the results of Section \ref{sec:so(7)}, we see that there is a
unique 14-dimensional Lie subalgebra of $\fso(7)$, namely $\fg_2$.
The reductive split $\fso(7) = \fg_2 \oplus \fm$ is such that $\fm$ is
the 7-dimensional irreducible representation of $\fg_2$, whence the
homogeneous space $\SO(7)/G_2$ is locally isometric to the round
7-sphere, which admits an isometric action of $\SO(8)$ with stabiliser
$\SO(7)$.  Now the only homogeneous background $\AdS_4 \times
\SO(8)/\SO(7)$ is a Freund--Rubin background, because there are no
$\SO(8)$-invariant 4-forms on $\SO(8)/\SO(7)$; however, there are
$\SO(7)$-invariant 4-forms on $\SO(7)/G_2$ and hence in principle one
can ask whether there are supergravity backgrounds on $\AdS_4 \times
\SO(7)/G_2$ which are not of Freund--Rubin type.  Metrically, of
course, such backgrounds are locally isometric to $\AdS_4 \times S^7$,
but where the radii of curvature of the two spaces are fixed by the
flux.  Recall that $\SO(7)$-invariant 4-forms on $X=\SO(7)/G_2$ are in
one-to-one correspondence with $G_2$-invariant elements of $\Lambda^4
\fm$.  It is well-known that $(\Lambda^p \fm)^{G_2}$ is
one-dimensional for $p=3,4$.  If we let $\varphi$ denote a nonzero
$\SO(7)$-invariant 3-form, then the $\SO(7)$-invariant 4-form is
proportional to $\star\varphi$.  Moreover it is also the case that
$d\varphi$ is proportional to $\star\varphi$, whence $d\star\varphi =
0$.  It follows by dimension that $\star\varphi \wedge \star\varphi =
0$, and that $\delta\star\varphi \neq 0$, since it is in fact
proportional to $\varphi$.  Therefore letting $F = \alpha
\dvol_{\AdS_4} + \beta \star\varphi$ with $\alpha,\beta \in \RR$, we
see that $dF=0$ and that both $-\half \star F \wedge F$ and $\delta F$
are proportional to $\varphi$, whence we get an identity relating
$\beta$ and $\alpha\beta$, which means that either $\beta=0$
(Freund--Rubin) or else $\alpha$ is fixed and $\beta$ free. The former
background is the standard Freund--Rubin background $\AdS_4 \times
S^7$, whereas the latter is the Englert solution
\cite{Englert:1982vs}.  It may be worth writing these solutions
explicitly in our conventions.

We have $\fg = \fso(3,2) \oplus \fso(7)$ with bases $J_{\mu\nu}$ for
$\fso(3,2)$ and $L_{ab}$ for $\fso(7)$.  We have $\fh = \fso(3,1)
\oplus \fg_2$, where $\fso(3,1)$ is spanned by $J_{\mu\nu}$, with
$\mu,\nu=1,2,3,4$, and $\fg_2$ is spanned by the 14 linear
combinations in equation \eqref{eq:g2inso7}.  This means that an
ordered basis $(X_0,X_1,\dots,X_\ten)$ for the complement of $\fh$ in
$\fg$ is given by the following elements of $\fg$ in the order given:
\begin{multline}
  J_{45}, J_{15}, J_{25}, J_{35}, L_{12} + L_{34} - L_{67}, L_{13} - L_{24} + L_{57}, L_{14} + L_{23} - L_{56},\\
  L_{15} - L_{37} + L_{46}, L_{16} + L_{27} - L_{45}, L_{17} - L_{26} + L_{35}, L_{25} + L_{36} + L_{47}~.
\end{multline}
We let $(\theta^0,\dots,\theta^\ten)$ denote the canonical dual basis
for $\fm^*$.  The most general $H$-invariant lorentzian inner product
on $\fm$ is given by
\begin{equation}
  g = \gamma_0 \left(-(\theta^0)^2 + (\theta^1)^2 + (\theta^2)^2 + (\theta^3)^2\right) + \gamma_1 \sum_{i=4}^\ten (\theta^i)^2~,
\end{equation}
where $\gamma_0>0$ and $\gamma_1>0$.  Similarly, the most general
invariant 4-form on $\fm$ is given by
\begin{equation}
  F = \varphi_1 \theta^{0123} + \varphi_2 \left(\theta^{4578} + \theta^{459\ten} + \theta^{4679} - \theta^{468\ten} + \theta^{567\ten} + \theta^{5689} - \theta^{789\ten}\right)~,
\end{equation}
which is closed for all $\varphi_1$ and $\varphi_2$.  The homothety
invariance of the field equations allows us to set $\gamma_0=1$, and
we will do so.  The Maxwell equations \eqref{eq:maxwell} then become
\begin{equation}
  \varphi_2 \left(\varphi_1 + \frac6{\sqrt{\gamma_1}}\right)=0~,
\end{equation}
which has two branches: one where $\varphi_2=0$ and $\varphi_1$
remains free, and one where $\varphi_1 = - 6/\sqrt{\gamma_1}$ and
$\varphi_2$ remains free.  The Einstein equations \eqref{eq:einstein}
become
\begin{equation}
  18\gamma_1^4 = 2 \gamma_1^4 \varphi_1^2 + 7 \varphi_2^2  \qquad\text{and}\qquad 81 \gamma_1^3 = \gamma_1^4 \varphi_1^2 + 5\varphi_2^2~.
\end{equation}

The branch where $\varphi_2=0$ corresponds to the original
Freund--Rubin background \cite{FreundRubin}, in which $\varphi_1=\pm
3$ and $\gamma_1 = 9$.  Reintroducing the scale
$\lambda\in\RR^\times$, we have
\begin{equation}
  \begin{aligned}[m]
    \lambda^{-2}g &= -(\theta^0)^2 + (\theta^1)^2 + (\theta^2)^2 + (\theta^3)^2 + 9 \sum_{i=4}^\ten (\theta^i)^2\\
    \lambda^{-3}F &= 3 \theta^{0123}~.
  \end{aligned}
\end{equation}

In the second branch, and reintroducing the scale, we have
\begin{equation}\label{eq:EnglertSoln=7}  
  \begin{aligned}[m]
    \lambda^{-2}g &= -(\theta^0)^2 + (\theta^1)^2 + (\theta^2)^2 + (\theta^3)^2 + \tfrac{15}2 \sum_{i=4}^\ten (\theta^i)^2\\
    \lambda^{-3}F &= -2\sqrt{\tfrac65} \theta^{0123} \pm 3 \left(\tfrac{15}2\right)^{3/2} \left(\theta^{4578} + \theta^{459\ten} + \theta^{4679} - \theta^{468\ten} + \theta^{567\ten} + \theta^{5689} - \theta^{789\ten}\right)~.
  \end{aligned}
\end{equation}
This is Englert solution \cite{Englert:1982vs}, which is known not to
be supersymmetric.

\subsection{No new $n=6$ $\AdS_4$ backgrounds}
\label{sec:n=6}

As we saw in Section \ref{sec:so(6)}, there is unique subalgebra of
$\fso(6)$ of dimension $8$, namely $\fsu(3)$.  The reductive split
$\fso(6) = \fsu(3) \oplus \fp$ is such that $\fp$ is a reducible
representation of $\fsu(3)$, whose complexification $\fp\otimes_\RR
\CC = V_{[00]} \oplus V_{[10]}\oplus V_{[01]}$, where $[mn]$ are the
Dynkin labels of the representations, with $[00]$ corresponding to the
trivial one-dimensional representation and $[10]$ and $[01]$ the
fundamental and anti-fundamental three-dimensional representations,
respectively.  As a real representation, $\fp$ decomposes into the
direct sum of a the trivial one-dimensional representation and an
irreducible six-dimensional real representation whose complexification
is $V_{[10]}\oplus V_{[01]}$.  This means that there are two
parameters for the inner product on $\fp$, which together with the
radius of curvature of $\AdS_4$ makes three metric parameters.  There
is a four-dimensional space of invariant 4-forms: the volume form on
$\AdS_4$, the square of the invariant symplectic form on $\fp$ and two
more coming from the 3-forms on $V_{[10]}$ and on $V_{[01]}$ wedged
with any nonzero element of $V_{[00]}^*$.  Let us be more explicit.

We have $\fg = \fso(3,2) \oplus \fso(6)$ with bases $J_{\mu\nu}$ for
$\fso(3,2)$ and $L_{ab}$ for $\fso(6)$.  We have $\fh = \fso(3,1)
\oplus \fsu(3)$, where $\fso(3,1)$ is spanned by $J_{\mu\nu}$, with
$\mu,\nu=1,2,3,4$, and $\fsu(3)$ is spanned by the 8 linear
combinations in equation \eqref{eq:su3inso6}.  This means that an
ordered basis $(X_0,X_1,\dots,X_\ten)$ for the complement of $\fh$ in
$\fg$ is given by the following elements of $\fg$ in the order given:
\begin{multline}
  J_{45}, J_{15}, J_{25}, J_{35}, L_{12} + L_{34} + L_{56}, L_{13} - L_{24}, L_{14} + L_{23}, L_{15} - L_{26}, L_{16} + L_{25}, L_{35} - L_{46}, L_{36} + L_{45}~.
\end{multline}
We let $(\theta^0,\dots,\theta^\ten)$ denote the canonical dual basis
for $\fm^*$.  The most general $H$-invariant lorentzian inner product
on $\fm$ is given by
\begin{equation}
  g = \gamma_0 \left(-(\theta^0)^2 + (\theta^1)^2 + (\theta^2)^2 + (\theta^3)^2\right) + \gamma_1 (\theta^4)^2 + \gamma_2 \sum_{i=5}^\ten (\theta^i)^2~,
\end{equation}
where $\gamma_{0,1,2}>0$.  Similarly, the most general invariant
4-form on $\fm$ is given by
\begin{multline}
  F = \varphi_1 \theta^{0123} + \varphi_2 \left(\theta^{4579} - \theta^{458\ten} - \theta^{467\ten} - \theta^{4689}\right)\\
 + \varphi_3 \left(\theta^{457\ten} + \theta^{4589}+ \theta^{4679}- \theta^{468\ten}\right) + \varphi_4 \left(\theta^{5678} + \theta^{569\ten} + \theta^{789\ten}\right)~,
\end{multline}
which is closed for all values of $\varphi_i$.  The homothety
invariance of the field equations allows us to set $\gamma_0=1$, and
we will do so.  The Maxwell equations \eqref{eq:maxwell} then become
\begin{equation}
  \varphi_2 \left(\frac{6}{\sqrt{\gamma_1}}+\varphi_1\right) = 0 \qquad \varphi_3
   \left(\frac{6}{\sqrt{\gamma_1}}+\varphi_1\right) = 0 \qquad \varphi_4 \left(\frac{8
   \sqrt{\gamma_1}}{3 \gamma_2}+\varphi_1\right) = 0~,
\end{equation}
which has several branches:
\begin{enumerate}
\item $\varphi_2=\varphi_3=\varphi_4 =0$;
\item $\varphi_2=\varphi_3=0$, $\varphi_4 \neq 0$: whence $\varphi_1 = -\frac{8\sqrt{\gamma_1}}{3\gamma_2}$;
\item $\varphi_2^2 + \varphi_3^2 \neq 0$, $\varphi_4=0$: whence $\varphi_1 = -\frac{6}{\sqrt{\gamma_1}}$;
\item $\varphi_2^2 + \varphi_3^2 \neq 0$, $\varphi_4\neq 0$: whence $\varphi_1 = -\frac{6}{\sqrt{\gamma_1}}$ and $\gamma_2 = 4\gamma_1/9$.
\end{enumerate}
The Einstein equations \eqref{eq:einstein} become
\begin{equation}
  \begin{aligned}[m]
    0 &= 3 \gamma_1 \gamma_2^4 \varphi_1^2+6 \gamma_2 \varphi_2^2+6 \gamma_2 \varphi_3^2+9 \gamma_1 \varphi_4^2-144 \gamma_1 \gamma_2^3+16 \gamma_1^2 \gamma_2^2\\
    0 &= 2 \gamma_1 \gamma_2^4 \varphi_1^2+4 \gamma_2 \varphi_2^2+4 \gamma_2 \varphi_3^2+3 \gamma_1 \varphi_4^2-18 \gamma_1 \gamma_2^4\\
    0 &= \gamma_1 \gamma_2^4 \varphi_1^2+8 \gamma_2 \varphi_2^2+8 \gamma_2 \varphi_3^2-3 \gamma_1 \varphi_4^2-16 \gamma_1^2 \gamma_2^2
  \end{aligned}
\end{equation}
It is now a simple matter to specialise the Einstein equations to each
of the branches of solutions of the Maxwell equations.  We find three
kinds of solutions: the original Freund--Rubin solution, the Englert
solution and a circle's worth of solutions found by Pope and Warner
\cite{Pope:1984bd,Pope:1984jj}.  In detail, we have the following
results for the above four branches of solutions of the Maxwell
equations.

\begin{enumerate}
\item This is the Freund--Rubin background.  The only solution to the
  Einstein equations are $\varphi_1 = \pm 3$, $\gamma_1 = 9$ and
  $\gamma_2 = 4$.  Reintroducing the scale, we have
  \begin{equation}
    \begin{aligned}[m]
      \lambda^{-2}g &= -(\theta^0)^2 + (\theta^1)^2 + (\theta^2)^2 + (\theta^3)^2 + 9 (\theta^4)^2 + 4 \sum_{i=5}^\ten (\theta^i)^2\\
      \lambda^{-3}F &= 3 \theta^{0123}~.
    \end{aligned}
  \end{equation}

\item There are no real solutions to the Einstein equations.

\item Here $\varphi_1 = -\sqrt{3}$, $\varphi_2 + i \varphi_3 = 32
  \sqrt{\frac23} e^{i\alpha}$, $\gamma_1 = 12$ and $\gamma_2 =
  \frac83$, where $\alpha$ is an angle.  Reintroducing the scale, we
  have
  \begin{equation}
    \begin{aligned}[m]
      \lambda^{-2}g &= -(\theta^0)^2 + (\theta^1)^2 + (\theta^2)^2 + (\theta^3)^2 + 12 (\theta^4)^2 + \tfrac83 \sum_{i=5}^\ten (\theta^i)^2\\
      \lambda^{-3}F &= -\sqrt{3} \theta^{0123} + 32\sqrt{\tfrac23} \cos\alpha \left(\theta^{4579} - \theta^{458\ten} - \theta^{467\ten} - \theta^{4689}\right)\\
      & \qquad {} + 32\sqrt{\tfrac23} \sin\alpha \left(\theta^{457\ten} + \theta^{4589}+ \theta^{4679}- \theta^{468\ten}\right)~.
    \end{aligned}
  \end{equation}
  This is the solution found by Pope and Warner.

\item Here $\varphi_1 = - 2\sqrt{\frac65}$, $\varphi_2 + i \varphi_3 = 10 \sqrt{\frac{10}3} e^{i\alpha}$, $\varphi_4 = \pm \frac{20}3 \sqrt{\frac{10}3}$, $\gamma_1 = \frac{15}2$ and $\gamma_2 = \frac{10}3$, where $\alpha$ is again an angle.  Reintroducing the scale, we have
  \begin{equation}
    \begin{aligned}[m]
      \lambda^{-2}g &= -(\theta^0)^2 + (\theta^1)^2 + (\theta^2)^2 + (\theta^3)^2 + \tfrac{15}2 (\theta^4)^2 + \tfrac{10}3 \sum_{i=5}^\ten (\theta^i)^2\\
      \lambda^{-3}F &= -2\sqrt{\tfrac65} \theta^{0123} + 10\sqrt{\tfrac{10}3} \cos\alpha \left(\theta^{4579} - \theta^{458\ten} - \theta^{467\ten} - \theta^{4689}\right)\\
      & \qquad {} + 10\sqrt{\tfrac{10}3}  \sin\alpha \left(\theta^{457\ten} + \theta^{4589}+ \theta^{4679}- \theta^{468\ten}\right) \pm \tfrac{20}3 \sqrt{\tfrac{10}3} \left(\theta^{5678} + \theta^{569\ten} + \theta^{789\ten}\right)~.
    \end{aligned}
  \end{equation}
  This is again Englert's solution, but in a language where only the
  $\SO(6)$ symmetry is manifest.  This explains the fact that we get a
  circle of solutions.  The normaliser of $\SO(6)$ in $\SO(8)$
  contains an $\SO(2)$ subgroup (in fact, in the centraliser) and the
  circle is nothing but the orbit of this subgroup.  Each point in the
  circle is invariant under a different $\SO(7)$ subgroup of $\SO(8)$
  containing the same $\SO(6)$ subgroup.  These $\SO(7)$ subgroups are
  conjugate in $\SO(8)$ under the action of the normaliser of
  $\SO(6)$.  We will see below when discussing $n=5$ backgrounds that
  we get a 2-sphere's worth of Englert solutions, where the 2-sphere
  is the orbit of the centraliser of $\SO(5)$ in $\SO(8)$, which is an
  $\SO(3)$ subgroup.

\end{enumerate}

\subsection{Possible new $n=5$ $\AdS_4$ backgrounds}
\label{sec:n=5}

From the results of Section \ref{sec:so(5)} we have three
3-dimensional subalgebras of $\fso(5)$, all isomorphic to $\fso(3)$.
We can distinguish them by what they do to the five-dimensional
irreducible representation of $\fso(5)$.  One acts irreducibly, a
second $\fso(3)$ subalgebra decomposes the five-dimensional
representation as $2V_0 \oplus V_2$, where $V_n$ is the
($n+1$)-dimensional irreducible representation of $\fso(3)$, and the
third $\fso(3)$ subalgebra decomposes it as $V_0 \oplus V_3$.  If we
let $L_{ab}$ denote the standard basis for $\fso(5)$, then the three
$\fso(3)$ subalgebras are the following:
\begin{enumerate}
\item $\fso(3)_{\text{irr}}$ with basis given by equation \eqref{eq:so3irrinso5}.
\item $\fso(3) = \RR\left< L_{12}, L_{13}, L_{23}\right>$
\item $\fso(3)_+ = \RR\left< L_1^+, L_2^+, L_3^+\right>$
\end{enumerate}

\subsubsection{$\fso(3)_{\text{irr}}$ isotropy}
\label{sec:fso3_textirr}

The first case, where the subalgebra is $\fso(3)_{\text{irr}}$, is the
simplest.  The complement of $\fso(3)_{\text{irr}}$ in $\fso(5)$ is
irreducible, so there is (up to the overall homothety) one metric
parameter.  There is a two-dimensional space of closed invariant
4-forms, spanned by the volume form on $\AdS_4$ and a 4-form on the
riemannian factor.  The supergravity field equations reveal two
backgrounds, which can be shown to be the original Freund--Rubin and
Englert backgrounds.

\subsubsection{$\fso(3)$ isotropy}
\label{sec:fso3-isotropy}

In the second case, the isotropy subalgebra is $\fso(3)$, whose
complement in $\fso(5)$ decomposes as $2V_2 \oplus V_0$, whence there
are four metric parameters, which are reduced to three by the action
of the normaliser.  In particular, we can choose the inner product to
be diagonal relative to the following ordered basis for $\fm$:
\begin{equation}
  (X_0,X_1,\dots,X_\ten) = (J_{45}, J_{15}, J_{25}, J_{35}, L_{14}, L_{24}, L_{34}, L_{15}, L_{25}, L_{35}, L_{45})~.
\end{equation}
Indeed, in terms of the canonical dual bases for $\fm^*$, we can write
the invariant metric as
\begin{equation}\label{eq:gn=5so3}
  g = \gamma_0 \left(-(\theta^0)^2 + (\theta^1)^2 + (\theta^2)^2 + (\theta^3)^2 \right) + \gamma_1 \sum_{i=4}^6 (\theta^i)^2 + \gamma_2 \sum_{i=7}^9 (\theta^i)^2  + \gamma_3 (\theta^\ten)^2~,
\end{equation}
with $\gamma_{0,1,2,3}>0$.  Using the homothety invariance, we can set
$\gamma_0=1$ without loss of generality.  The space of invariant
closed 4-forms is 6-dimensional, whence the most general such $F$ is
\begin{multline}
  \label{eq:Fn=5so3}
  F = \varphi_1 \theta^{0123} + \varphi_2 \theta^{456\ten} + \varphi_3 \left(\theta^{4578} + \theta^{4679} + \theta^{5689} \right) + \varphi_4 \left(\theta^{459\ten} - \theta^{468\ten} + \theta^{567\ten} \right)\\
  + \varphi_5 \left(\theta^{489\ten} - \theta^{579\ten} +  \theta^{678\ten} \right) + \varphi_6 \theta^{789\ten}~.
\end{multline}

The Maxwell equations \eqref{eq:maxwell} become the following equations
\begin{equation}
  \begin{aligned}[t]
    0 &= \frac{\sqrt{\gamma_2} \varphi_1 \varphi_2}{\sqrt{\gamma_1}}+\frac{3 \varphi_5}{\sqrt{\gamma_3}}\\
    0 &= \frac{3 \varphi_4}{\sqrt{\gamma_3}}+\frac{\sqrt{\gamma_1} \varphi_1 \varphi_6}{\sqrt{\gamma_2}}\\
    0 &= \frac{2 \sqrt{\gamma_3} \varphi_3}{\sqrt{\gamma_1}}-\frac{\varphi_1 \varphi_3}{\sqrt{\gamma_2}}
 \end{aligned}
  \qquad\qquad
  \begin{aligned}[t]
   0 &= -\frac{\sqrt{\gamma_2} \varphi_1 \varphi_4}{\sqrt{\gamma_1}}+\frac{2 \gamma_2 \varphi_4}{\gamma_1
      \sqrt{\gamma_3}}-\frac{\gamma_1 \varphi_6}{\gamma_2 \sqrt{\gamma_3}}\\
    0 &= -\frac{\gamma_2 \varphi_2}{\gamma_1 \sqrt{\gamma_3}}-\frac{\sqrt{\gamma_1} \varphi_1 \varphi_5}{\sqrt{\gamma_2}}+\frac{2 \gamma_1 \varphi_5}{\gamma_2 \sqrt{\gamma_3}}
  \end{aligned}
\end{equation}
The first two equations on the left allow us to solve for
$\varphi_{4,5}$:
\begin{equation}
  \varphi_4 = - \frac{\sqrt{\gamma_1\gamma_3}}{3\sqrt{\gamma_2}}\varphi_1\varphi_6\qquad\text{and}\qquad
  \varphi_5 = - \frac{\sqrt{\gamma_2\gamma_3}}{3\sqrt{\gamma_1}}\varphi_1\varphi_2~.
\end{equation}
Inserting this in the remaining equations we are left with
\begin{equation}
  \begin{aligned}[m]
    0 &= \varphi_3 \left(\varphi_1 - \frac{2 \sqrt{\gamma_2\gamma_3}}{\sqrt{\gamma_1}}\right)\\
    0 &= \varphi_2 \left(\varphi_1^2 - \frac{2 \sqrt{\gamma_1} \varphi_1}{\sqrt{\gamma_2\gamma_3}} - \frac{3 \gamma_2}{\gamma_1 \gamma_3}\right)\\
    0 &= \varphi_6 \left(\varphi_1^2 - \frac{2 \sqrt{\gamma_2} \varphi_1}{\sqrt{\gamma_1\gamma_3}} - \frac{3 \gamma_1}{\gamma_2 \gamma_3}\right)~,
  \end{aligned}
\end{equation}
which leads to eight branches depending on whether $\varphi_{2,3,6}$
do or do not vanish.  For each such branch we have analysed the
Einstein equations \eqref{eq:einstein}, given by
\begin{equation}
  \begin{aligned}[m]
   0 &= \varphi_1 \varphi_2 \varphi_6 \left(3 (\gamma_1 + \gamma_2) - 2 \sqrt{\gamma_1\gamma_2\gamma_3} \varphi_1\right)\\
   0 &= -3 \gamma_2{}^3 \gamma_3 \gamma_1{}^3 \varphi_1{}^2-6 \gamma_1{}^3 \varphi_6{}^2-\gamma_2{}^2 \gamma_3 \gamma_1 \varphi_1{}^2 \varphi_2{}^2-9 \gamma_2 \gamma_3 \gamma_1 \varphi_3{}^2+3 \gamma_2{}^3 \varphi_2{}^2-9 \gamma_2{}^2 \gamma_1{}^4\\
   & \qquad {} + 54 \gamma_2{}^2 \gamma_3 \gamma_1{}^3+9 \gamma_2{}^4 \gamma_1{}^2-9 \gamma_2{}^2 \gamma_3{}^2 \gamma_1{}^2\\
   0 &= -3 \gamma_2{}^3 \gamma_3 \gamma_1{}^3 \varphi_1{}^2+3 \gamma_1{}^3 \varphi_6{}^2-\gamma_2 \gamma_3 \gamma_1{}^2 \varphi_1{}^2 \varphi_6{}^2-9 \gamma_2 \gamma_3 \gamma_1 \varphi_3{}^2-6 \gamma_2{}^3 \varphi_2{}^2+9 \gamma_2{}^2 \gamma_1{}^4\\
   & \qquad { }- 9 \gamma_2{}^4 \gamma_1{}^2-9 \gamma_2{}^2 \gamma_3{}^2 \gamma_1{}^2+54 \gamma_2{}^3 \gamma_3 \gamma_1{}^2\\
   0 &= -3 \gamma_2{}^3 \gamma_3 \gamma_1{}^3 \varphi_1{}^2-6 \gamma_1{}^3 \varphi_6{}^2-2 \gamma_2 \gamma_3 \gamma_1{}^2 \varphi_1{}^2 \varphi_6{}^2-2 \gamma_2{}^2 \gamma_3 \gamma_1 \varphi_1{}^2 \varphi_2{}^2+9 \gamma_2 \gamma_3 \gamma_1 \varphi_3{}^2\\
   & \qquad {} - 6 \gamma_2{}^3 \varphi_2{}^2-27 \gamma_2{}^2 \gamma_1{}^4+54 \gamma_2{}^3 \gamma_1{}^3-27 \gamma_2{}^4 \gamma_1{}^2+27 \gamma_2{}^2 \gamma_3{}^2 \gamma_1{}^2\\
   0 &= -6 \gamma_2{}^3 \gamma_3 \gamma_1{}^3 \varphi_1{}^2-3 \gamma_1{}^3 \varphi_6{}^2-\gamma_2 \gamma_3 \gamma_1{}^2 \varphi_1{}^2 \varphi_6{}^2-\gamma_2{}^2 \gamma_3 \gamma_1 \varphi_1{}^2 \varphi_2{}^2-9 \gamma_2 \gamma_3 \gamma_1 \varphi_3{}^2\\
   &\qquad {} - 3 \gamma_2{}^3 \varphi_2{}^2+54 \gamma_2{}^3 \gamma_3 \gamma_1{}^3~.
  \end{aligned}
\end{equation}

The end result is that beyond a known Freund--Rubin background $\AdS_4
\times V_2(\RR^5)$ and the Pope--Warner background, we obtain the
following backgrounds, where we have reintroduced the scale:
\begin{enumerate}
\item With $\sigma$ a sign, we have
  \begin{equation}\label{eq:N=5.1}
    \begin{aligned}[m]
      \lambda^{-2}g &= -(\theta^0)^2 + (\theta^1)^2 + (\theta^2)^2 + (\theta^3)^2 + \tfrac{39 - \sigma\sqrt{201}}{22} \sum_{i=4}^6 (\theta^i)^2 + \tfrac{27 + \sigma\sqrt{201}}{22} \sum_{i=7}^9 (\theta^i)^2  + (\theta^\ten)^2~,\\
      \lambda^{-3}F &= \sqrt{\tfrac{19 + \sigma \sqrt{201}}{5}} \theta^{0123} \pm \tfrac{2\sqrt{15\left(2155 - \sigma 31 \sqrt{201}\right)}}{121} \left(\theta^{4578} + \theta^{4679} + \theta^{5689} \right)~;
    \end{aligned}
  \end{equation}
  
\item A background we can only approximate numerically: $\varphi_2=
  \varphi_3 = \varphi_5 = 0$, together with
  \begin{equation}\label{eq:N=5.2}
    \begin{aligned}[m]
      \varphi_1 &= -1.3538010207764224\\
      \varphi_4 &= \pm 4.562584323795499\\
      \varphi_6 &= \pm 2.51893274180765
    \end{aligned}
    \qquad\qquad
    \begin{aligned}[m]
      \gamma_1 &= 2.0506059513936354\\
      \gamma_2 &= 0.5588242551644832\\
      \gamma_3 &= 4.390505589439397~,
    \end{aligned}
  \end{equation}
  and another background obtained from this by $\varphi_6
  \leftrightarrow \varphi_2$, $\varphi_4 \leftrightarrow \varphi_5$
  and $\gamma_1 \leftrightarrow \gamma_2$.

\item A circle's worth of backgrounds with metric
  \begin{equation}\label{eq:N=5.3g}
    \lambda^{-2} g = -(\theta^0)^2 + (\theta^1)^2 + (\theta^2)^2 +
    (\theta^3)^2 + \tfrac{3+\sqrt{6}}4 \sum_{i=4}^9 (\theta^i)^2 +
    \tfrac32 (\theta^\ten)^2~, 
  \end{equation}
  and 4-form
  \begin{multline}\label{eq:N=5.3F}
    \lambda^{-3} F = \sqrt{6} \theta^{0123} \pm \tfrac38
    \sqrt{\tfrac{29 + 12 \sqrt{6}}2} \left(\theta^{4578} +
      \theta^{4679} + \theta^{5689} \right)\\
    + \tfrac{3\sqrt{3 \left(3+\sqrt{6}\right)}}{8} \left( \cos\alpha
      \left( \theta^{456\ten} - \theta^{489\ten} + \theta^{579\ten} -
        \theta^{678\ten} \right) + \sin\alpha \left(\theta^{789\ten} -
        \theta^{459\ten} + \theta^{468\ten} - \theta^{567\ten}
      \right)\right)~,
  \end{multline}
  with $\alpha$ an angle.  We can write $F$ in a more transparent way
  as follows:
  \begin{multline}\label{eq:N=5.3cali}
    \lambda^{-3} F = \sqrt{6} \theta^{0123} \mp \tfrac3{16}
    \sqrt{\tfrac{29 + 12 \sqrt{6}}2} (\theta^{47} + \theta^{58} +
    \theta^{69})^{\wedge 2}\\
    + \tfrac{3\sqrt{3 \left(3+\sqrt{6}\right)}}{8} \Re
    \left(e^{i\alpha}(\theta^4+i\theta^7)\wedge
      (\theta^5+i\theta^8)\wedge (\theta^6+i\theta^9)\right)
    \wedge\theta^\ten~,
  \end{multline}
where we recognise the transverse Kähler calibration $\theta^{47} +
\theta^{58} + \theta^{69}$ and the transverse special lagrangian
calibration $\Re \left(e^{i\alpha}(\theta^4+i\theta^7)\wedge
  (\theta^5+i\theta^8)\wedge (\theta^6+i\theta^9)\right)$.

This background is obtained from the Freund--Rubin background $\AdS_4
\times V_2(\RR^5)$ \cite[Appendix~C]{CastellaniRomansWarner} by the
Englert procedure of constructing a 4-form out of the Killing spinors
\cite{Englert:1982vs,BEdWN}.  The angle $\alpha$ parametrises the
choice of the two Killing spinors out of which we make up the part of
$F$ with no legs along $\AdS_4$.  The background $\AdS_4 \times
V_2(\RR^5)$ has $\SU(3)$ holonomy.  This means that $V_2(\RR^5)$
admits a two-dimensional space of real Killing spinors.  Depending on
which two spinors we pick, we can construct a Kähler calibration and
one of the circle's worth of special lagrangian calibrations on the
codimension-one subbundle of the tangent bundle whose fibre at the
origin is spanned by $X_4,\dots,X_9$.  In addition, the tangent
representation of $\SU(3)$ leaves invariant one direction, which is
spanned by $X_\ten$ at the origin.  We now recognise the second term
of $F$ in the expression \eqref{eq:N=5.3cali} as the square of the
transverse Kähler calibration (itself a calibration) and the third
term as one of the transverse special lagrangian calibrations wedged
with the invariant form $\theta^\ten$.  This solution is also
mentioned in \cite[Appendix~C]{CastellaniRomansWarner} but not given
explicitly.  As usual in the Englert procedure, supersymmetry is
broken.  The form of the solution suggests that we should be able to
obtain it as well via the Pope--Warner procedure in
\cite{Pope:1984jj}, but we have not tried to do this.
\end{enumerate}

The first two backgrounds seem to be new.

\subsubsection{$\fso(3)_+$ isotropy}
\label{sec:fso3_+-isotropy}

In the final case, the isotropy algebra is $\fso(3)_+$, whose
complement in $\fso(5)$ decomposes into $3V_0 \oplus V_3$, with $3V_0$
corresponding to the $\fso(3)_-$ subalgebra of $\fso(5)$ and $V_3$
corresponding to the four-dimensional representation spanned by the
$L_{a5}$ with $a=1,2,3,4$.  This would seem to require 8 parameters to
describe the metric, but in fact we can use the action of the
normaliser $\SO(3)_-$ in order to diagonalise the inner product.
Indeed, defining the following ordered basis for $\fm$:
\begin{equation}
  (X_0,X_1,\dots,X_\ten) = (J_{45}, J_{15}, J_{25}, J_{35}, L^-_1, L^-_2, L^-_3, L_{15}, L_{25}, L_{35}, L_{45})~,
\end{equation}
and the canonical dual bases for $\fm^*$, we can write the invariant metric as
\begin{equation}
  g = \gamma_0 \left(-(\theta^0)^2 + (\theta^1)^2 + (\theta^2)^2 + (\theta^3)^2 \right) + \gamma_1 (\theta^4)^2 + \gamma_2 (\theta^5)^2 + \gamma_3 (\theta^6)^2 + \gamma_4 \sum_{i=7}^\ten (\theta^i)^2~,
\end{equation}
with $\gamma_{0,1,2,3,4}>0$.  Using the homothety invariance, we can
set $\gamma_0=1$ without loss of generality.  The space of invariant
closed 4-forms is 8-dimensional, with the most general such $F$ given
by
\begin{multline}
  \label{eq:Fn=5so3+}
  F = \varphi_1 \theta^{0123} + \varphi_2
  \left(\theta^{4578}-\theta^{459\ten}\right) + \varphi_3
  \left(\theta^{4579} + \theta^{458\ten} +  \theta^{4678} -
    \theta^{469\ten}\right)\\
  + \varphi_4 \left(\theta^{457\ten} - \theta^{4589} - \theta^{5678} + \theta^{569\ten}\right)
  + \varphi_5 \left(\theta^{4679} + \theta^{468\ten} \right)\\ + \varphi_6 \left(\theta^{467\ten} - \theta^{4689} -  \theta^{5679} - \theta^{568\ten}\right) + \varphi_7 \left(\theta^{567\ten} - \theta^{5689}\right) + \varphi_8 \theta^{789\ten}~.
\end{multline}

The Maxwell equations \eqref{eq:maxwell} are the following:
\begin{equation}\label{eq:MWEqn=5so3+}  
  \begin{aligned}[m]
    0 &= \varphi_4 \left(\sqrt{\gamma_1 \gamma_2 \gamma_3} \varphi_1+\gamma_1+\gamma_3\right)\\
    0 &= \varphi_6 \left(\sqrt{\gamma_1 \gamma_2 \gamma_3} \varphi_1+\gamma_1+\gamma_2\right)\\
    0 &= \varphi_3 \left(\sqrt{\gamma_1 \gamma_2 \gamma_3} \varphi_1+\gamma_2+\gamma_3\right)\\
    0 &= \gamma_3 \varphi_2 + \gamma_2 \varphi_5 + \gamma_1 \varphi_7 + \sqrt{\gamma_1\gamma_2 \gamma_3} \varphi_1 \varphi_7 + \gamma_1 \gamma_2 \gamma_3 \gamma_4^{-2}\varphi_8\\
    0 &= \gamma_3 \varphi_2 - \gamma_2 \varphi_5 - \gamma_1 \varphi_7 - \sqrt{\gamma_1\gamma_2\gamma_3} \varphi_1 \varphi_5 +  \gamma_1 \gamma_2\gamma_3 \gamma_4^{-2}\varphi_8\\
    0 &= \gamma_3 \varphi_2 - \gamma_2 \varphi_5 + \gamma_1 \varphi_7 + \sqrt{\gamma_1\gamma_2\gamma_3} \varphi_1 \varphi_2 - \gamma_1\gamma_2\gamma_3 \gamma_4^{-2}\varphi_8\\
    0 &= \gamma_3 \varphi_2+\gamma_2 \varphi_5-\gamma_1 \varphi_7-\tfrac12 \sqrt{\gamma_1 \gamma_2 \gamma_3} \varphi_1\varphi_8
  \end{aligned}
\end{equation}

The Einstein equations are
\begin{align*}
  0 &= (\varphi_2 +\varphi_5) \varphi_3 + \varphi_4 \varphi_6\\
  0 &= \varphi_3 \varphi_4+(\varphi_5 -\varphi_7) \varphi_6\\
   0 &= (\varphi_2 -\varphi_7) \varphi_4 + \varphi_3 \varphi_6\\
  9 &= \varphi_1^2 + \frac{\varphi_8^2}{2 \gamma_4^4} + \frac{\gamma_3\varphi_2^2 + \gamma_2 \varphi_5^2 + \gamma_1 \varphi_7^2 + (\gamma_2+\gamma_3)\varphi_3^2 + (\gamma_1+\gamma_3)\varphi_4^2  + (\gamma_1+\gamma_2)\varphi_6^2}{\gamma_1 \gamma_2 \gamma_3 \gamma_4^2} \\
 3 \left(\frac2{\gamma_4} - 1\right) &= \frac{\varphi_8^2}{2\gamma_4^4} + \frac{(\gamma_1+\gamma_2+\gamma_3)}{\gamma_4^2}\\
 0 &= \frac{\gamma_1^2}{2 \gamma_2 \gamma_3}+\frac{\gamma_1^2}{\gamma_4^2}-\frac{1}{6} \varphi_1^2 \gamma_1+\frac{\varphi_4^2 \gamma_1}{3 \gamma_2 \gamma_3 \gamma_4^2}+\frac{\varphi_6^2 \gamma_1}{3 \gamma_2 \gamma_3 \gamma_4^2}+\frac{\varphi_7^2 \gamma_1}{3 \gamma_2 \gamma_3 \gamma_4^2}+\frac{\varphi_8^2 \gamma_1}{6 \gamma_4^4}-\frac{2 \varphi_2^2}{3 \gamma_2 \gamma_4^2}\\
 & \qquad -\frac{2\varphi_3^2}{3 \gamma_2 \gamma_4^2}-\frac{2 \varphi_3^2}{3 \gamma_3 \gamma_4^2}-\frac{2 \varphi_4^2}{3 \gamma_2 \gamma_4^2}-\frac{2 \varphi_5^2}{3 \gamma_3 \gamma_4^2}-\frac{2 \varphi_6^2}{3 \gamma_3 \gamma_4^2}-\frac{\gamma_3}{2 \gamma_2}-\frac{\gamma_2}{2 \gamma_3}+1\\
 0 &= \frac{\gamma_2^2}{2 \gamma_1 \gamma_3}+\frac{\gamma_2^2}{\gamma_4^2}-\frac{1}{6} \varphi_1^2 \gamma_2+\frac{\varphi_3^2 \gamma_2}{3 \gamma_1 \gamma_3 \gamma_4^2}+\frac{\varphi_5^2 \gamma_2}{3 \gamma_1 \gamma_3 \gamma_4^2}+\frac{\varphi_6^2 \gamma_2}{3 \gamma_1 \gamma_3 \gamma_4^2}+\frac{\varphi_8^2 \gamma_2}{6 \gamma_4^4}-\frac{2 \varphi_2^2}{3 \gamma_1 \gamma_4^2}\\
 & \qquad -\frac{2 \varphi_3^2}{3 \gamma_1 \gamma_4^2}-\frac{2 \varphi_4^2}{3 \gamma_1 \gamma_4^2}-\frac{2 \varphi_4^2}{3 \gamma_3 \gamma_4^2}-\frac{2 \varphi_6^2}{3 \gamma_3 \gamma_4^2}-\frac{2 \varphi_7^2}{3 \gamma_3 \gamma_4^2}-\frac{\gamma_3}{2
 \gamma_1}-\frac{\gamma_1}{2 \gamma_3}+1\\
 0 &= \frac{\gamma_3^2}{2 \gamma_1 \gamma_2}+\frac{\gamma_3^2}{\gamma_4^2}-\frac{1}{6} \varphi_1^2
 \gamma_3+\frac{\varphi_2^2 \gamma_3}{3 \gamma_1 \gamma_2 \gamma_4^2}+\frac{\varphi_3^2 \gamma_3}{3 \gamma_1 \gamma_2 \gamma_4^2}+\frac{\varphi_4^2 \gamma_3}{3 \gamma_1 \gamma_2 \gamma_4^2}+\frac{\varphi_8^2 \gamma_3}{6 \gamma_4^4}-\frac{2 \varphi_3^2}{3 \gamma_1 \gamma_4^2}\\
 & \qquad -\frac{2 \varphi_4^2}{3 \gamma_2 \gamma_4^2}-\frac{2 \varphi_5^2}{3 \gamma_1 \gamma_4^2}-\frac{2 \varphi_6^2}{3 \gamma_1 \gamma_4^2}-\frac{2 \varphi_6^2}{3 \gamma_2 \gamma_4^2}-\frac{2 \varphi_7^2}{3 \gamma_2 \gamma_4^2}-\frac{\gamma_2}{2 \gamma_1}-\frac{\gamma_1}{2 \gamma_2}+1
\end{align*}

First of all, let us remark that the equations have a symmetry of
order 3 which fixes $\gamma_4, \varphi_1, \varphi_8$ and transforms
the remaining parameters as
\begin{equation}
  \label{eq:order3symm}
  (\gamma_1,\gamma_2,\gamma_3,\varphi_2,\varphi_3,\varphi_4,\varphi_5,\varphi_6,\varphi_7) \mapsto
  (\gamma_2,\gamma_3,\gamma_1,-\varphi_7,\varphi_4,\varphi_6,\varphi_2,\varphi_3,-\varphi_5)~.
\end{equation}
The Einstein equations allows us to solve for $\varphi_1$ and
$\varphi_8$ and one sees quickly that $\gamma_4 < 2$.  The Maxwell
equations are then linear equations on the remaining
$\varphi_{2,3,4,5,6,7}$.  The generic solution sets them all to zero,
but then this sets $\varphi_8=0$ as well.  There are two solutions,
both of which have $\varphi_1 = \pm 3$ and $\gamma_1 = \gamma_2 =
\gamma_3$.  In one solution we have $\gamma_1 = \gamma_2 = \gamma_3 =
\gamma_4 = 1$, which corresponds to the original Freund--Rubin
background, whereas in the other case we have $\gamma_1 = \gamma_2 =
\gamma_3 = \frac{25}9$ and $\gamma_4 = \frac59$, which is the squashed
7-sphere solution of \cite{Awada:1982pk}.  Both of these backgrounds
have $\gamma_1 = \gamma_2 = \gamma_3$, but there are others in this
class.  In fact, one finds a 2-sphere's worth of Englert solutions as
well as the squashed Englert solution of
\cite{Englert:1983qe,Duff:1983nu,Bais:1983wc}.  In our conventions,
the squashed Englert solution looks like
\begin{equation}\label{eq:SquashedEnglertSoln=5so3+}  
  \begin{aligned}[m]
    \lambda^{-2}g &= -(\theta^0)^2 + (\theta^1)^2 + (\theta^2)^2 + (\theta^3)^2 + \tfrac3{10} \sum_{i=4}^6 (\theta^i)^2 + \tfrac32 \sum_{i=7}^\ten (\theta^i)^2\\
    \lambda^{-3}F &= 2 \sqrt{\tfrac65} \theta^{0123} \pm 3 \left(\tfrac3{10}\right)^{3/2} \left( \theta^{4578}-\theta^{459\ten} + \theta^{4679} + \theta^{468\ten} - \theta^{567\ten} + \theta^{5689} + 5 \theta^{789\ten} \right)~,
  \end{aligned}
\end{equation}
where we recognise the second term in $F$ as the $G_2$-invariant
coassociative calibration built out of one of the Killing spinors of
the Freund--Rubin background.

It remains to look at the cases where the $\gamma_1,\gamma_2,\gamma_3$
are not all equal.  If all of $\varphi_{3,4,6}$ are different from
zero, then the Maxwell equations show that
$\gamma_1=\gamma_2=\gamma_3$, hence we must have at least one of
$\varphi_{2,3,6}$ equal to zero.  Due to the order-3 symmetry
\eqref{eq:order3symm} we can take $\varphi_3=0$ without loss of
generality, but then the first of the Einstein equations say that
$\varphi_4\varphi_6=0$ and whence we must have at least two of
$\varphi_{2,3,6}$ equal to zero.  This gives two cases to consider.
In the first case, $\varphi_3=\varphi_4=0$ and $\varphi_6\neq 0$,
whereas in the second case $\varphi_3=\varphi_4=\varphi_6=0$.

Let us consider the first case, with $\varphi_3=\varphi_4=0$ but
$\varphi_6\neq 0$.  This last condition has two immediate
consequences: the second Maxwell equation allows us to solve for
$\varphi_1$:
\begin{equation}
  \varphi_1 = -\frac{\gamma_1 + \gamma_2}{\sqrt{\gamma_1\gamma_2\gamma_3}}~;
\end{equation}
whereas the second Einstein equation forces $\varphi_7=\varphi_5$.
The fourth Maxwell equation allows us to solve for $\varphi_2$ in
terms of $\varphi_8$:
\begin{equation}
  \varphi_2 = - \frac{\gamma_1\gamma_2}{\gamma_4^2} \varphi_8~,
\end{equation}
and the fifth Einstein equation allows us to solve for $\varphi_8$ up
to a sign:
\begin{equation}
  \varphi_8 = \pm \sqrt{12 \gamma_4^3 - 6 \gamma_4^4 - 2 (\gamma_1 + \gamma_2 + \gamma_3) \gamma_4^2}~,
\end{equation}
whence
\begin{equation}
  \varphi_2 = \mp \gamma_1\gamma_2 \sqrt{12 \gamma_4^{-1} - 6 - 2 (\gamma_1 + \gamma_2 + \gamma_3) \gamma_4^{-2}}~.
\end{equation}
We now have to distinguish two cases, according to whether or not
$\gamma_1 = \gamma_2$.  If $\gamma_1 \neq \gamma_2$, we find that
there are no admissible solutions to the equations, whereas if
$\gamma_1 = \gamma_2$, we get precisely one admissible background:
\begin{equation}\label{eq:PopeWarnerSoln=5so3+}  
  \begin{aligned}[m]
     \lambda^{-2}g &= -(\theta^0)^2 + (\theta^1)^2 + (\theta^2)^2 + (\theta^3)^2 + \tfrac23 \left( (\theta^4)^2 + (\theta^5)^2 + 2 (\theta^6)^2 \right) + \tfrac23 \sum_{i=7}^\ten (\theta^i)^2\\
     \lambda^{-3}F &= -\sqrt3 \theta^{0123} + 2 \left(\tfrac23\right)^{3/2} \cos \alpha \left(\theta^{4679} + \theta^{468\ten} + \theta^{567\ten} - \theta^{5689}\right)\\
     & \qquad + 2 \left(\tfrac23\right)^{3/2} \sin \alpha \left(\theta^{467\ten} - \theta^{4689} -  \theta^{5679} - \theta^{568\ten}\right)~,
  \end{aligned}
\end{equation}
which is seen to be the Pope--Warner solution.

Finally, it remains to analyse the case where $\varphi_3 = \varphi_4 =
\varphi_6=0$, but yet $\gamma_{1,2,3}$ are not all the same.  It is
not difficult to solve for $\varphi^2_i$ using the Einstein equations,
since these equations are linear in these variables.  One finds
\begin{equation}\label{eq:EinstEqn=5so3+}  
  \begin{aligned}[m]
    \varphi_1^2 &= \frac{2(\gamma_1 + \gamma_2 + \gamma_3 + 15 \gamma_4^2 - 12 \gamma_4)}{\gamma_4^2} + \frac{2(\gamma_1^2 + \gamma_2^2 + \gamma_3^2) - (\gamma_1+\gamma_2+\gamma_3)^2}{\gamma_1 \gamma_2 \gamma_3}\\
    \varphi_2^2 &= 6 \gamma_1 \gamma_2 \gamma_4 (1-\gamma_4) +\gamma_1 \gamma_4^2+\gamma_2 \gamma_4^2-\gamma_3 \gamma_4^2 -\gamma_1 \gamma_2 \gamma_3\\
    \varphi_5^2 &= 6 \gamma_1 \gamma_3 \gamma_4 (1-\gamma_4) +\gamma_1 \gamma_4^2-\gamma_2 \gamma_4^2+\gamma_3 \gamma_4^2-\gamma_1 \gamma_2 \gamma_3\\
    \varphi_7^2 &= 6 \gamma_2 \gamma_3 \gamma_4 (1-\gamma_4) - \gamma_1 \gamma_4^2+\gamma_2 \gamma_4^2+\gamma_3 \gamma_4^2-\gamma_1 \gamma_2 \gamma_3\\
    \varphi_8^2 &= 2 \left(3 \gamma_4 (2- \gamma_4) - (\gamma_1 + \gamma_2 + \gamma_3)\right) \gamma_4^2
  \end{aligned}
\end{equation}
It follows easily from the last of these equations that $\gamma_4 < 2$
and that $\gamma_1 + \gamma_2 + \gamma_3 \leq 3$, with the bound
attained for $\gamma_4=1$ and $\varphi_8=0$.  Since now all
$\gamma_{1,2,3}$ are the same, we may exploit the order-3 symmetry
\eqref{eq:order3symm} of the equations in order to make a choice that
$\gamma_1 \neq \gamma_3$ and $\gamma_2 \neq \gamma_3$.  This leaves
open the possibility that $\gamma_1$ and $\gamma_2$ may be equal.

Alas, we have been unable to solve this system at the present
time.\footnote{Alexander~S.~Haupt has managed to complete this
  analysis and we include his results in
  Appendix~\ref{sec:addendum_fso3_+-isotropy}.  There are no new
  backgrounds.}  We know that there is a Freund--Rubin solution,
discovered in \cite{CastellaniRomansWarner}, and known to possess
$N=1$ supersymmetry and therefore also an associated Englert-like
background with the four-form $F$ constructed out of the Killing
spinor.  We expect that there are others, since preliminary numerical
investigations suggest that there should be a positive-dimensional
moduli space of solutions of these equations.

In summary, we find a number of novel $n=5$ AdS backgrounds whose
geometry is studied in more detail in the following section.

\section{The geometry of some $n=5$ backgrounds}
\label{sec:geometry-n=5}

In this section we discuss the geometry of the non-numerical $n=5$
backgrounds we discovered in Sections~\ref{sec:possible-new-n=5} and
\ref{sec:n=5}.

\subsection{$n=5$ anti de Sitter backgrounds}
\label{sec:ads-backgrounds}

In Section~\ref{sec:n=5}, we exhibited a number of new (to us)
homogeneous $\AdS_4$ backgrounds with isometry Lie algebra
$\fso(3,2)\oplus \fso(5)$.  In all cases the geometry is $\AdS_4
\times P^7$, where $P$ is a riemannian manifold homogeneous under the
action of $\SO(5)$.  The backgrounds in Section
\ref{sec:fso3-isotropy}, whose geometry is given by equations
\eqref{eq:N=5.1}, \eqref{eq:N=5.2} and \eqref{eq:N=5.3g}, all have
metric of the general form given by equation \eqref{eq:gn=5so3} with
$\gamma_{0,1,2,3}>0$.  Using the homothety invariance of the
supergravity field equations, we can set $\gamma_0=1$ without much
loss of generality.  This fixes the scale of the $\AdS_4$ factor and
the remaining metric freedom resides in the riemannian factor $P^7$
with a metric depending on $\gamma_{1,2,3}$.  This metric is Einstein
when $\gamma_1 = \gamma_2 = \frac23 \gamma_3$.  In that case, the
Einstein condition is $R_{ab} = \lambda g_{ab}$ with $\lambda =
\frac{9}{4\gamma_2}$.  For those values of the parameters the
supergravity field equations are not satisfied, which is to be
expected, since the four-form \eqref{eq:N=5.3F} is not of
Freund--Rubin type.

Let us now discuss the isometries of this family of geometries.  Using
the method described in Appendix~\ref{sec:isom-homog-space} it is
possible to show that the isometry Lie algebra of the general metric
\eqref{eq:gn=5so3} with $\gamma_{0,1,2,3}>0$ is generically indeed
$\fso(3,2) \oplus \fso(5)$, but if $\gamma_1 = \gamma_2$ it is
enhanced to $\fso(3,2) \oplus \fso(5)\oplus \fso(2)$.  This is the
case for the background described by equations \eqref{eq:N=5.3g} and
\eqref{eq:N=5.3F} (or \eqref{eq:N=5.3cali}).  The extra Killing vector
$\chi$, spans the centre of the isometry Lie algebra and, like any
Killing vector, is uniquely defined by its value and that of its
derivative at the origin.  In the notation of
Section~\ref{sec:fso3-isotropy},
\begin{equation}
\label{eq:centralKilling}
  \chi\bigr|_o = \frac{2\gamma_1}{\gamma_3} L_{45}~, \qquad \nabla_{L_{i4}} \chi \bigr|_o = L_{i5} \qquad\text{and}\qquad \nabla_{L_{i5}} \chi \bigr|_o = - L_{i4}~,
\end{equation}
where $i=1,2,3$ and all other derivatives vanish at the origin.  From
now on we will take $\gamma_1 = \gamma_2$.

Let us characterise Killing vector fields by pairs $K=(\xi,\Phi)$
where $\xi \in \fm$ is the value of the Killing vector field at the
origin and $\Phi = -\nabla\xi \in \fso(\fm)$ is (minus) its derivative
at the origin.  Let $(\be_1,\dots,\be_7)=(L_{14}, L_{24}, L_{34},
L_{15}, L_{25}, L_{35}, L_{45})$ be an ordered basis for $\fm$.  A
basis for the isometry Lie algebra of the metric on $P$ with the
choice $\gamma_1=\gamma_2$ is given by $K_a=(\xi_a,\Phi_a)$ for
$a=1,\dots,11$, where
\begin{equation}
  \xi_1=\xi_2=\xi_3=\xi_4 = 0 \qquad\text{and}\qquad (\xi_5,\dots,\xi_{11}) = (\be_1,\dots,\be_7)~,
\end{equation}
and
\begin{gather*}
    \Phi_1 = \begin{pmatrix}
  0 & 1 & 0 & 0 & 0 & 0 & 0 \\
  -1 & 0 & 0 & 0 & 0 & 0 & 0 \\
  0 & 0 & 0 & 0 & 0 & 0 & 0 \\
  0 & 0 & 0 & 0 & 1 & 0 & 0 \\
  0 & 0 & 0 & -1 & 0 & 0 & 0 \\
  0 & 0 & 0 & 0 & 0 & 0 & 0 \\
  0 & 0 & 0 & 0 & 0 & 0 & 0
\end{pmatrix}\qquad
\Phi_2 = \begin{pmatrix}
    0 & 0 & 1 & 0 & 0 & 0 & 0 \\
    0 & 0 & 0 & 0 & 0 & 0 & 0 \\
    -1 & 0 & 0 & 0 & 0 & 0 & 0 \\
    0 & 0 & 0 & 0 & 0 & 1 & 0 \\
    0 & 0 & 0 & 0 & 0 & 0 & 0 \\
    0 & 0 & 0 & -1 & 0 & 0 & 0 \\
    0 & 0 & 0 & 0 & 0 & 0 & 0
\end{pmatrix}\\
\Phi_3 = \begin{pmatrix}
  0 & 0 & 0 & 1 & 0 & 0 & 0 \\
  0 & 0 & 0 & 0 & 1 & 0 & 0 \\
  0 & 0 & 0 & 0 & 0 & 1 & 0 \\
  -1 & 0 & 0 & 0 & 0 & 0 & 0 \\
  0 & -1 & 0 & 0 & 0 & 0 & 0 \\
  0 & 0 & -1 & 0 & 0 & 0 & 0 \\
  0 & 0 & 0 & 0 & 0 & 0 & 0
\end{pmatrix}\qquad
\Phi_4 = \begin{pmatrix}
 0 & 0 & 0 & 0 & 0 & 0 & 0 \\
 0 & 0 & 1 & 0 & 0 & 0 & 0 \\
 0 & -1 & 0 & 0 & 0 & 0 & 0 \\
 0 & 0 & 0 & 0 & 0 & 0 & 0 \\
 0 & 0 & 0 & 0 & 0 & 1 & 0 \\
 0 & 0 & 0 & 0 & -1 & 0 & 0 \\
 0 & 0 & 0 & 0 & 0 & 0 & 0
\end{pmatrix}\\
\Phi_5 = \begin{pmatrix}
 0 & 0 & 0 & 0 & 0 & 0 & 0 \\
 0 & 0 & 0 & 0 & 0 & 0 & 0 \\
 0 & 0 & 0 & 0 & 0 & 0 & 0 \\
 0 & 0 & 0 & 0 & 0 & 0 & -\frac{\gamma _3}{2 \gamma _1} \\
 0 & 0 & 0 & 0 & 0 & 0 & 0 \\
 0 & 0 & 0 & 0 & 0 & 0 & 0 \\
 0 & 0 & 0 & \frac{1}{2} & 0 & 0 & 0
\end{pmatrix}\qquad
\Phi_6 = \begin{pmatrix}
 0 & 0 & 0 & 0 & 0 & 0 & 0 \\
 0 & 0 & 0 & 0 & 0 & 0 & 0 \\
 0 & 0 & 0 & 0 & 0 & 0 & 0 \\
 0 & 0 & 0 & 0 & 0 & 0 & 0 \\
 0 & 0 & 0 & 0 & 0 & 0 & -\frac{\gamma _3}{2 \gamma _1} \\
 0 & 0 & 0 & 0 & 0 & 0 & 0 \\
 0 & 0 & 0 & 0 & \frac{1}{2} & 0 & 0
\end{pmatrix}\\
\Phi_7 = \begin{pmatrix}
 0 & 0 & 0 & 0 & 0 & 0 & 0 \\
 0 & 0 & 0 & 0 & 0 & 0 & 0 \\
 0 & 0 & 0 & 0 & 0 & 0 & 0 \\
 0 & 0 & 0 & 0 & 0 & 0 & 0 \\
 0 & 0 & 0 & 0 & 0 & 0 & 0 \\
 0 & 0 & 0 & 0 & 0 & 0 & -\frac{\gamma _3}{2 \gamma _1} \\
 0 & 0 & 0 & 0 & 0 & \frac{1}{2} & 0
\end{pmatrix}\qquad
\Phi_8 = \begin{pmatrix}
 0 & 0 & 0 & 0 & 0 & 0 & \frac{\gamma _3}{2 \gamma _1} \\
 0 & 0 & 0 & 0 & 0 & 0 & 0 \\
 0 & 0 & 0 & 0 & 0 & 0 & 0 \\
 0 & 0 & 0 & 0 & 0 & 0 & 0 \\
 0 & 0 & 0 & 0 & 0 & 0 & 0 \\
 0 & 0 & 0 & 0 & 0 & 0 & 0 \\
 -\frac{1}{2} & 0 & 0 & 0 & 0 & 0 & 0
\end{pmatrix}\\
\Phi_9 = \begin{pmatrix}
 0 & 0 & 0 & 0 & 0 & 0 & 0 \\
 0 & 0 & 0 & 0 & 0 & 0 & \frac{\gamma _3}{2 \gamma _1} \\
 0 & 0 & 0 & 0 & 0 & 0 & 0 \\
 0 & 0 & 0 & 0 & 0 & 0 & 0 \\
 0 & 0 & 0 & 0 & 0 & 0 & 0 \\
 0 & 0 & 0 & 0 & 0 & 0 & 0 \\
 0 & -\frac{1}{2} & 0 & 0 & 0 & 0 & 0
\end{pmatrix}\qquad
\Phi_{10} = \begin{pmatrix}
 0 & 0 & 0 & 0 & 0 & 0 & 0 \\
 0 & 0 & 0 & 0 & 0 & 0 & 0 \\
 0 & 0 & 0 & 0 & 0 & 0 & \frac{\gamma _3}{2 \gamma _1} \\
 0 & 0 & 0 & 0 & 0 & 0 & 0 \\
 0 & 0 & 0 & 0 & 0 & 0 & 0 \\
 0 & 0 & 0 & 0 & 0 & 0 & 0 \\
 0 & 0 & -\frac{1}{2} & 0 & 0 & 0 & 0
\end{pmatrix}
\end{gather*}
and $\Phi_{11}=0$.

Let $S = K_3 + \alpha K_{11}$, where $\alpha \in \RR$ is a parameter.
Then $S$ is given by the data $(\xi,\Phi)$ with
\begin{equation}
  \label{eq:Reeb7.4.2}
  \xi = \alpha \be_7 \qquad\text{and}\qquad \Phi = \Phi_3~.
\end{equation}

We claim that for some choice of metric, $S$ is the Reeb vector field
of a Sasakian structure on $P$.  We recall that an odd-dimensional
riemannian manifold $(P,g)$ is Sasakian if and only if its metric cone
$(\Pti = \RR^+\times P, \gti = dr^2 + r^2 g)$, for $r>0$ the
coordinate on $\RR^+$, is Kähler.  Let us see what this means
intrinsically.

On $\Pti$ there is a hermitian structure $J$ which is compatible with
$\gti$ and is parallel relative to the riemannian connection $\nablat$
of $\gti$.  Let $\omega$ denote the associated Kähler form:
$\omega(X,Y) = \gti(JX,Y)$.  Let $E = r\frac{\partial}{\partial r}$
denote the Euler vector field.  Its derivative $\nablat E$ is the
identity endomorphism: $\nablat_X E = X$ for all vector fields $X$ on
the cone, whence it generates homotheties of $(\Pti,\gti)$.  Define a
1-form $\etati$ on the cone by $\etati = \iota_E \omega$.  In other
words, $\etati(X) = \omega(E,X) = \gti(JE,X)$.  A quick calculation
shows that $\etati$ scales with weight $2$ under the homotheties
generated by $E$.  Since $\iota_E \etati=0$, we see that
$r^{-2}\etati$ is \emph{basic}; that is, there is a 1-form $\eta$ on
$P$ such that $r^{-2}\etati = \pi^*\eta$ with $\pi: \Pti \to P$ the
natural projection $(r,p) \mapsto p$.  It is another relatively
straightforward calculation to show that $d\etati = 2 \omega$, whence
\begin{equation}
  \omega = \half d\etati = \half d (r^2 \pi^* \eta) = r dr \wedge \pi^*\eta + \half r^2 \pi^* d\eta~.
\end{equation}
Let $P$ be $(2n+1)$-dimensional.  Since $\omega$ is a Kähler form,
\begin{equation}
  \omega^{n+1} = n 2^{-n} r^{2n+1} \pi^*(\eta \wedge (d\eta)^n)
\end{equation}
is nowhere vanishing, which implies that $\eta \wedge (d\eta)^n$ is
nowhere vanishing and thus $\eta$ defines a contact structure on $P$.

Let us define the vector field $\Sti = JE$.  Since $\gti(\Sti,E) =
\gti(JE,E) = 0$, $\Sti$ restricts to a vector field on $P$, which we
may and will think of as the $r=1$ slice of $\Pti$.  It is an easy
calculation to show that $\Sti$ is a Killing vector field with norm
$r^2$, whence it restricts to a unit-norm Killing vector $S$ on $P$
and hence $\eta(S) = 1$.  The covariant derivative $\phi = -\nabla S$
of $S$ on $M$ defines a complex structure on the distribution $\eD$
orthogonal to the one spanned by $S$ itself.  Indeed, $\nablat \Sti =
\nablat JE = J \circ \nablat E = J$ and hence $\phi$ is defined by
declaring it to coincide with $J$ on the orthogonal complement to the
distribution spanned by $E$ and $\Sti$ and to annihilate $S$: $\phi
(S) =0$.  In other words, one can show that
\begin{equation}
  \phi^2 = -\id + S \otimes \eta~.
\end{equation}

The compatibility between the riemannian and complex structures on the
cone becomes the compatibility between the riemannian and contact
structures on $P$; namely,
\begin{equation}
  \label{eq:metriccontact}
  g(\phi X,\phi Y)= g(X,Y) - \eta(X)\eta(Y)~.
\end{equation}
This simply says that $\phi$ is an isometry on the distribution $\eD$.

Finally, the integrability condition of the hermitian structure (i.e.,
the vanishing of the Nijenhuis tensor) becomes a differential
condition on $\phi$:
\begin{equation}
 (\nabla_X \phi)(Y) = g(X,Y) S - g(S,Y) X~.
\end{equation}
We may rewrite the left-hand side as
\begin{equation}
  (\nabla_X \phi)(Y) = \nabla_X (\phi(Y)) - \phi(\nabla_X Y) = -\nabla_X \nabla_Y S + \nabla_{\nabla_X Y} S = -R(S,X)Y~,
\end{equation}
whence the integrability condition becomes an algebraic equation
involving the curvature tensor:
\begin{equation}
  \label{eq:integrability}
  R(S,X) Y + g(X,Y) S - g(S,Y) X = 0~.
\end{equation}

A Sasaki structure is actually too strict for our needs.  This derives
from the fact that Sasaki structures are not scale invariant, whereas
the supergravity field equations are.  We actually need a somewhat
more relaxed notion of Sasaki structure which says that
$(P,g,\eta,S,\phi)$ is only \emph{homothetic} to a Sasaki structure.
In other words, we (tentatively) say that $(P,g,\eta,S,\phi)$ defines
an \textbf{$r$-Sasaki structure}, if for some $r > 0$, $(P,r^2
g,\eta,S,\phi)$ is a Sasaki structure.  The name comes from the fact
that the $r$-slice of a Kähler cone has such a structure.  Since both
the riemannian connection $\nabla$ and the riemann curvature $R$ are
invariant under homotheties, $(P,g,\eta,S,\phi)$ is an $r$-Sasaki
structure if all equations of a Sasaki structure are obeyed, except
for the following changes:
\begin{enumerate}
\item the normalisation of $S$ is now $g(S,S) = r^{-2}$,
\item the metric compatibility condition \eqref{eq:metriccontact} is now
  \begin{equation}
    \label{eq:r-metriccontact}
    g(\phi X,\phi Y)= g(X,Y) - r^{-2} \eta(X)\eta(Y)~,
  \end{equation}
\item and the integrability condition \eqref{eq:r-integrability} is replaced by
  \begin{equation}
    \label{eq:r-integrability}
    R(S,X) Y + r^2 g(X,Y) S - r^2 g(S,Y) X = 0~.
  \end{equation}
\end{enumerate}

These conditions are easy to check for the homogeneous backgrounds of
interest.  We notice that the new integrability condition
\eqref{eq:r-integrability} is tensorial and only depends on the value
of the Reeb vector field $S$ at the origin and does so linearly.
Therefore in the expression \eqref{eq:Reeb7.4.2} for $S$, the
parameter $\alpha$ is not fixed by \eqref{eq:r-integrability}.
Indeed, it is not difficult to verify that for all $\alpha \neq 0$,
the integrability condition \eqref{eq:r-integrability} is satisfied
provided that $r^2 = \frac{\gamma_3}{4\gamma_1^2}$.  The parameter
$\alpha$ is fixed by normalising $S$ to $g(S,S)=r^{-2}$, which means
$\alpha = \frac{2\gamma_1}{\gamma_3}$.  Comparing with
\eqref{eq:centralKilling} we see that $S=\chi$ the generator of the
centre of the isometry Lie algebra!

In the case of the background with flux given by equation
\eqref{eq:N=5.3F}, the Reeb vector field $S$ does not preserve it,
whence the symmetry Lie algebra of the background is precisely
$\fso(5) \oplus \fso(3,2)$.  Let us briefly explain the calculation of
the Lie derivative $L_S F$ of $F$ along $S$.  Using the Cartan formula
and the fact that $F$ is closed, $L_S F = d i_S F $, so we need to
compute the exterior derivative of the 3-form $i_S F$.  We saw in
Section \ref{sec:comp-homog-spac}, particularly equation \eqref{eq:d},
that the exterior derivative of an invariant form is easy to compute
algebraically.  Now let $X$ be a Killing vector and let us see whether
$i_S F$ is invariant.  We calculate the Lie derivative of $i_S F$
along $X$ to obtain
\begin{equation}
  L_X i_S F = i_S L_X F + i_{[X,S]} F~.
\end{equation}
The first term in the RHS vanishes because $F$ is invariant and the
second term vanishes precisely because the Reeb vector field $S$ is
central, whence $[X,S]=0$ for all Killing vectors $X$.  This means
that we can use equation \eqref{eq:d} to compute $d i_S F$ and we find
that it is not zero.

\subsection{Other $n=5$ backgrounds}
\label{sec:other-backgrounds}

We now look in some detail at two of the backgrounds found in Section
\ref{sec:possible-new-n=5}.

\subsubsection{A supersymmetric Freund--Rubin background}
\label{sec:freund-rubin-backgr}

In this background, the geometry is $S^4 \times X^7$ with $X$ a
lorentzian Sasaki--Einstein manifold.  Since the background is
Freund--Rubin, this means \cite{FigLeiSim} that it is supersymmetric.
Explicitly, in terms of the basis given earlier in this section, and
reintroducing the scale $\lambda\neq 0$, we have
\begin{equation}
  \label{eq:susyFRLES}
  \begin{aligned}[m]
    \lambda^{-2} g &= - (\theta^0)^2 + \tfrac49 \left((\theta^1)^2 + \cdots + (\theta^4)^2\right) + \tfrac23 \left((\theta^5)^2 + \cdots + (\theta^\ten)^2\right)\\
    \lambda^{-3}F &= \tfrac89 \theta^{1234}~.
  \end{aligned}
\end{equation}

Let us now show that $X$ admits an invariant lorentzian
Sasaki--Einstein structure.  In fact, let us consider more generally
any homogeneous lorentzian geometry of the type $(\fso(3,2),\fso(3))$
with $\gamma_2 = \gamma_3 =: \gamma$.  This describes two of the three
backgrounds we have found.  Lorentzian Sasaki structures
$(M,g,S,\eta,\phi)$ are described by a lorentzian odd-dimensional
manifold $(M,g)$ with a timelike Killing vector $S$ normalised to
$g(S,S)=-1$, a contact structure $\eta$ with $\eta(S)=1$ and
endomorphism $\phi = - \nabla S$ with $\phi^2 = -\id + S \otimes \eta$
and subject to slight modifications of the metric compatibility and
integrability conditions \eqref{eq:metriccontact} and
\eqref{eq:r-integrability}, respectively; namely, the lorentzian
metric compatibility condition is now
\begin{equation}
  \label{eq:lor-metric-contact}
  g(\phi X,\phi Y)= g(X,Y) + \eta(X)\eta(Y)~, 
\end{equation}
whereas the integrability condition reads
\begin{equation}
  \label{eq:lor-integrability}
    R(S,X) Y - g(X,Y) S + g(S,Y) X = 0~.
\end{equation}
Similarly, we can consider \textbf{lorentzian $r$-Sasaki structures}
$(M,g,S,\eta,\phi)$, defined in such a way that $(M,r^2
g,S,\eta,\phi)$ is lorentzian Sasaki.  This means that now $S$ is
normalised to $g(S,S)= - r^{-2}$ and that the metric compatibility and
integrability conditions change to
\begin{equation}
  \label{eq:lor-r-metric-contact}
  g(\phi X,\phi Y)= g(X,Y) + r^{-2}\eta(X)\eta(Y)~,
\end{equation}
whereas the integrability condition reads
\begin{equation}
  \label{eq:lor-r-integrability}
    R(S,X) Y - r^2 g(X,Y) S + r^2 g(S,Y) X = 0~.
\end{equation}

The situation here is very similar to that of Section
\ref{sec:ads-backgrounds}.  When $\gamma_2 = \gamma_3$, there is an
enhancement of symmetry to $\fso(3,2) \oplus \fso(2)$, where the
central Killing vector $S$ is determined by the pair $(\xi,\Phi) \in
\fm \oplus \fso(\fm)$, with $\xi = 2 \gamma J_{45}$ and $\Phi$ given
relative to the ordered basis $(J_{i4},J_{i5},J_{45})$ by the matrix
\begin{equation}
  \begin{pmatrix}
    0 & 0 & 0 & 1 & 0 & 0 & 0 \\
    0 & 0 & 0 & 0 & 1 & 0 & 0 \\
    0 & 0 & 0 & 0 & 0 & 1 & 0 \\
    -1 & 0 & 0 & 0 & 0 & 0 & 0 \\
    0 & -1 & 0 & 0 & 0 & 0 & 0 \\
    0 & 0 & -1 & 0 & 0 & 0 & 0 \\
    0 & 0 & 0 & 0 & 0 & 0 & 0
  \end{pmatrix}~.
\end{equation}

We claim, as suggested by the choice of notation, that $S$ is the Reeb
vector field of a lorentzian $r$-Sasaki structure.  It is normalised
to $g(S,S) = - 4 \gamma^2$, whence we expect that $r=\frac1{2\gamma}$
and indeed a straightforward calculation shows that
\begin{equation}
  4 \gamma^2 R(S,X) Y - g(X,Y) S + g(S,Y) X = 0~.
\end{equation}

The nonzero components of the Levi-Civita connection for this metric
are readily calculated using the formulas \eqref{eq:levicivita} and
\eqref{eq:utensorn=5}:
\begin{equation}
  \begin{aligned}[m]
    \nabla_{J_{45}} J_{i4} &= \left(\tfrac1{2\gamma} - 1\right) J_{i5}\\
    \nabla_{J_{i4}} J_{45} &= \tfrac1{2\gamma} J_{i5}\\
    \nabla_{J_{i4}} J_{j5} &= \tfrac12 \delta_{ij} J_{45}\\
  \end{aligned}
  \qquad\qquad
  \begin{aligned}[m]
    \nabla_{J_{45}} J_{i5} &= \left(1-\tfrac1{2\gamma}\right) J_{i4}\\
    \nabla_{J_{i5}} J_{45} &= -\tfrac1{2\gamma} J_{i4}\\
    \nabla_{J_{i5}} J_{j4} &= -\tfrac12 \delta_{ij} J_{45}~.
  \end{aligned}
\end{equation}
Using the formulae \eqref{eq:riccin=5} in the case where $\gamma_2 =
\gamma_3 =: \gamma$, we see that $\Lambda = 1 - 4\gamma^2$ and
\begin{equation}
  \Ric(J_{45},J_{45}) = \frac3{2\gamma^2} \qquad\text{and}\qquad
  \Ric(J_{i4},J_{j4}) = \Ric(J_{i5},J_{j5}) = \delta_{ij} \left( \tfrac1{2\gamma} - 3\right)~,
\end{equation}
whence the metric is Einstein if and only if $-\frac3{2\gamma^2} =
\frac1\gamma (\frac1{2\gamma} -3)$, or equivalently if and only if
$\gamma = \frac23$.

In summary, we have shown that the background with $\gamma_1 =
\frac49$ and $\gamma_2 = \gamma_3 = \frac23$ is a Freund--Rubin
background of the form $S^4 \times X^7$, with $X$ (homothetic to) a
lorentzian Sasaki--Einstein manifold.  Lischewski
\cite{2014arXiv1409.2664L} has shown that this background admits
$N{=}2$ supersymmetry, a fact that can also be deduced in this case
from unpublished results \cite{FigLeiSim}.  Indeed, Killing spinors
take the form
\begin{equation}
  \varepsilon^{Ia} = \zeta^{Ia} \otimes \psi~,
\end{equation}
where $I=1,\dots,4$ and $a=1,2$, $\zeta^{Ia}$ are geometric Killing
spinors on $S^4$ and $\psi$ is a geometric Killing spinor on the
lorentzian Sasaki--Einstein manifold.  The spinor $\varepsilon$ is
subject to a symplectic Majorana condition
\begin{equation}
  \left(\varepsilon^{Ia}\right)^* = \Omega_{IJ}\epsilon_{ab} \varepsilon^{Jb}~,
\end{equation}
with $\Omega_{IJ}$ the $\Sp(2)$-invariant symplectic structure on the
space of Killing spinors of $S^4$, which by Bär's cone construction is
isomorphic as an $\Sp(2)$-module to the space of parallel spinors on
$\RR^5$, which is just the spinor irreducible representation of
$\Spin(5)\cong \Sp(2)$, hence a quaternionic representation.  With a
suitable normalisation of the Killing spinors, the Killing
superalgebra is given by
\begin{equation}
  [\varepsilon^{Ia},\varepsilon^{Jb}] = \Omega^{IJ} \epsilon^{ab}\chi~,
\end{equation}
where $\chi$ is the Reeb vector field of the lorentzian Sasaki
structure.  Therefore the $\SO(3,2) \times \SO(5)$ symmetry is
accidental and only the central $\SO(2)$ symmetry is induced by the
supersymmetry.

This solution looks like it could be obtained via a Wick rotation from
a background of the type $\AdS_4 \times X^7$, with $X$ a
Sasaki--Einstein 7-manifold.\footnote{We are grateful to James
  Lucietti for this suggestion.}  If that is the case, the $\AdS_4$
background must be a Freund--Rubin background and hence must be one of
the backgrounds classified in \cite{CastellaniRomansWarner} and
discussed here in Section~\ref{sec:fso3-isotropy}.  Let us try to
identify it.  The solution we found is described, as a homogeneous
space, by the data $(\fg,\fh) =
(\fso(5)\oplus\fso(3,2),\fso(4)\oplus\fso(3))$.  We may think of
$\fso(5)\oplus \fso(3,2)$ as a real Lie subalgebra of the complex Lie
algebra $\fgl(10,\CC)$ of $10\times 10$ complex matrices.  We are
after a homogeneous $\AdS_4$ Freund--Rubin background which, as a
homogeneous space is described algebraically by the data $(\fg',\fh')
= (\fso(3,2)\oplus \fso(5),\fso(3,1)\oplus \fso(3))$.  Therefore the
Wick rotation we are after is an element $\varpi \in \GL(10,\CC)$ such
that multiplying on both the left and the right by $\varpi$ maps
$(\fg,\fh)$ to $(\fg',\fh')$.  A little experimentation leads us to
the following diagonal matrix
\begin{equation}
  \varpi =
  \begin{pmatrix}
    I_3 & & & \\
    & i I_2 & & \\
    & & I_3 & \\
    & & & i I_2
  \end{pmatrix}\in \GL(10,\CC)~,
\end{equation}
where $I_n$ is the $n\times n$ identity matrix.  The element $\varpi$
thus defines a ``quadruple'' Wick rotation.  The map $X \mapsto \varpi
X \varpi$ sends the Lie subalgebra $\fg = \fso(5)\oplus\fso(3,2)
\subset \fgl(10,\CC)$ which preserves the inner product
\begin{equation}
  \eta =
  \begin{pmatrix}
    I_3 & & & \\
    & I_2 & & \\
    & & I_3 & & \\
    & & & - I_2
  \end{pmatrix}
\end{equation}
to the isomorphic Lie subalgebra $\fg' = \fso(3,2) \oplus \fso(5)$
which preserves the inner product
\begin{equation}
  \eta' =
  \begin{pmatrix}
    I_3 & & & \\
    & -I_2 & & \\
    & & I_3 & & \\
    & & & I_2
  \end{pmatrix}~.
\end{equation}
At the same it sends the subalgebra $\fh = \fso(4) \oplus \fso(3)$ of
$\fg$ to the subalgebra $\fh' = \fso(3,1)\oplus \fso(3)$ of $\fg'$.
It is not hard to show that the homogeneous space described by
$(\fg',\fh')$ admits a Freund--Rubin background, given relative to the
Wick-rotated basis by a similar expression to that of
\eqref{eq:susyFRLES}, namely
\begin{equation}
    \begin{aligned}[m]
    \lambda^{-2} g &= (\theta^0)^2 + \tfrac49 \left((\theta^1)^2 + (\theta^2)^2 +(\theta^3)^2 - (\theta^4)^2\right) + \tfrac23 \left((\theta^5)^2 + \cdots + (\theta^\ten)^2\right)\\
    \lambda^{-3}F &= \tfrac89 \theta^{1234}~.
  \end{aligned}
\end{equation}
The underlying geometry is $\AdS_4 \times X^7$, where $X^7$ is a
Sasaki--Einstein manifold with isometry Lie algebra $\fso(5)\oplus
\fso(2)$ because of the enhancement due to $\gamma_2 = \gamma_3$.  In
fact, it is possible to identify $X^7$ with the real Stiefel manifold
$V_2(\RR^5)$ of orthonormal $2$-frames in $\RR^5$ with the Einstein
metric, equivalently the unit tangent bundle to $S^4$.   This
background is discussed in \cite[Appendix~C]{CastellaniRomansWarner}
and is also discussed in \cite{SorokinKK}, which contains references
to earlier papers.  It is shown in \cite{CastellaniRomansWarner} that
the solution has $N=2$ supersymmetry, just as the Wick-rotated
background found here.

\subsubsection{A circle of backgrounds}
\label{sec:circle-backgrounds}

This background depends on a parameter $\alpha$ which shares the same
underlying geometry:
\begin{equation}
  \begin{aligned}
    \lambda^{-2}g &= - (\theta^0)^2 + \tfrac49 \left((\theta^1)^2 + \cdots + (\theta^\ten)^2\right)\\
    \lambda^{-3}F &= -\tfrac13 \theta^{1234} +\tfrac{1}{\sqrt{3}} \cos\alpha \left(\theta^{0567} - \theta^{059\ten} + \theta^{068\ten} - \theta^{0789}\right)\\
    &\qquad - \tfrac{1}{\sqrt{3}} \sin\alpha \left(\theta^{056\ten} - \theta^{0579} + \theta^{0678} - \theta^{089\ten}\right)\\
    &= -\tfrac13 \theta^{1234} + \tfrac{1}{\sqrt{3}} \theta^0 \wedge \Re \left(e^{i\alpha} (\theta^5 + i \theta^8) \wedge (\theta^6 + i \theta^9)  \wedge (\theta^7 + i \theta^\ten) \right)~.
  \end{aligned}
\end{equation}
The geometry is again $S^4 \times X^7$ with $X$ (homothetic to) a
lorentzian Sasaki manifold, but now it is not Einstein.  This does not
necessarily imply that it is not supersymmetric, since the background
is not of Freund--Rubin type: the 4-form has components in both
factors.  However Lischewski \cite{2014arXiv1409.2664L} has shown that
this background is not supersymmetric by an explicit calculation of
the holonomy algebra of the connection $D$.  This background does not
seem to be Wick-related to an $AdS_4$ background.  In the second
expression for $F$ we recognise a transverse special lagrangian
calibration, which suggests that this background is obtained from a
supersymmetric Freund--Rubin background via the Englert procedure.  It
seems likely that the supersymmetric Freund--Rubin background in
question is the background described in the previous section.
Finally, we remark that although there is an enhancement of the
isometry algebra by an additional central Killing vector, this is not
a symmetry of $F$, whence the symmetry Lie algebra of the background
remains isomorphic to $\fso(3,2)\oplus\fso(5)$.

\section{Summary of results and open problems}
\label{sec:summary-results-open}

We have presented the results of a systematic search for
eleven-dimensional supergravity backgrounds homogeneous under a Lie
group with Lie algebra $\fg_n := \fso(n) \oplus \fso(3,2)$ for
$n=5,6,7$.  The aim of this search is to explore the existence of new
candidate backgrounds with $N>4$ supersymmetry dual to
three-dimensional superconformal field theories.  It is known that
such backgrounds are homogeneous and the structure of the
superconformal algebra is such that the bosonic subalgebra is
isomorphic to $\fg_n$.  Since backgrounds with $N=8$ supersymmetry
have been classified, we have restricted ourselves to $n=5,6,7$;
although we have many partial results for $n=4$ which have not made it
to this paper.

Such homogeneous backgrounds come in two families: those with
underlying geometry $\AdS_4 \times P^7$ and the rest.  We find no new
backgrounds for $n=6,7$, but we find a number of possibly novel
backgrounds with $n=5$ of both types.  Curiously all backgrounds we
find are metrically products.

We find three new backgrounds with underlying geometry $\AdS_4 \times
P^7$, where $P$ is a homogeneous riemannian manifold $\SO(5)/\SO(3)$,
where $\SO(3)$ is the subgroup of $\SO(5)$ which leaves pointwise
invariant a plane in $\RR^5$.  One of the backgrounds can only be
approximated numerically.  Of the other two backgrounds, one of them
is discussed in detail in Section~\ref{sec:ads-backgrounds}, where it
is shown that it is (homothetic to) a Sasaki manifold, whence the
geometry has an enhanced isometry Lie algebra $\fso(3,2) \oplus
\fso(5) \oplus \fso(2)$, where $\fso(2)$ is generated by the Reeb
vector field of the Sasaki structure. The flux is not preserved by the
Reeb vector field, whence the background's symmetry is not enhanced.
In fact, this background is not new, since its existence was mentioned
in \cite[Appendix~C]{CastellaniRomansWarner}.  It can be identified
with the result of applying the Englert procedure to a background
$\AdS_4 \times V_2(\RR^5)$.  As a result it breaks all the
supersymmetry.  We have not analysed the supersymmetry of the other
two backgrounds.

We also have found backgrounds which do not have an $\AdS_4$ factor,
yet still have an $\fso(3,2)$ summand in the symmetry algebra.  We
have found three such backgrounds, all with underlying geometry $S^4
\times Q^7$ with $S^4$ the round 4-sphere and $Q$ a homogeneous
lorentzian manifold $\SO(3,2)/\SO(3)$ but with different kinds of
fluxes.  One of the backgrounds, discussed in detail in
Section~\ref{sec:freund-rubin-backgr} is of Freund--Rubin type since
the flux is proportional to the volume form on $S^4$.  In this case
$Q$ is (homothetic to) a lorentzian Sasaki--Einstein manifold and this
means that the background is supersymmetric, albeit only with $N=2$.
As shown in Section~\ref{sec:freund-rubin-backgr}, this background is
Wick-related to a Freund--Rubin background $\AdS_4 \times V_2(\RR^5)$
already known from classical times \cite{CastellaniRomansWarner}.
There is an enhancement of symmetry and the full isometry algebra is
$\fso(5) \oplus \fso(3,2) \oplus \fso(2)$, with the $\fso(2)$
generated by the Reeb vector field of the Sasaki structure.  We also
find find a circle's worth of backgrounds, described in Section
\ref{sec:circle-backgrounds}, which seems to be the result of applying
the Englert procedure to the Freund--Rubin background just mentioned.
If this is indeed the case, then the background preserves no
supersymmetry.  Here the geometry is lorentzian Sasaki and although
there is an enhancement of the isometry algebra to $\fso(5) \oplus
\fso(3,2) \oplus \fso(2)$, the Reeb vector field does not preserve the
rather complicated flux.  Finally, we also find a background which we
can only approximate numerically.  For this background there is no
enhancement of the symmetry and in particular $Q$ does not have a
homogeneous Sasaki structure.  This numerical background is given by
equation \eqref{eq:numericalS4bg} and we have yet to investigate
whether it preserves any supersymmetry.

\section*{Acknowledgments}

The work of JMF was supported in part by grants ST/G000514/1 ``String
Theory Scotland'' and ST/J000329/1 ``Particle Theory at the Tait
Institute'' from the UK Science and Technology Facilities Council and
that of MU in part by the Nuffield Undergraduate Research Bursary
URB/39258. We are grateful to both institutions for their support.
Twice during the long gestation of this project, JMF visited the
University of Tōkyō and he would like to thank Teruhiko Kawano for
organising the visits and for providing such a pleasant research
environment. We are also grateful to Christoph Nölle for some initial
collaboration on this project, as well as to Paul de Medeiros, Joan
Simón and especially James Lucietti for lending us their ears. We are
particularly grateful to Robert Bryant for his answer to a question on
MathOverflow, recapped in the Appendix, and on which some of the
calculations in Section \ref{sec:geometry-n=5} are based; and to
Andree Lischewski for developing an algorithm to check the
supersymmetry of a homogeneous background and for performing some
checks on the results of Section~8. Finally, we are grateful to Dima
Sorokin for making us aware of his early work on supergravity
compactifications.

The additional work carried out in
Appendix~\ref{sec:addendum_fso3_+-isotropy} was the direct result of
conversations with Alexander~Haupt during the MITP Topical Workshop
``Geometry, Gravity and Supersymmetry'' (GGSUSY2017), held at the
Mainz Institute for Theoretical Physics, to whom we are grateful for
their support and hospitality and for providing such a stimulating
research atmosphere.

\appendix

\section{Isometries of a homogeneous space}
\label{sec:isom-homog-space}

Let $(M,g)$ be a homogeneous riemannian manifold admitting a
transitive action of $G$ with generic stabiliser $H$, so that $M$ is
diffeomorphic to $G/H$.  This means that $G$ is a subgroup of the
isometry group of $(M,g)$, but it could very well be the case that $G$
is a proper subgroup.  The Lie algebra of the group of isometries can
be determined by solving the Killing vector equation on $(M,g)$.  A
Killing vector is determined uniquely by its value at a point and that
of its covariant derivative relative to the Levi-Civita connection.
Indeed, as shown in \cite{KostantHol,Geroch} and discussed in
\cite{FMPHom}, Killing vectors are in bijective correspondence with
parallel sections of $TM \oplus \fso(TM)$, with $\fso(TM) \cong
\Lambda^2T^*M$ the bundle of skewsymmetric endomorphisms of the
tangent bundle, relative to the connection defining the so-called
Killing transport:
\begin{equation}
  D_X
  \begin{pmatrix}
    \xi \\ A
  \end{pmatrix}
  =
  \begin{pmatrix}
    \nabla_X\xi + A(X)\\ \nabla_XA - R(X,\xi)
  \end{pmatrix}~.
\end{equation}
In a homogeneous space, since both $\nabla$ and $R$ are invariant
under isometries, it is possible to turn this into a linear system of
equations with constant coefficients, which can be succinctly
described by lifting the problem to the group $G$.  The following
treatment owes a lot to Robert Bryant \cite{MO75887} via
MathOverflow.

As usual we think of $(M,g)$ as described algebraically by a reductive
split $\fg = \fh \oplus \fm$ together with an $H$-invariant inner
product $\left<-,-\right>$ on $\fm$.  Let us choose bases $(X_a)$ for
$\fh$ and $(Y_i)$ for $\fm$.  The structure constants of $\fg$
relative to these bases are given by
\begin{equation}
  \label{eq:strconsts}
  [X_a, X_b] = f_{ab}{}^c X_c \qquad [X_a, Y_i] = f_{ai}{}^j Y_j \qquad [Y_i, Y_j] = f_{ij}{}^a X_a + f_{ij}{}^k Y_k~.
\end{equation}
Let $(\psi^a)$ and $(\theta^i)$ denote the left-invariant
Maurer--Cartan one-forms on $G$ dual to the chosen bases for $\fg$.
The structure equations are
\begin{equation}
  \label{eq:structeqns}
  d\psi^c = - \half f_{ab}{}^c \psi^a \wedge \psi^b - \half f_{ij}{}^c \theta^i \wedge \theta^j \qquad
  d\theta^k = -\half f_{ij}{}^k \theta^i \wedge \theta^j - f_{ai}{}^k \psi^a \wedge \theta^i~.
\end{equation}
The $H$-invariant inner product on $\fm$ has components
$\eta_{ij}:=\left<Y_i,Y_j\right>$ relative to the chosen basis and
$\eta_{ij}\theta^i\theta^j$ is the pullback to $G$ of the invariant
metric on $M=G/H$.  The invariance of the inner product means that
$f_{aij} = - f_{aji}$, where here and in the sequel, we lower indices
using $\eta_{ij}$, so that $f_{aij} = f_{ai}{}^k\eta_{kj}$.

Let $\omega^i{}_j$ denote the connection 1-form defined by
\begin{equation}
  \label{eq:connection1form}
  d\theta^i = - \omega^i{}_j \wedge \theta^j \qquad\text{and}\qquad \omega_{ij} = - \omega_{ji}~.
\end{equation}
The structure equations allow us to solve for $\omega^i{}_j$:
\begin{equation}
  \label{eq:omega}
  \omega^i{}_j = f_{aj}{}^i \psi^a - \half f_{jk}{}^i \theta^k + \half \eta^{il} \left(f_{ljk} + f_{lkj}\right)\theta^k~.
\end{equation}

The curvature 2-form
\begin{equation}
  \label{eq:curv2form}
  \Omega^i{}_j =  d\omega^i{}_j + \omega^i{}_k \wedge \omega^k{}_j
\end{equation}
can be shown to be horizontal, whence it can be expressed only in
terms of the $(\theta^i)$:
\begin{equation}
  \Omega^i{}_j = \half R^i{}_{jkl} \theta^k \wedge \theta^l~.
\end{equation}

A Killing vector field on $M=G/H$ lifts to a vector field on $G$ which
is defined by the data $(\xi,A)$, where
\begin{equation}
  \label{eq:KillingEqn}
  d\xi^i + \omega^i{}_j \xi^j = - A^i{}_j \theta^j \qquad\text{and}\qquad A_{ij} = - A_{ji}~.
\end{equation}
Differentiating this equation and using the structure equations
\eqref{eq:connection1form}, \eqref{eq:curv2form} and Killing's
equation \eqref{eq:KillingEqn} itself, we arrive at
\begin{equation}
  \label{eq:Killingtransport}
  (dA^i{}_j + [\omega,A]^i{}_j)\wedge \theta^j = - \Omega^i{}_j \xi^j~.
\end{equation}

The following simple result is very useful.

\begin{lem}
  Let $M^i{}_j$ be a matrix of 1-forms such that
  \begin{equation}
    M^i{}_j \wedge \theta^j = 0 \qquad\text{and}\qquad M_{ij}=-M_{ji}~.
  \end{equation}
  Then $M=0$.
\end{lem}

\begin{proof}
  Write $M^i{}_j = M^i{}_{jk}\theta^k$.  The condition
  $M^i{}_j \wedge \theta^j = 0$ becomes
  $M^i{}_{jk}\theta^k \wedge \theta^j = 0$, which is equivalent to
  $M^i{}_{jk} = M^i{}_{kj}$.  Lowering the index with $\eta$, this is
  equivalent to $M_{ijk} = M_{ikj}$; but since $M_{ijk} = - M_{jik}$,
  we see that
  \begin{equation}
    M_{ijk} = - M_{jik} = - M_{jki} = M_{kji} = M_{kij} = - M_{ikj} =
    - M_{ijk} \quad\text{hence}\quad M_{ijk} = 0~.
  \end{equation}
\end{proof}

Equation \eqref{eq:Killingtransport} says that
$\Omega^i{}_j \xi^j + \rho^i{}_j \wedge \theta^j =0$, where
$\rho^i{}_j = dA^i{}_j + [\omega,A]^i{}_j$.  Using the lemma, we can
give an alternate expression for $\rho$ in terms of the curvature.
Indeed,
\begin{align*}
  \Omega^i{}_j \xi^j + \rho^i{}_l \wedge \theta^l &= \half R^i{}_{jkl}\xi^j \theta^k \wedge \theta^l + \rho^i{}_l \wedge \theta^l\\
  &= \left(\half R^i{}_{jkl}\xi^j \theta^k + \rho^i{}_l\right) \wedge \theta^l~,
\end{align*}
but the algebraic Bianchi identity says that $R^i{}_{jkl} = -
R^i{}_{ljk} - R^i{}_{klj}$, whence
\begin{align*}
  \half R^i{}_{jkl} \xi^j\theta^k \wedge \theta^l &= - \half (R^i{}_{ljk} + R^i{}_{klj}) \xi^j\theta^k \wedge \theta^l \\
  &= R^i{}_{lkj} \xi^j\theta^k \wedge \theta^l~.
\end{align*}
This implies that
\begin{equation}
 \left( R^i{}_{lkj} \xi^j\theta^k + \rho^i{}_l \right)\wedge \theta^l = 0~,
\end{equation}
but since $R_{ijkl} = - R_{jikl}$ and $\rho_{ij} = -\rho_{ji}$, the
lemma says that
\begin{equation}
  \label{eq:rho}
  \rho^i{}_l = R^i{}_{ljk} \xi^j\theta^k~.
\end{equation}
It is convenient to think of $\rho$ as a bilinear in $\xi$ and
$\theta$ and define
\begin{equation}
  \label{eq:rhobi}
  \rho(\xi,\theta)^i{}_l = R^i{}_{ljk} \xi^j\theta^k~.
\end{equation}
Therefore equation \eqref{eq:Killingtransport} together with the
lemma, imply that
\begin{equation}
  \label{eq:Killingtrans2}
  dA + [\omega,A] = \rho(\xi,\theta)~,
\end{equation}
where we have dropped the indices and interpreted this equation as a
matricial equation.

Differentiating equation \eqref{eq:Killingtrans2} and using the
various structure equations to eliminate the derivatives  $d\xi$,
$dA$, $d\omega$ and $d\theta$ from the expression, we arrive at
\begin{equation}
  \label{eq:Killingtrans3}
  [\Omega,A] - [\omega,\rho(\xi,\theta)] + \rho(\xi,\omega\wedge\theta) + \rho(\omega\xi,\theta) + \rho(A\theta,\theta)= 0~,
\end{equation}
where $[\Omega,A] = \Omega A - A \Omega$ and
$[\omega,\rho(\xi,\theta)] = \omega \wedge \rho(\xi,\theta) +
\rho(\xi,\theta) \wedge \omega$.  The beauty of equation
\eqref{eq:Killingtrans3} is that it is linear on $\xi,A$ with constant
coefficients!

Differentiating further and using the various structure equations
again to eliminate derivatives, yields new linear equations with
constant coefficients.  Eventually this process will terminate, in the
sense that no new equations are obtained.  When this happens, we are
left with a set of linear equations in $(\xi,A)$ whose solution space
is the Lie algebra of isometries of $(M,g)$ with Lie bracket given by
\begin{equation}
  \label{eq:Liebracket}
  [(\xi_1,A_1), (\xi_2,A_2)] = \left(A_1 \xi_2 - A_2 \xi_1, [A_1,A_2] - R(\xi_1,\xi_2)\right)~,
\end{equation}
as proved, for example, in \cite[§3]{FMPHom}.

\section{Addendum to Section~\ref{sec:fso3_+-isotropy}}
\label{sec:addendum_fso3_+-isotropy}
\begin{center}
  by Alexander S. Haupt
\end{center}


In section~\ref{sec:fso3_+-isotropy}, the subalgebra $\fso(3)_+$ of
$\fso(5)$ is chosen as the isotropy algebra.  This leads to a priori
four metric parameters $\gamma_{1,2,3,4}>0$ and eight real parameters
$\varphi_{1,2,3,4,5,6,7,8}$ describing the space of invariant closed
4-forms.  The section concludes with a partial analysis of the
remaining case where $\varphi_3 = \varphi_4 = \varphi_6=0$ (see
page~\pageref{eq:EinstEqn=5so3+}). The purpose of this appendix is to
analyse the remaining case in full generality, thereby closing a small
gap in the systematic search for eleven-dimensional supergravity
backgrounds homogeneous under a Lie group with Lie algebra
$\fso(5) \oplus \fso(3,2)$.

For the case where $\varphi_3 = \varphi_4 = \varphi_6=0$, it remains
to solve the Maxwell and the Einstein equations expressed
in~\eqref{eq:MWEqn=5so3+} and~\eqref{eq:EinstEqn=5so3+},
respectively. The first three equations in~\eqref{eq:MWEqn=5so3+} are
identically satisfied owing to $\varphi_3 = \varphi_4 =
\varphi_6=0$. The remaining nine equations in~\eqref{eq:MWEqn=5so3+}
and~\eqref{eq:EinstEqn=5so3+} can be turned into a system of
polynomial equations upon introducing a new auxiliary variable $A$
subject to the constraint $A^2 = \gamma_1 \gamma_2 \gamma_3$. In
total, this yields a system of ten polynomial equations in the ten
unknowns $(\gamma_{1,2,3,4}, \varphi_{1,2,5,7,8}, A)$ of degree at
most six,
\begin{align}
  0 &= A^2-\gamma _1 \gamma _2 \gamma _3 , \nonumber \\
  0 &= \gamma _3 \gamma _4^2 \varphi _2+\gamma _2 \gamma _4^2 \varphi _5+\gamma _1 \gamma _4^2 \varphi _7 \pm A \gamma _4^2 \varphi _1 \varphi _7 + A^2 \varphi _8 , \nonumber \\
  0 &= \gamma _3 \gamma _4^2 \varphi _2-\gamma _2 \gamma _4^2 \varphi _5-\gamma _1 \gamma _4^2 \varphi _7 \mp A \gamma _4^2 \varphi _1 \varphi _5 + A^2 \varphi _8 , \nonumber \\
  0 &= \gamma _3 \gamma _4^2 \varphi _2-\gamma _2 \gamma _4^2 \varphi _5+\gamma _1 \gamma _4^2 \varphi _7 \pm A \gamma _4^2 \varphi _1 \varphi _2 - A^2 \varphi _8 , \nonumber \\
  0 &= 2 \gamma _3 \varphi _2+2 \gamma _2 \varphi _5-2 \gamma _1 \varphi _7 \mp A \varphi _1 \varphi _8 , \label{eq:Polysn=5so3+} \\
  0 &= 2 A^2 (\gamma_1 + \gamma_2 + \gamma_3 + 15 \gamma_4^2 - 12 \gamma_4) + 2(\gamma_1^2 + \gamma_2^2 + \gamma_3^2) \gamma_4^2 - (\gamma_1+\gamma_2+\gamma_3)^2 \gamma_4^2 - A^2 \gamma_4^2 \varphi_1^2 , \nonumber \\
  0 &= 6 \gamma_1 \gamma_2 \gamma_4 (1-\gamma_4) +\gamma_1 \gamma_4^2+\gamma_2 \gamma_4^2-\gamma_3 \gamma_4^2 - A^2 - \varphi_2^2 , \nonumber \\
  0 &= 6 \gamma_1 \gamma_3 \gamma_4 (1-\gamma_4) +\gamma_1 \gamma_4^2-\gamma_2 \gamma_4^2+\gamma_3 \gamma_4^2 - A^2 - \varphi_5^2 , \nonumber \\
  0 &= 6 \gamma_2 \gamma_3 \gamma_4 (1-\gamma_4) - \gamma_1 \gamma_4^2+\gamma_2 \gamma_4^2+\gamma_3 \gamma_4^2 - A^2 - \varphi_7^2 , \nonumber \\
  0 &= 2 \left(3 \gamma_4 (2- \gamma_4) - (\gamma_1 + \gamma_2 + \gamma_3)\right) \gamma_4^2 - \varphi_8^2 , \nonumber
\end{align}
where the sign ambiguity of the terms linear in $A$ in the second to
fifth equations is due to the choice of the positive or negative
branch of the square root $A = \pm \sqrt{\gamma_1 \gamma_2
  \gamma_3}$. The system of polynomial
equations~\eqref{eq:Polysn=5so3+} is well-suited for a computer-based
Gr\"obner basis computation, with the polynomials on the right-hand
sides forming the input set. Using the computer algebra system Magma,
we compute a Gr\"obner basis with lexicographic monomial ordering,
where in addition the order of variables is taken to be
$(\varphi_1, \varphi_8, \varphi_2, \varphi_7, \varphi_5, \gamma_3,
\gamma_1, \gamma_2, A, \gamma_4)$. The computation, performed on a
compute-server with 24 Intel Xeon E5-2643 3.40 GHz processors and 512
GB of RAM, took 55 minutes to run and consumed about 545 MB of RAM.

The resulting Gr\"obner basis contains 561 polynomials with on average
132 terms per polynomial. The numerical coefficients range up to order
$10^{32}$. These numbers are independent of the sign ambiguity
stemming from the choice $A = \pm \sqrt{\gamma_1 \gamma_2
  \gamma_3}$. Regardless of the apparent complexity of the resulting
Gr\"obner basis, it is straightforward to find the vanishing locus of
these polynomials by virtue of the so-called elimination property
satisfied, under certain conditions, by Gr\"obner bases obtained with
respect to lexicographic monomial orderings. Restricting to the
physically relevant solutions where $\gamma_{1,2,3,4}>0$ and
$\varphi_{1,2,5,7,8} \in \mathbb{R}$ (this implies $\gamma_4 < 2$ and
$\gamma_1 + \gamma_2 + \gamma_3 \leq 3$, as noted
below~\eqref{eq:EinstEqn=5so3+}, as well as
$A\in\mathbb{R}\setminus\{0\}$), we find a priori seven types of
discrete solutions, as summarized in the following table.
\begin{center}\def\arraystretch{1.4}
\begin{tabular}[h]{|c|c|c|c|c|c|c|c|c|c|c|c|}\hline
 counter & $\gamma _1$ & $\gamma _2$ & $\gamma _3$ & $\gamma _4$ & $\varphi _1^2$ & $\varphi _2^2$ & $\varphi _5^2$ & $\varphi _7^2$ & $\varphi _8^2$ & $\varphi_i$-signs & $\#$ \\ \hline\hline
 $(1)$ & $1$ & $1$ & $1$ & $1$ & $9$ & $0$ & $0$ & $0$ & $0$ & $(\pm,0,0,0,0)$ & $2$ \\ \hline
 $(2)$ & $\frac{9}{25}$ & $\frac{9}{25}$ & $\frac{9}{25}$ & $\frac{9}{5}$ & $9$ & $0$ & $0$ & $0$ & $0$ & $(\pm,0,0,0,0)$ & $2$ \\ \hline
 $(3)$ & $\frac{3}{10}$ & $\frac{3}{10}$ & $\frac{3}{10}$ & $\frac{3}{2}$ & $\frac{24}{5}$ & $\frac{243}{1000}$ & $\frac{243}{1000}$ & $\frac{243}{1000}$ & $\frac{243}{40}$ & $(+,\pm,\pm,\mp,\pm)$ & $2$ \\ \hline
 $(4)$ & $\frac{5}{6}$ & $\frac{5}{6}$ & $\frac{5}{6}$ & $\frac{5}{6}$ & $\frac{24}{5}$ & $\frac{125}{216}$ & $\frac{125}{216}$ & $\frac{125}{216}$ & $\frac{125}{216}$ & \begin{tabular}[c]{@{}c@{}}$(-,\pm,\pm,\pm,\mp)$,\\ $(-,\pm,\mp,\pm,\pm)$,\\ $(-,\mp,\pm,\pm,\pm)$\end{tabular} & $6$ \\ \hline
 $(5)$ & $\frac{2}{3}$ & $\frac{2}{3}$ & $\frac{4}{3}$ & $\frac{2}{3}$ & $3$ & $0$ & $\frac{32}{27}$ & $\frac{32}{27}$ & $0$ & $(-,0,\pm,\pm,0)$ & $2$ \\ \hline
 $(6)$ & $\frac{4}{3}$ & $\frac{2}{3}$ & $\frac{2}{3}$ & $\frac{2}{3}$ & $3$ & $\frac{32}{27}$ & $\frac{32}{27}$ & $0$ & $0$ & $(-,\pm,\mp,0,0)$ & $2$ \\ \hline
 $(7)$ & $\frac{2}{3}$ & $\frac{4}{3}$ & $\frac{2}{3}$ & $\frac{2}{3}$ & $3$ & $\frac{32}{27}$ & $0$ & $\frac{32}{27}$ & $0$ & $(-,\pm,0,\pm,0)$ & $2$ \\ \hline
\end{tabular}
\end{center}
Here, the first column represents a counter in order to distinguish
the solutions and the last column contains the multiplicities of the
solutions originating from the sign choices stated in the penultimate
column. The values of the variables $\varphi_{1,2,5,7,8}$ are given by
the square roots of columns six to ten with the possible combinations
of positive and negative branches of the square roots listed in the
penultimate column. The sign choices for the values of the variables
$\varphi_{1,2,5,7,8}$ are understood to be correlated.

Note that solutions~$(6)$ and~$(7)$ reduce to solution~$(5)$ due to
the order-3 symmetry~\eqref{eq:order3symm}. This effectively reduces
the above list to the five cases $(1)$--$(5)$.  It is also worth
noting that the solutions are insensitive to the sign ambiguity
in~\eqref{eq:Polysn=5so3+} originating from the choice
$A = \pm \sqrt{\gamma_1 \gamma_2 \gamma_3}$. In addition, we remark
that intermediate steps of the calculation indicate the presence of
continuous families of solutions. It turns out, however, that the
one-parameter families of solutions are located in unphysical branches
of solution space where at least one of the four variables
$\gamma_{1,2,3,4}$ vanishes.

Comparing with the solutions already obtained in
section~\ref{sec:homogeneous-anti-de}, we conclude that solution~$(1)$
corresponds to the original Freund-Rubin background, whereas
solution~$(2)$ can be identified, upon redefining
$\gamma_{1,2,3,4} \to \gamma_{1,2,3,4}^{-1}$, with the squashed
7-sphere solution (cf. remarks below
\eqref{eq:order3symm}). Solution~$(3)$ is equal to the squashed
Englert solution~\eqref{eq:SquashedEnglertSoln=5so3+}. In addition,
solution~$(4)$ corresponds to the Englert
solution~\eqref{eq:EnglertSoln=7}, upon rescaling
$\theta^i \to 3 \theta^i$, $i=4,\ldots,9, \ten$. Finally,
solution~$(5)$ is seen to be the Pope-Warner
solution~\eqref{eq:PopeWarnerSoln=5so3+}.

In summary, the system of polynomial equations~\eqref{eq:Polysn=5so3+}
corresponding to the remaining case
$\varphi_3 = \varphi_4 = \varphi_6=0$ of
section~\ref{sec:fso3_+-isotropy} yields five distinct physically
relevant solutions, all of which can be mapped to solutions already
obtained in section~\ref{sec:homogeneous-anti-de}. This concludes our
analysis of the remaining case where
$\varphi_3 = \varphi_4 = \varphi_6=0$.


\providecommand{\href}[2]{#2}\begingroup\raggedright\endgroup

\end{document}